\documentclass[prl,twocolumn,amsmath,amssymb,10pt,aps,longbibliography,superscriptaddress,citeautoscript,bibnotes,floatfix,nobibnotes]{revtex4-2}

\usepackage{graphicx} 
\usepackage{dcolumn}
\usepackage{bm}
\usepackage{graphicx,color,xcolor}
\usepackage{physics}
\usepackage{comment}
\usepackage[normalem]{ulem}
\graphicspath{ {./Figures/} }
\usepackage{textcomp}
\usepackage[utf8]{inputenc}
\usepackage[T1]{fontenc}
\usepackage{amsmath}
\usepackage{amssymb}
\usepackage{amsthm}
\usepackage{graphicx}
\usepackage{subfigure}
\usepackage{tabularx}
\usepackage{booktabs}
\usepackage{multirow}
\usepackage[colorlinks=true,linkcolor=blue,citecolor=blue]{hyperref} 

\usepackage{xcolor}

\newcommand{\stkout}[1]{\ifmmode\text{\sout{\ensuremath{#1}}}\else\sout{#1}\fi}

\usepackage{soul}

\begin{document}

\title{Pseudochaotic Many-Body Dynamics as a Pseudorandom State Generator}

\author{Wonjun Lee}
\email{wonjun1998@postech.ac.kr}
\affiliation{Department of Physics, Pohang University of Science and Technology, Pohang 37673, South Korea}
\affiliation{Center for Artificial Low Dimensional Electronic Systems, Institute for Basic Science, Pohang, 37673, South Korea}

\author{Hyukjoon Kwon}
\email{hjkwon@kias.re.kr}
\affiliation{School of Computational Sciences, Korea Institute for Advanced Study, Seoul 02455, South Korea}

\author{Gil Young Cho}
\email{gilyoungcho@kaist.ac.kr}
\affiliation{Department of Physics, Korea Advanced Institute of Science and Technology, Daejeon 34141, South Korea}
\affiliation{Center for Artificial Low Dimensional Electronic Systems, Institute for Basic Science, Pohang, 37673, South Korea}
\affiliation{Asia-Pacific Center for Theoretical Physics, Pohang, Gyeongbuk 37673, South Korea}

\date{\today}

\begin{abstract}
Quantum chaos is central to understanding quantum dynamics and is crucial for generating random quantum states, a key resource for quantum information tasks. In this work, we introduce a new class of quantum many-body dynamics, termed pseudochaotic dynamics. Although distinct from chaotic dynamics, out-of-time-ordered correlators, the key indicators of quantum chaos, fail to distinguish them. Moreover, pseudochaotic dynamics generates pseudorandom states that are computationally indistinguishable from Haar-random states. We construct pseudochaotic dynamics by embedding a smaller $k$-qubit subsystem into a larger $n$-qubit system. We demonstrate that a subsystem of size $k=\omega(\log n)$ is sufficient to induce pseudochaotic behavior in the entire $n$-qubit system. Furthermore, we construct a quantum circuit exhibiting pseudochaotic dynamics and demonstrate that it generates pseudorandom states within $\mathrm{polylog}(n)$ depth. In summary, our results constitute the discovery of new quantum dynamics that are computationally indistinguishable from genuine quantum chaos, which provides efficient routes to generate useful pseudorandom states.

\end{abstract}

\maketitle

\noindent{\large{\bf Introduction}}\\
Quantum many-body dynamics represents a forefront of our modern understanding of quantum mechanics with profound implications across fields such as quantum information science~\cite{Saffman_2010,Wendin_2017,Pezze_2018,Monroe2021}, thermodynamics~\cite{Eisert_2015,Nandkishore_2015,D_Alessio_2016}, condensed matter physics~\cite{Dziarmaga_2010,Heyl_2018,Abanin_2019}, and high-energy physics~\cite{Lunin_2002,Maldacena2016SYK,Maldacena2016chaos}. However, due to the hardness of simulating the dynamics, their properties are still not fully understood. One area of particular interest is quantum chaotic dynamics, with prototypical examples including the Sachdev–Ye–Kitaev (SYK) model~\cite{Sachdev1993,Kitaev2015} and random quantum circuits~\cite{Matthew2023}. A defining characteristic of chaos, both in classical and quantum systems, is the butterfly effect~\cite{Lieb1972,Shenker_2014,Roberts2017}, which asserts that local information in an initial state quickly scrambles across an exponentially large space. Since tracking this scrambled information requires exponential resources, simulating such quantum dynamics using classical computers is generally intractable. This challenge has spurred the use of quantum devices to study quantum chaos~\cite{Lewis2019,Monroe2021,Joshi2022} with potential applications in quantum supremacy tasks~\cite{preskill2012, Boixo_2018,arute2019, aaronson2020}. Discovering new classes of quantum many-body dynamics could similarly yield unexpected insights across these fields.    
 
A deep connection between quantum chaos and randomness offers a promising route for generating ensembles of quantum states~\cite{Brand_o_2016, Nakata2017, Ho_2022, Cotler2023} close to uniformly random (i.e., Haar-random) quantum ensembles. As a quantum cryptographic primitive, Haar-random quantum ensembles have crucial applications in quantum information science including quantum cryptography~\cite{ananth2022, kretschmer2023}, quantum estimation theory~\cite{Knill_2008, Huang_2020, huang2022}, and quantum complexity theory~\cite{bouland2019}. However, preparing a genuine Haar-random ensemble of quantum states demands exponentially deep circuits~\cite{knill1995}, which current technology struggles to achieve. The recent formulation of pseudorandom quantum states~\cite{Ji2018} has shed light on this problem by considering an ensemble of quantum states that even quantum computers cannot distinguish from Haar-random states within limited computation time, i.e., computationally indistinguishable, but are preparable with lower circuit depth. The pseudorandomness in quantum states and associated computational indistinguishability of various quantum resources like entanglement~\cite{aaronson2023quantum}, magic~\cite{gu2023little}, and coherence~\cite{haug2023pseudorandom}, have found a wide range of applications in quantum information processing~\cite{kretschmer2021, morimae2022, bostanci2023, elben2023, movassagh2023}.

Motivated by these previous developments,  we introduce a new class of quantum many-body dynamics for quantum simulations, called `pseudochaotic dynamics,’ capable of generating pseudorandom states. Although pseudochaotic dynamics are not chaotic and thus fundamentally distinct from conventional chaotic dynamics, we demonstrate that they are surprisingly indistinguishable within limited computation time from chaotic ones through the defining metric of chaos, namely out-of-time-ordered correlators (OTOCs)~\cite{hashimoto2017}, which quantify the butterfly effect.

{In this work}, we provide a systematic construction of these dynamics by embedding the unitary dynamics of a $k$-qubit subsystem into the entire $n$-qubit quantum system. Remarkably, this approach is feasible even with a very small subsystem size with $k=\omega(\log n)$ and a circuit depth of $\omega(\log n)$. This depth without any conditions almost touches the shallowest depth for generating pseudorandom states made by assuming cryptographic assumptions~\cite{aaronson2023quantum,Naor2004}. Moreover, the subsystem's dynamics need not be inherently chaotic, which contrasts strongly with our common understanding of quantum chaos. We also discuss how the properties of pseudochaos are related to coherence generated by the subsystem dynamics~\cite{haug2023pseudorandom}, and other quantum resources studied in the context of pseudorandom quantum states, such as entanglement~\cite{aaronson2023quantum}, magic~\cite{gu2023little}.

\begin{figure*}[t]
    \includegraphics[width=\linewidth]{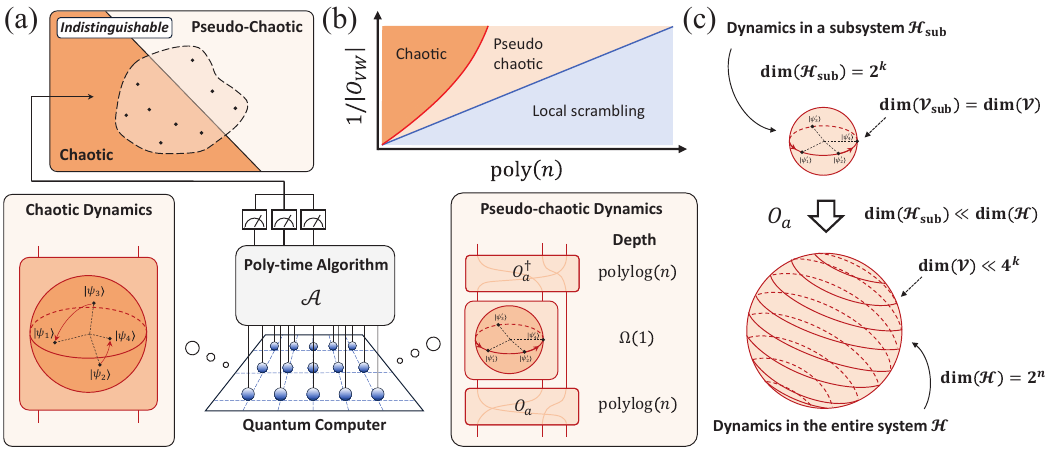}
    \centering
    \caption{Overview. (a) A quantum computer simulates dynamics either from a chaotic ensemble or a pseudochaotic ensemble. Any poly-time quantum algorithm $\mathcal{A}$ on a polynomial number of copies with poly-time classical post-processing fails to pin down whether the dynamics is chaotic or not by measuring OTOC. An example of pseudochaotic dynamics is obtained by conjugating the dynamics in a subspace by a random permutation. To make this dynamics pseudochaotic, the dimension of the subspace $2^k$ should be given by $k=\omega(\log n)$ with the number of qubits $n$, which is much smaller than the entire space dimension $2^n$. A circuit implementation of this dynamics requires $\operatorname{polylog}(n)$ depth with all-to-all connectivity. (b) We can schematically classify how the late-time $1/O_{VW}$ scales with $n$ into three different regimes. In a chaotic system {(orange color)}, this scaling is exponential in $n$. For a system with local scrambling {(blue color)}, the scaling is at most polynomial in $n$. Pseudochaotic dynamics {(peach color)} exhibits the scaling which falls between these two. (c) The $2^k$-dimensional subspace is mapped to the entire space by isometries $\{O_a\}$. $\mathcal{V}_\mathrm{sub}$ is the space spanned by an ensemble of unitary operators $\{u\}$ in the subspace, which could be non-chaotic. Through the action of ${O_a}$, $\mathcal{V}_\mathrm{sub}$ is mapped to $\mathcal{V}$ within the entire Hilbert space, preserving its dimension. Remarkably, even if the subspace dimension is negligibly smaller than the entire space dimension, its ensemble average over random isometries cannot be distinguished from chaotic dynamics by any poly-time quantum algorithms with access to polynomially many copies of evolved states.}
    \label{fig:overview}
\end{figure*}
 
\vspace{1em} 
\noindent{\large{\bf Results}}\\
\noindent {\bf Pseudochaotic Dynamics.}\\
We define pseudochaotic dynamics as a {non-chaotic} unitary time evolution with {i.} {computational} indistinguishability from chaotic time evolution via OTOCs, and {ii.} capability to generate pseudorandom states. The latter implies that the pseudochaotic dynamics loses its initial information over a certain time scale as like the former does.

\textit{i. {Computational} indistinguishability via OTOC}. The OTOC of a unitary operator $U$ with Pauli like local operators $V$ and $W$, $V^2=W^2=1$, at infinite temperature is~\cite{hashimoto2017} 
\begin{equation}\label{eq:def-OTOC}
    O_{VW}(U) = \frac{1}{2^n}\operatorname{tr}\left( V U W U^\dagger V U W U^\dagger \right).
\end{equation}
This can be estimated in experiments by reverse time evolution or interferometric measurements~\cite{Brian2016,yao2016}. $O_{VW}(U)$ can be thought of as an observable when we replace the trace operation by the inner product with the Bell pair state in the double copy space, and thus it can be estimated by measurements in that space.

In quantum chaotic systems, $O_{VW}(U)$ converges to zero as the time interval increases for any local operators $V$ and $W$. In addition, the values of these converged OTOCs decrease exponentially with the system size~\cite{Roberts2017}. Detecting this characteristic exponential decay in system size requires the uncertainties in the estimated OTOCs to also decay exponentially. However, if the number of realizable copies is limited to a polynomial in the system size, the uncertainties in estimating the OTOCs by a poly-time quantum algorithm scale at best as $\Omega(1/\operatorname{poly}(n))$ when the algorithm saturates the Heisenberg limit, making it impossible to observe the exponential decay of the OTOCs. Consequently, any dynamics with OTOCs scaling negligibly, {i.e.,} $O_{VW} \sim o(1/\operatorname{poly}(n))$, for all time $t\geq t_*$ for some constant $t_*>0$ becomes {computationally} indistinguishable from chaotic one by OTOCs. We take this indistinguishability as one criterion for $U$ to be pseudochaotic.

\textit{ii. Capability to generate pseudorandom states}. An ensemble of chaotic unitary operators can generate Haar-random states~\cite{Schiulaz_2019}. We require that pseudochaotic dynamics produces a pseudorandom state ensemble~\cite{Ji2018} in the same way. This capability of generating pseudorandom states is another manifestation of the indistinguishability of pseudochaotic dynamics from chaotic ones. For this, we note that a pseudorandom state ensemble is an ensemble that cannot be distinguished from a truly random ensemble using only a polynomial number of measurements and a poly-time quantum algorithm. This immediately implies that an ensemble of states generated through pseudochaotic dynamics as we explained above should be indistinguishable from one generated by fully chaotic dynamics within polynomial copies by any polynomial time (quantum) algorithm. Figure~\ref{fig:overview}(a) illustrates this concept of the pseudochaotic dynamics.

\vspace{1em}
\noindent {\bf Explicit Construction.}\\
We introduce a systematic construction for pseudochaotic dynamics, which we term `random subsystem-embedded dynamics (RSED)'. The RSED consists of two components: a random subset isometry $O_a$ and an embedded unitary operator $u$ in the subsystem. The random subset isometry is defined as 
\begin{equation}\label{eq:isometry}
    O_a=\sum_{b\in\{0,1\}^k}(-1)^{f(ba)}\ketbra{p(ba)}{ba},
\end{equation}
where $k\leq n$ represents a subsystem of size $k$ within the total system of size $n$, and $p$ and $f$ are random permutation and function, respectively. Our key observation is that $k=\omega(\log n)$ is necessary for RSED to be pseudochaotic. Below we always set $k=\omega(\log n)$. The term $a\in\{0,1\}^{n-k}$ serves as the seed for these random mappings. This isometry embeds a unitary operator $u$, acting on the $2^k$-dimensional subsystem Hilbert space, into the $2^n$-dimensional Hilbert space of the entire system.

The full unitary evolution of RSED is given by
\begin{equation}\label{eq:full-evolution}
    U=\sum_{a\in\{0,1\}^{n-k}}O_a u O_a^\dagger.
\end{equation}
Figure~\ref{fig:numerics}(b) illustrates how dynamics in the subsystem is embedded into the entire system by an isometry $O_a$. The effect of the conjugation by $\{O_a\}$ is equivalent to applying random permutation with random sign factors on the unitary operator in the entire space, $u\otimes I^{\otimes(n-k)}$. The time evolution operator {with an evolution time $t$} is  
\begin{equation} \label{eq:tau-evolution}
    U^t=\sum_{a\in\{0,1\}^{n-k}}O_a u^t O_a^\dagger, 
\end{equation}
because of $O_a^{\dagger} O_{a'} = \delta_{a,a'}$. In principle, arbitrary $u$ is allowed. However, for a RSED to be pseudochaotic, it is sufficient {for $u^t$ to have negligibly small elements for all time $t\geq t_{*}$ for some positive constant $t_{*}$ as we explain below.} {More details on the RSED can be found in {Supplementary Note~1}}. 

{Interestingly, the level statistics of a pseudochaotic RSED differ drastically from those of conventional chaotic systems due to exponential degeneracies in its energy spectrum. Thus, even when a chaotic \( u \) is chosen, the level statistics of the corresponding RSED deviate from the standard Wigner-Dyson distribution~\cite{Cotler2017Chaos,livan2018}. In contrast, the spectral form factor of the RSED closely follows the behavior of that of \( u \). Further details can be found in Supplementary Note~10.}

\vspace{1em}
\noindent {\bf Negligible OTOC.}\\
We first show that if elements of $u$ have negligible magnitudes in the computational basis, {i.e.,} the diagonalizing basis of $O_a$, then {an individual realization of} the RSED has negligible OTOCs and thus cannot be distinguished from chaotic unitary evolutions via OTOC. {Here, negligible, or $\operatorname{negl}(n)$ appeared below, means the magnitude of a quantity decays faster than inverse of any polynomial function of $n$.}

\textit{Theorem 1:}
OTOCs with local operators are negligible {with probabilities higher than $1-\operatorname{negl(n)}$} in the system size $n$ for {an individual realization of} the RSED with an embedded unitary operator {$u$} of the dimension $2^k$ with $k=\omega(\log n)$ {(sampled from an ensemble)} if the maximum (averaged) magnitude of elements of the embedded operator {$u$} is $O(2^{-k/2})$.

\textit{Proof:} Details are in {Theorem 5} of {Supplementary Note~3}. 

Such $u$ naturally includes general chaotic dynamics following the random matrix theory, whose time evolution operators have the matrix elements of order $O(2^{-k/2})$ {for all time $t\geq t_*$ for some constant $t_*>0$}. {In literature, such $t_*$ is called the intermediate time regime for chaotic dynamics~\cite{Garc_a_Mata_2023}.} 

Notably, non-chaotic $u$ can also have such property. An example is the product of Hadamard gates $H^{\otimes k}$ with a random sign operator $P$ in the subsystem, namely $u=H^{\otimes k}P$. The matrix elements of $u^t$ are on average order of $2^{-k/2}$ for $t\geq t_* \approx 1$ (See Supplementary Note~7). Thus, this RSED is expected to demonstrate negligible OTOCs for all $t\gtrsim 1$, according to {Theorem} 1. Indeed, we numerically confirm that individual realizations of the RSED exhibit negligibly small $O_{VW}$ as shown in Figure~\ref{fig:numerics}(a,b). By passing, we mention that without the sign randomization $P$, $u= H^{\otimes k}$ alone cannot produce pseudochaotic dynamics and OTOCs are not suppressed as the matrix elements of $u^t$ do not persistently scale with $2^{-k/2}$ (Supplementary Note~8).

We also compute OTOCs of the RSED by embedding the Pauli SYK model~\cite{Hanada2024}, which is chaotic. As expected, this RSED demonstrates vanishing OTOCs, see Figure~\ref{fig:numerics}(c,d). {Importantly, the late-time saturated values of \( O_{VW} \) for $u=H^{\otimes k}P$ scale as \(\operatorname{negl}(n)\), as clearly demonstrated in the log-log plot of \( O_{VW} \) versus system size \( n \) (Figure~\ref{fig:OTOC-sys-scale}).} Minor numerical details and additional data are available in {Supplementary Note 7}. By passing, we note that OTOCs {at finite temperatures} and {those} with non-local Pauli operators are also negligible {(Supplementary Notes~4,8)}.

{The exponential decay of OTOCs in evolution time and their associated exponents, known as Lyapunov exponents, are also well-established signatures of quantum chaotic systems~\cite{Maldacena2016chaos,Garc_a_Mata_2023}. However, for systems governed by nonlocal Hamiltonians, such as the pseudochaotic dynamics considered here, a well-defined Lyapunov exponent does not exist. Further details on this issue are provided in Supplementary Note~6.}

\begin{figure}[t]
    \includegraphics[width=\linewidth]{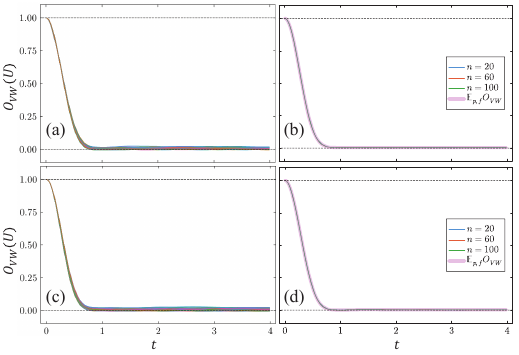}
    \centering
    \caption{
    Time dependence of OTOCs $O_{VW}$ of (a,b) the Pauli SYK model and (c,d) $H^{\otimes k}P$ with $V=Z_i$ and $W=Z_j$ with $i\neq j$. (a,c) $O_{VW}$ of independent realizations. (b,d) averaged $O_{VW}$ over random subset isometries and random realizations of the embedded $u$'s for various system sizes $n$. $\mathbb{E}_{p,f} O_{VW}$ in the caption denotes that the averaged OTOCs are computed using the closed formula in {Method} Eq.~\eqref{eq:closed-form-ZZ-OTOC}. For all cases, the subspace dimension is $2^k$ with $k=10$.}
    \label{fig:numerics}
\end{figure}

\begin{figure}[t]
    \includegraphics[width=0.8\linewidth]{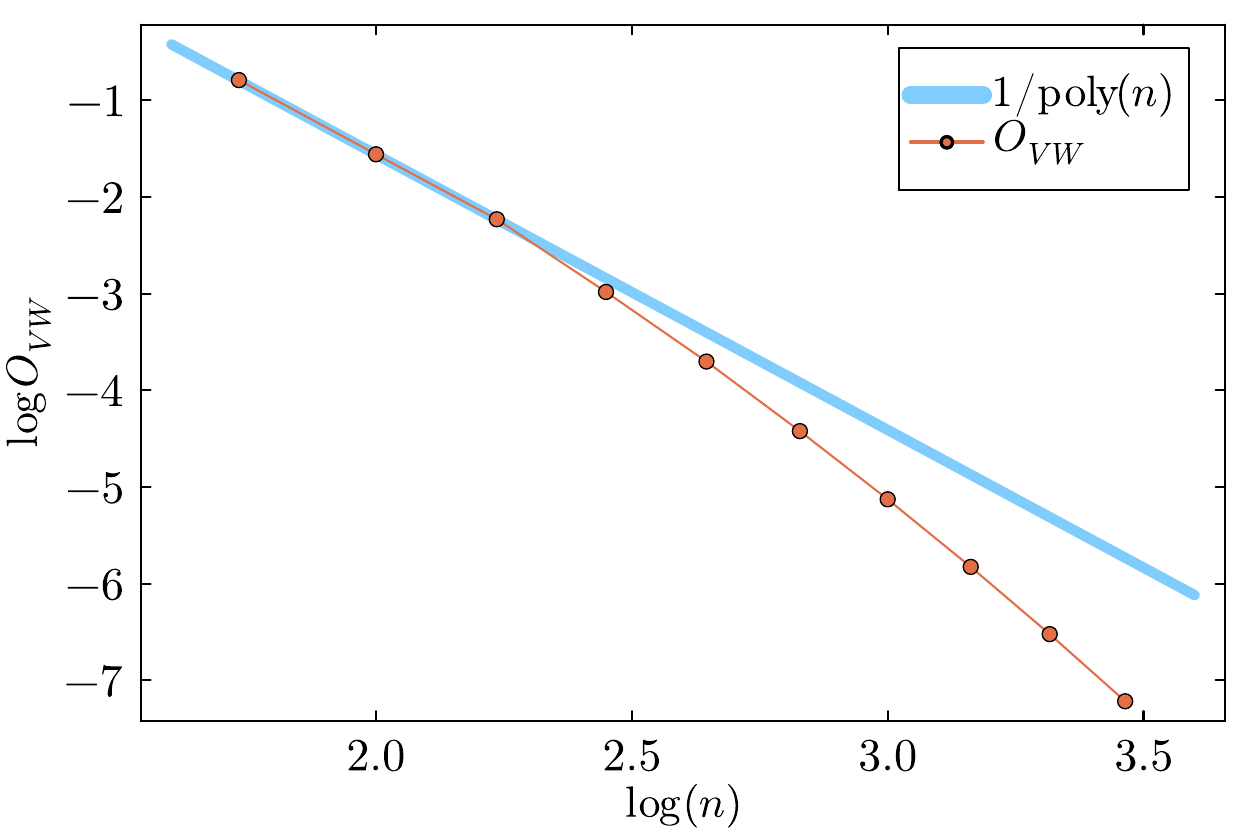}
    \centering
    \caption{
    Scaling of late-time OTOCs $O_{VW}$ for random-phase Hadamard gates \( u=H^{\otimes k}P \). In evaluating OTOCs Eq.\eqref{eq:def-OTOC}, we chose \( V = Z_i \), \( W = Z_j \), \( i \neq j \), and \( k = (\log n)^2 \). The OTOCs are averaged over random isometries and random phase gates \( P \). The time \( t \) at which the OTOC values are extracted is fixed to \( t = 4 \) for all system sizes \( n \). Since the OTOCs approaches to the infinite-time values for \( t \gtrsim 1 \), any such value of \( t \) is sufficient to capture the late-time behavior. If the OTOCs scaled as an inverse polynomial in \( n \), the data would appear as straight lines in the log-log plot. However, the observed curve (red solid line) is concave, indicating that the OTOCs decay faster than any inverse polynomial. The linear function fit with the numerical data (blue straight line) has the slope $\approx -2.85$ in this log-log plot.}
    \label{fig:OTOC-sys-scale}
\end{figure}

\vspace{1em}
\noindent {\bf Pseudorandom State Generator.}\\
We first demonstrate that RSED with chaotic $u$ is a pseudorandom state generator. More precisely, we show that if an ensemble of $u$ generates a pseudorandom state ensemble in the subsystem, then the corresponding ensemble of $U$ in Eq.~\eqref{eq:full-evolution} produces a pseudorandom state ensemble in the entire Hilbert space.

\textit{Theorem 2:}
RSEDs with an ensemble of embedded unitary operators that generate a pseudorandom state ensemble in the subspace with a negligible error generate a pseudorandom state ensemble in the entire space.

\textit{Proof:}
This is proven in {Theorem 6} of {Supplementary Note 9}.

Such an ensemble of chaotic $u$ can be constructed by sampling time evolution operators from a {single, fixed} chaotic $u$, provided that the time interval exceeds its relaxation time. This immediately implies that the corresponding $U$ produces a pseudorandom state ensemble by sampling states in time. This is nicely parallel to the generation of the Haar random state ensemble by sampling states in the time trajectory of a state under chaotic dynamics at a sufficiently long time interval~\cite{Schiulaz_2019}.

Next, we consider an ensemble of $u$ with negligibly small elements with the unbiased mean magnitude of $2^{-k/2}$, e.g., {$u^t$ for $t \geq t_* \approx 1$ with} $u = H^{\otimes k}P$, and show that the corresponding ensemble of RSED can also produce a pseudorandom state ensemble. Hence, such RSED serves as another example of pseudochaotic dynamics. We highlight that the subspace dynamics $u$ does not need to be ergodic, as illustrated in Figure~\ref{fig:overview}(c).

\textit{Theorem 3:}
Let $\mathcal{E}_k$ be an ensemble of unitary operators in {a $k$-qubit subsystem with dimension $K=2^k$}. Let us assume that for all $u\in\mathcal{E}_k$, there exists $\epsilon>0$ such that 
\begin{equation}\label{eq:magnitude-of-prob-for-randomness}
    \operatorname{Pr}\left[ \abs{u_{b,b'}}^2 \geq K^{-\epsilon} \right] \leq \operatorname{negl}(n)
\end{equation}
{for all $b$ and $b'$}. In addition, let us assume that $\mathbb{E}_{u\sim\mathcal{E}_k}\left[ \abs{u_{b,b'}}^2 \right]=K^{-1}$ holds for all $b$ and $b'$. Then, an ensemble of RSEDs with $\mathcal{E}_k$ generates an ensemble of pseudorandom states.

\textit{Proof:}
Let us consider an initial computational state $\ket{p(ba)}$. This evolves under the subsystem embedded dynamics $U$ with an embedded dynamics $u$ as
\begin{equation}\label{eq:evolved-state}
    U\ket{p(ba)} = \sum_{b'\in\{0,1\}^k}u_{b',b}(-1)^{f(ba)+f(b'a)}\ket{p(b'a)}.
\end{equation}
Let $\rho$ be {Hybrid 3} of Ref.~\cite{aaronson2023quantum}, and $\sigma$ be the ensemble average of $U\ket{p(ba)}$ over $\mathcal{E}_k$, random permutations $p$, and random functions $f$. Then, the triangular inequality gives 
\begin{equation}
    \operatorname{TD}\left(\sigma,\rho_\mathrm{Haar}\right) \leq \operatorname{TD}\left(\rho,\rho_\mathrm{Haar}\right) + \operatorname{TD}(\rho,\sigma)
\end{equation}
with $\rho_\mathrm{Haar}=\int d\psi_\mathrm{Haar}\ketbra{\psi}{\psi}^{\otimes t}$. The first term on the right-hand side is negligible due to {Lemma 3} of Ref.~\cite{aaronson2023quantum}. In addition, the second term $\operatorname{TD(\rho,\sigma)}$ is is negligible due to the assumption of $\mathbb{E}_{u\sim\mathcal{E}_k}\left[ \abs{u_{b,b'}}^2 \right]=K^{-1}$ as shown in the proof of {Theorem 7} of {Supplementary Note~9}.

While unitary operators in a general random ensemble have unbiased elements, it is not a necessary condition for a pseudochaotic dynamics. Indeed, even when elements are biased, RSED is pseudochaotic if $u$ generates maximal relative entropy of coherence~\cite{Baumgratz2014} in computational basis with a negligible deviation. This sufficient condition is consistent with the necessary condition of $\omega(\log n)$ coherence introduced in {Ref.}~\cite{haug2023pseudorandom}. More details can be found from {Theorem 8} of {Supplementary Note~11}. 

\vspace{1em}
\noindent {\bf Quantum Circuit and Its Resources}\\
We present an explicit quantum circuit for pseudochaotic RSED and the resources required to implement it.

The circuit consists of three steps, as shown in Figure~\ref{fig:circuit}. The first and last steps apply {quantum secure pseudorandom function and permutation in the entire system} with {$\operatorname{polylog}(n)$ depth circuits under the assumption of the sub-exponential hardness of the Learning with Errors (LWE) problem~\cite{Banerjee2012,Hosoyamada2019,Arunachalam2021,Zhao2024}.} In the middle, {only} the subsystem evolves under dynamics generating nearly maximal coherence in the computational basis. The generation of maximal coherence can be achieved by products of Hadamard gates, so the time complexity for the middle step is $\Omega(1)$.

\begin{figure}[ht]
    \includegraphics[width=\linewidth]{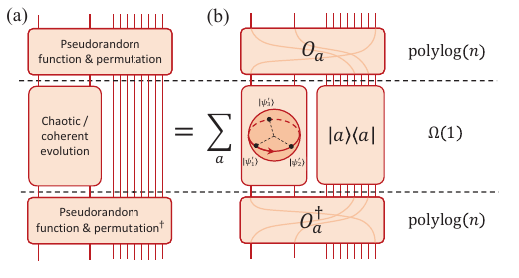}
    \centering
    \caption{Schematic circuit for RSED. Coherent evolution {in (a)} refers to unitary dynamics generating nearly maximal coherence in the subsystem. {The sum of subsystem dynamics conjugated by random isometries {in (b)} can be implemented by} pseudorandom function and permutation {in the entire space in (a)}. {This identity can be derived by inserting a resolution identity operator in the computational basis at the straight lines of the middle step in (a). See Supplementary Note~1 for details.}}
    \label{fig:circuit}
\end{figure}

The pseudochaotic RSED can be turned into a truly chaotic dynamics by increasing the number of entangling and non-Clifford gates in {random permutation and function}
(the first and last steps in Figure~\ref{fig:circuit}(a)), coherent gates in the subsystem dynamics{, and the size of the subsystem} (the middle step in Figure~\ref{fig:circuit}(a)), which brings the ensemble of resulting states closer to Haar random states. More explicitly, the exponential decay of OTOC {in system size}, which is expected for general quantum chaotic systems, can be achieved once pseudorandom isometries are replaced by truly random unitaries by using exponentially many non-Clifford and entangling gates and embedding $\Omega(n)$ Hadamard gates~\cite{lee2024fast}. On the other hand, using only one of these resourceful gates cannot decrease OTOCs and the trace distance with Haar random states. Thus, such a process cannot turn the RSED into chaotic dynamics. 

\textit{Theorem 4:}
Each of entanglement, magic, and coherence of a pseudochaotic RSED can be increased independently without changing other resources and making the RSED chaotic.

\textit{Proof:}
Entanglement, magic, and coherence can be controlled independently by attaching random Clifford gates, random T-gates, and random Hadamard gates, respectively, to the end of the circuit in Figure~\ref{fig:circuit}. Details can be found in {Supplementary Note~12}.

\vspace{1em} 
\noindent{\large{\bf Discussion}}\\
\noindent
In this work, we propose a new concept called pseudochaotic dynamics, which is a non-chaotic dynamics that cannot be distinguished from the maximally chaotic dynamics using polynomial resources. We further introduce the RSED as a systematic way to construct pseudochaotic dynamics and pseudorandom states. {This dynamics can be implemented by pseudorandom permutation and function, which can be implemented by polylogarithmic depth circuits~\cite{NAOR1999,Banerjee2012,Zhandry2015,Hosoyamada2019}, for example in Rydberg atoms or ion-trapped qubits~\cite{Evered2023,Bluvstein2024}. Using this, we expect that pseudochaotic dynamics can be realized in near-term devices with a few dozens of qubits (Supplementary Note~13).} {Although it is not the primary focus of this work, an RSED can generate an approximate state \( t \)-design with the shallowest circuit depth among currently known protocols~\cite{Brand_o_2016, Nakata2017, Ho_2022, Cotler2023, feng2024, haah2024, schuster2024}, which will be detailed separately in Ref.~\cite{lee2024fast}.} {These together} make our RSED highly efficient for tasks such as classical shadow tomography~\cite{Huang_2020}, benchmarking quantum circuits~\cite{Helsen2022}, and even studying black holes~\cite{yoshida2017,Yoshida2019,bouland2019computational, Piroli_2020}.

{Let us highlight the distinctions between our work and previous studies~\cite{aaronson2023quantum,gu2023little,CHAMON2022169086,Chamon2024}. First, our approach clearly contrasts with prior investigations of pseudorandom quantum states~\cite{aaronson2023quantum,gu2023little}, which primarily focused on quantum resource requirements. Instead, we emphasize the dynamical properties such as OTOCs of quantum circuits that generate pseudorandom states. Second, our construction of a pseudorandom state generator is distinct from prior work based on random gates~\cite{Chamon2024}, which exhibit suppression of OTOCs on average by the operator mean field theory~\cite{CHAMON2022169086}. In contrast, our pseudochaotic RSED exhibit suppression of OTOCs for each individual realization. Additionally, we assume the hardness of the LWE problem which is believed to be secure against quantum attacks~\cite{Arunachalam2021,aaronson2023quantum,Zhao2024}, {while Ref.~\cite{Chamon2024} relies on {the assumption that security against a {classical} adaptive chosen-plaintext and chosen-ciphertext attack implies {quantum} security}}. Third, our work introduces a Hamiltonian-based RSED for generating pseudorandom states, opening the door to their realization in analog quantum simulators. This stands in sharp contrast to previous works~\cite{aaronson2023quantum,gu2023little,Chamon2024}, which rely on quantum circuits implemented on digital quantum computers. Lastly, our pseudochaotic RSED achieves the known lower bound on circuit depth for generating pseudorandom states~\cite{aaronson2023quantum,Chamon2024}.} Finally, we note a recent related work that introduces a similar concept of pseudochaos~\cite{gu2024chaos}. While our definition of pseudochaotic dynamics only requires negligible OTOCs and the generation of a pseudorandom state ensemble, Ref.~\cite{gu2024chaos} imposes an additional constraint on the definition of pseudochaos to avoid producing states with high entanglement and magic. Notably, our construction of RSED encompasses the pseudo-Gaussian unitary ensemble introduced in Ref.~\cite{gu2024chaos}, by embedding an ensemble of unitaries whose eigenvalues follow Wigner's semicircle distribution without level repulsions. Clarifying the precise relationship between the two definitions of pseudochaos remains an important direction for future work.

We finish by discussing interesting future research directions. First, it will be interesting to clarify the relation between the two properties, having negligible OTOCs and generating pseudorandom states, of the pseudochaotic dynamics are related. {Second, it would be also interesting to study dynamical properties of RSEDs with various embedded Hamiltonians including integrable ones, which could potentially lead to the discovery of a new class of quantum many-body dynamics.} Another interesting question is to investigate whether typical pseudochaotic dynamics can be used to construct pseudorandom unitaries and whether these dynamics are difficult to simulate with classical algorithms. If both are true, the quantum advantage in random circuit sampling, which relies on typical circuits being close to Haar-random unitaries~\cite{Bremner2016,bouland2019,movassagh2023}, could be demonstrated with significantly lower circuit depth by replacing them with pseudochaotic dynamics. Answering these questions could, therefore, provide a new perspective on the connection between quantum computational advantage and quantum chaos~\cite{Brand_o_2016,Harrow2023}. 


\vspace{1em}  
\noindent{\large{\bf Methods}}\\
\noindent {\bf Out-of-time ordered correlators}\\
A Poisson-bracket out-of-time ordered correlator quantifies how a local operator $W$ spreads under a unitary evolution $U$ by measuring the magnitude of parts of $UWU^\dagger$ commuting with another local operator $V$. Formally, it is defined as
\begin{equation}
    C_{VW}(U) = \frac{1}{2^{n+1}}\tr\left([UWU^\dagger,V]^\dagger[UWU^\dagger,V]\right). 
\end{equation}
If $U$ does not spread $W$ much, then $V$ at almost everywhere commutes with $UWU^\dagger$. Thus, $C_{VW}(U)$ is vanishing. On the other hand, if $U$ is chaotic so makes $UWU^{\dagger}$ be a sum of arbitrary non-local Pauli strings, then $C_{VW}$ saturates to unity. When local operators satisfy $V^2=W^2=I$ like as Pauli operators, then it becomes
\begin{equation}
    C_{VW}(U) = 1 - \real [O_{VW}(U)].
\end{equation}
Here, $O_{VW}(U)$ is the OTOC used in the main text. Any chaotic $U$ makes $O_{VW}(U)$ vanishingly small.

\vspace{1em}  
\noindent {\bf Calculation of OTOCs}\\
Estimation of $O_{VW}(U)$ of a chaotic system is generally challenging as it requires to simulate the system. However, for the RSED, it is possible to calculate $O_{VW}(U)$ both analytically and numerically. Here, we compute $O_{VW}(U)$ with $V=Z_i$ and $W=Z_j$ with $i\neq j$. 

Let $p$ and $f$ be a random permutation and function, respectively. Then, $O_{VW}(U)$ is given by
\begin{equation}
\begin{split}
    O_{VW}(U) 
    &= \frac{1}{2^n} \sum_{\{a_i\}_{i=1}^4,\{b_i\}_{i=1}^8} (-1)^{\sum_{i=1}^8 f(b_i a_i)}\\
    &\times V_{p(b_8 a_4),p(b_1 a_1)}U_{b_1,b_2}W_{p(b_2a_1),p(b_3a_2)}U^\dagger_{b_3,b_4}\\
    &\times V_{p(b_4a_2),p(b_5a_3)}U_{b_5,b_6}W_{p(b_6a_3),p(b_7a_4)}U^\dagger_{b_7,b_8}.
\end{split}
\end{equation}
Here, $\{b_i\}$ and $\{a_i\}$ are summed over $\{0,1\}^k$ and $\{0,1\}^{n-k}$, respectively. Since $V=Z_i$ and $W=Z_j$, this can be simplified as
\begin{equation}
    \begin{split}
    O_{VW}(U) 
    &= \frac{1}{2^n} \sum_{a,\{b_i\}_{i=1}^4} U_{b_1,b_2} U_{b_2,b_3}^\dagger U_{b_3,b_4} U_{b_4,b_1}^\dagger \\
    &\times (-1)^{[p(b_1a)]_i+[p(b_2a)]_j+[p(b_3a)]_i+[p(b_4a)]_j}.
\end{split}
\end{equation}
Numerically, this can be approximately computed by the importance sampling on $a\in\{0,1\}^{n-k}$. The ensemble average of $O_{VW}(U)$ over $f$ is given by 
\begin{equation}\label{eq:closed-form-ZZ-OTOC}
    \mathbb{E}_f [O_{VW}(U)] = \frac{1}{2^k} \tr \left( (U.*U.*U)U^\dagger \right),
\end{equation}
since that of $(-1)^{[p(b_1a)]_i+[p(b_3a)]_i}$ is $\delta_{b_1,b_3}$. Here, $A.*B$ is the element-wise multiplication of $A$ and $B$. More details are deferred to {Supplementary Notes~2 and 3} .

\vspace{1em}  
\noindent {\bf Pseudorandom state ensemble}\\
A pseudorandom state ensemble $\mathcal{E}$ is an ensemble of states that cannot be distinguished by any polynomial copies of states and any poly-time quantum algorithms. For any $t=\operatorname{poly}(n)$, there is no poly-time quantum algorithm $\mathcal{A}$ that satisfies 
\begin{equation}
    \abs{\mathcal{A}(\rho)-\mathcal{A}(\sigma)}
    \geq
    \frac{1}{O(\operatorname{poly}(n))}
\end{equation}
with $\rho = \mathbb{E}_{\phi\sim\mathcal{E}}\left[\ketbra{\phi}{\phi}^{\otimes t}\right]$ and $\sigma=\mathbb{E}_{\psi\sim\mathrm{Haar}}\left[\ketbra{\psi}{\psi}^{\otimes t}\right]$.

\vspace{1em}  
\noindent {\bf Data Availability}\\
{The authors declare that the main data supporting the findings of this study are available within the article and its Supplementary Information files. Source data have been deposited in the Mendeley Data (\href{https://data.mendeley.com/datasets/h7gsjtv27p/1}{DOI:10.17632/h7gsjtv27p.1})(Ref.~\cite{Lee2025pseudodata}).}

\vspace{1em}  
\noindent {\bf References}\\

%


\begin{thebibliography}{78}%
\makeatletter
\providecommand \@ifxundefined [1]{%
 \@ifx{#1\undefined}
}%
\providecommand \@ifnum [1]{%
 \ifnum #1\expandafter \@firstoftwo
 \else \expandafter \@secondoftwo
 \fi
}%
\providecommand \@ifx [1]{%
 \ifx #1\expandafter \@firstoftwo
 \else \expandafter \@secondoftwo
 \fi
}%
\providecommand \natexlab [1]{#1}%
\providecommand \enquote  [1]{``#1''}%
\providecommand \bibnamefont  [1]{#1}%
\providecommand \bibfnamefont [1]{#1}%
\providecommand \citenamefont [1]{#1}%
\providecommand \href@noop [0]{\@secondoftwo}%
\providecommand \href [0]{\begingroup \@sanitize@url \@href}%
\providecommand \@href[1]{\@@startlink{#1}\@@href}%
\providecommand \@@href[1]{\endgroup#1\@@endlink}%
\providecommand \@sanitize@url [0]{\catcode `\\12\catcode `\$12\catcode `\&12\catcode `\#12\catcode `\^12\catcode `\_12\catcode `\%12\relax}%
\providecommand \@@startlink[1]{}%
\providecommand \@@endlink[0]{}%
\providecommand \url  [0]{\begingroup\@sanitize@url \@url }%
\providecommand \@url [1]{\endgroup\@href {#1}{\urlprefix }}%
\providecommand \urlprefix  [0]{URL }%
\providecommand \Eprint [0]{\href }%
\providecommand \doibase [0]{https://doi.org/}%
\providecommand \selectlanguage [0]{\@gobble}%
\providecommand \bibinfo  [0]{\@secondoftwo}%
\providecommand \bibfield  [0]{\@secondoftwo}%
\providecommand \translation [1]{[#1]}%
\providecommand \BibitemOpen [0]{}%
\providecommand \bibitemStop [0]{}%
\providecommand \bibitemNoStop [0]{.\EOS\space}%
\providecommand \EOS [0]{\spacefactor3000\relax}%
\providecommand \BibitemShut  [1]{\csname bibitem#1\endcsname}%
\let\auto@bib@innerbib\@empty
\bibitem [{\citenamefont {Saffman}\ \emph {et~al.}(2010)\citenamefont {Saffman}, \citenamefont {Walker},\ and\ \citenamefont {M\o{}lmer}}]{Saffman_2010}%
  \BibitemOpen
  \bibfield  {author} {\bibinfo {author} {\bibfnamefont {M.}~\bibnamefont {Saffman}}, \bibinfo {author} {\bibfnamefont {T.~G.}\ \bibnamefont {Walker}},\ and\ \bibinfo {author} {\bibfnamefont {K.}~\bibnamefont {M\o{}lmer}},\ }\bibfield  {title} {\bibinfo {title} {Quantum information with rydberg atoms},\ }\href {https://doi.org/10.1103/RevModPhys.82.2313} {\bibfield  {journal} {\bibinfo  {journal} {Rev. Mod. Phys.}\ }\textbf {\bibinfo {volume} {82}},\ \bibinfo {pages} {2313} (\bibinfo {year} {2010})}\BibitemShut {NoStop}%
\bibitem [{\citenamefont {Wendin}(2017)}]{Wendin_2017}%
  \BibitemOpen
  \bibfield  {author} {\bibinfo {author} {\bibfnamefont {G.}~\bibnamefont {Wendin}},\ }\bibfield  {title} {\bibinfo {title} {Quantum information processing with superconducting circuits: a review},\ }\href {https://doi.org/10.1088/1361-6633/aa7e1a} {\bibfield  {journal} {\bibinfo  {journal} {Reports on Progress in Physics}\ }\textbf {\bibinfo {volume} {80}},\ \bibinfo {pages} {106001} (\bibinfo {year} {2017})}\BibitemShut {NoStop}%
\bibitem [{\citenamefont {Pezz\`e}\ \emph {et~al.}(2018)\citenamefont {Pezz\`e}, \citenamefont {Smerzi}, \citenamefont {Oberthaler}, \citenamefont {Schmied},\ and\ \citenamefont {Treutlein}}]{Pezze_2018}%
  \BibitemOpen
  \bibfield  {author} {\bibinfo {author} {\bibfnamefont {L.}~\bibnamefont {Pezz\`e}}, \bibinfo {author} {\bibfnamefont {A.}~\bibnamefont {Smerzi}}, \bibinfo {author} {\bibfnamefont {M.~K.}\ \bibnamefont {Oberthaler}}, \bibinfo {author} {\bibfnamefont {R.}~\bibnamefont {Schmied}},\ and\ \bibinfo {author} {\bibfnamefont {P.}~\bibnamefont {Treutlein}},\ }\bibfield  {title} {\bibinfo {title} {Quantum metrology with nonclassical states of atomic ensembles},\ }\href {https://doi.org/10.1103/RevModPhys.90.035005} {\bibfield  {journal} {\bibinfo  {journal} {Rev. Mod. Phys.}\ }\textbf {\bibinfo {volume} {90}},\ \bibinfo {pages} {035005} (\bibinfo {year} {2018})}\BibitemShut {NoStop}%
\bibitem [{\citenamefont {Monroe}\ \emph {et~al.}(2021)\citenamefont {Monroe}, \citenamefont {Campbell}, \citenamefont {Duan}, \citenamefont {Gong}, \citenamefont {Gorshkov}, \citenamefont {Hess}, \citenamefont {Islam}, \citenamefont {Kim}, \citenamefont {Linke}, \citenamefont {Pagano}, \citenamefont {Richerme}, \citenamefont {Senko},\ and\ \citenamefont {Yao}}]{Monroe2021}%
  \BibitemOpen
  \bibfield  {author} {\bibinfo {author} {\bibfnamefont {C.}~\bibnamefont {Monroe}}, \bibinfo {author} {\bibfnamefont {W.~C.}\ \bibnamefont {Campbell}}, \bibinfo {author} {\bibfnamefont {L.-M.}\ \bibnamefont {Duan}}, \bibinfo {author} {\bibfnamefont {Z.-X.}\ \bibnamefont {Gong}}, \bibinfo {author} {\bibfnamefont {A.~V.}\ \bibnamefont {Gorshkov}}, \bibinfo {author} {\bibfnamefont {P.~W.}\ \bibnamefont {Hess}}, \bibinfo {author} {\bibfnamefont {R.}~\bibnamefont {Islam}}, \bibinfo {author} {\bibfnamefont {K.}~\bibnamefont {Kim}}, \bibinfo {author} {\bibfnamefont {N.~M.}\ \bibnamefont {Linke}}, \bibinfo {author} {\bibfnamefont {G.}~\bibnamefont {Pagano}}, \bibinfo {author} {\bibfnamefont {P.}~\bibnamefont {Richerme}}, \bibinfo {author} {\bibfnamefont {C.}~\bibnamefont {Senko}},\ and\ \bibinfo {author} {\bibfnamefont {N.~Y.}\ \bibnamefont {Yao}},\ }\bibfield  {title} {\bibinfo {title} {Programmable quantum simulations of spin systems with trapped ions},\ }\href {https://doi.org/10.1103/RevModPhys.93.025001}
  {\bibfield  {journal} {\bibinfo  {journal} {Rev. Mod. Phys.}\ }\textbf {\bibinfo {volume} {93}},\ \bibinfo {pages} {025001} (\bibinfo {year} {2021})}\BibitemShut {NoStop}%
\bibitem [{\citenamefont {Eisert}\ \emph {et~al.}(2015)\citenamefont {Eisert}, \citenamefont {Friesdorf},\ and\ \citenamefont {Gogolin}}]{Eisert_2015}%
  \BibitemOpen
  \bibfield  {author} {\bibinfo {author} {\bibfnamefont {J.}~\bibnamefont {Eisert}}, \bibinfo {author} {\bibfnamefont {M.}~\bibnamefont {Friesdorf}},\ and\ \bibinfo {author} {\bibfnamefont {C.}~\bibnamefont {Gogolin}},\ }\bibfield  {title} {\bibinfo {title} {Quantum many-body systems out of equilibrium},\ }\href {https://doi.org/10.1038/nphys3215} {\bibfield  {journal} {\bibinfo  {journal} {Nature Physics}\ }\textbf {\bibinfo {volume} {11}},\ \bibinfo {pages} {124–130} (\bibinfo {year} {2015})}\BibitemShut {NoStop}%
\bibitem [{\citenamefont {Nandkishore}\ and\ \citenamefont {Huse}(2015)}]{Nandkishore_2015}%
  \BibitemOpen
  \bibfield  {author} {\bibinfo {author} {\bibfnamefont {R.}~\bibnamefont {Nandkishore}}\ and\ \bibinfo {author} {\bibfnamefont {D.~A.}\ \bibnamefont {Huse}},\ }\bibfield  {title} {\bibinfo {title} {Many-body localization and thermalization in quantum statistical mechanics},\ }\href {https://doi.org/10.1146/annurev-conmatphys-031214-014726} {\bibfield  {journal} {\bibinfo  {journal} {Annual Review of Condensed Matter Physics}\ }\textbf {\bibinfo {volume} {6}},\ \bibinfo {pages} {15–38} (\bibinfo {year} {2015})}\BibitemShut {NoStop}%
\bibitem [{\citenamefont {D’Alessio}\ \emph {et~al.}(2016)\citenamefont {D’Alessio}, \citenamefont {Kafri}, \citenamefont {Polkovnikov},\ and\ \citenamefont {Rigol}}]{D_Alessio_2016}%
  \BibitemOpen
  \bibfield  {author} {\bibinfo {author} {\bibfnamefont {L.}~\bibnamefont {D’Alessio}}, \bibinfo {author} {\bibfnamefont {Y.}~\bibnamefont {Kafri}}, \bibinfo {author} {\bibfnamefont {A.}~\bibnamefont {Polkovnikov}},\ and\ \bibinfo {author} {\bibfnamefont {M.}~\bibnamefont {Rigol}},\ }\bibfield  {title} {\bibinfo {title} {From quantum chaos and eigenstate thermalization to statistical mechanics and thermodynamics},\ }\href {https://doi.org/10.1080/00018732.2016.1198134} {\bibfield  {journal} {\bibinfo  {journal} {Advances in Physics}\ }\textbf {\bibinfo {volume} {65}},\ \bibinfo {pages} {239–362} (\bibinfo {year} {2016})}\BibitemShut {NoStop}%
\bibitem [{\citenamefont {Dziarmaga}(2010)}]{Dziarmaga_2010}%
  \BibitemOpen
  \bibfield  {author} {\bibinfo {author} {\bibfnamefont {J.}~\bibnamefont {Dziarmaga}},\ }\bibfield  {title} {\bibinfo {title} {Dynamics of a quantum phase transition and relaxation to a steady state},\ }\href {https://doi.org/10.1080/00018732.2010.514702} {\bibfield  {journal} {\bibinfo  {journal} {Advances in Physics}\ }\textbf {\bibinfo {volume} {59}},\ \bibinfo {pages} {1063–1189} (\bibinfo {year} {2010})}\BibitemShut {NoStop}%
\bibitem [{\citenamefont {Heyl}(2018)}]{Heyl_2018}%
  \BibitemOpen
  \bibfield  {author} {\bibinfo {author} {\bibfnamefont {M.}~\bibnamefont {Heyl}},\ }\bibfield  {title} {\bibinfo {title} {Dynamical quantum phase transitions: a review},\ }\href {https://doi.org/10.1088/1361-6633/aaaf9a} {\bibfield  {journal} {\bibinfo  {journal} {Reports on Progress in Physics}\ }\textbf {\bibinfo {volume} {81}},\ \bibinfo {pages} {054001} (\bibinfo {year} {2018})}\BibitemShut {NoStop}%
\bibitem [{\citenamefont {Abanin}\ \emph {et~al.}(2019)\citenamefont {Abanin}, \citenamefont {Altman}, \citenamefont {Bloch},\ and\ \citenamefont {Serbyn}}]{Abanin_2019}%
  \BibitemOpen
  \bibfield  {author} {\bibinfo {author} {\bibfnamefont {D.~A.}\ \bibnamefont {Abanin}}, \bibinfo {author} {\bibfnamefont {E.}~\bibnamefont {Altman}}, \bibinfo {author} {\bibfnamefont {I.}~\bibnamefont {Bloch}},\ and\ \bibinfo {author} {\bibfnamefont {M.}~\bibnamefont {Serbyn}},\ }\bibfield  {title} {\bibinfo {title} {Colloquium: Many-body localization, thermalization, and entanglement},\ }\href {https://doi.org/10.1103/RevModPhys.91.021001} {\bibfield  {journal} {\bibinfo  {journal} {Rev. Mod. Phys.}\ }\textbf {\bibinfo {volume} {91}},\ \bibinfo {pages} {021001} (\bibinfo {year} {2019})}\BibitemShut {NoStop}%
\bibitem [{\citenamefont {Lunin}\ and\ \citenamefont {Mathur}(2002)}]{Lunin_2002}%
  \BibitemOpen
  \bibfield  {author} {\bibinfo {author} {\bibfnamefont {O.}~\bibnamefont {Lunin}}\ and\ \bibinfo {author} {\bibfnamefont {S.~D.}\ \bibnamefont {Mathur}},\ }\bibfield  {title} {\bibinfo {title} {Ads/cft duality and the black hole information paradox},\ }\href {https://doi.org/10.1016/s0550-3213(01)00620-4} {\bibfield  {journal} {\bibinfo  {journal} {Nuclear Physics B}\ }\textbf {\bibinfo {volume} {623}},\ \bibinfo {pages} {342–394} (\bibinfo {year} {2002})}\BibitemShut {NoStop}%
\bibitem [{\citenamefont {Maldacena}\ and\ \citenamefont {Stanford}(2016)}]{Maldacena2016SYK}%
  \BibitemOpen
  \bibfield  {author} {\bibinfo {author} {\bibfnamefont {J.}~\bibnamefont {Maldacena}}\ and\ \bibinfo {author} {\bibfnamefont {D.}~\bibnamefont {Stanford}},\ }\bibfield  {title} {\bibinfo {title} {Remarks on the sachdev-ye-kitaev model},\ }\href {https://doi.org/10.1103/PhysRevD.94.106002} {\bibfield  {journal} {\bibinfo  {journal} {Phys. Rev. D}\ }\textbf {\bibinfo {volume} {94}},\ \bibinfo {pages} {106002} (\bibinfo {year} {2016})}\BibitemShut {NoStop}%
\bibitem [{\citenamefont {Maldacena}\ \emph {et~al.}(2016)\citenamefont {Maldacena}, \citenamefont {Shenker},\ and\ \citenamefont {Stanford}}]{Maldacena2016chaos}%
  \BibitemOpen
  \bibfield  {author} {\bibinfo {author} {\bibfnamefont {J.}~\bibnamefont {Maldacena}}, \bibinfo {author} {\bibfnamefont {S.~H.}\ \bibnamefont {Shenker}},\ and\ \bibinfo {author} {\bibfnamefont {D.}~\bibnamefont {Stanford}},\ }\bibfield  {title} {\bibinfo {title} {A bound on chaos},\ }\bibfield  {journal} {\bibinfo  {journal} {Journal of High Energy Physics}\ }\textbf {\bibinfo {volume} {2016}},\ \href {https://doi.org/10.1007/jhep08(2016)106} {10.1007/jhep08(2016)106} (\bibinfo {year} {2016})\BibitemShut {NoStop}%
\bibitem [{\citenamefont {Sachdev}\ and\ \citenamefont {Ye}(1993)}]{Sachdev1993}%
  \BibitemOpen
  \bibfield  {author} {\bibinfo {author} {\bibfnamefont {S.}~\bibnamefont {Sachdev}}\ and\ \bibinfo {author} {\bibfnamefont {J.}~\bibnamefont {Ye}},\ }\bibfield  {title} {\bibinfo {title} {Gapless spin-fluid ground state in a random quantum heisenberg magnet},\ }\href {https://doi.org/10.1103/PhysRevLett.70.3339} {\bibfield  {journal} {\bibinfo  {journal} {Phys. Rev. Lett.}\ }\textbf {\bibinfo {volume} {70}},\ \bibinfo {pages} {3339} (\bibinfo {year} {1993})}\BibitemShut {NoStop}%
\bibitem [{\citenamefont {Kitaev}(2015)}]{Kitaev2015}%
  \BibitemOpen
  \bibfield  {author} {\bibinfo {author} {\bibfnamefont {A.}~\bibnamefont {Kitaev}},\ }\bibfield  {title} {\bibinfo {title} {Talks at {KITP}, {University of California}, {Santa Barbara}},\ }\href {https://online.kitp.ucsb.edu/online/entangled15} {\bibfield  {journal} {\bibinfo  {journal} {Entanglement in Strongly-Correlated Quantum Matter}\ } (\bibinfo {year} {2015})}\BibitemShut {NoStop}%
\bibitem [{\citenamefont {Fisher}\ \emph {et~al.}(2023)\citenamefont {Fisher}, \citenamefont {Khemani}, \citenamefont {Nahum},\ and\ \citenamefont {Vijay}}]{Matthew2023}%
  \BibitemOpen
  \bibfield  {author} {\bibinfo {author} {\bibfnamefont {M.~P.}\ \bibnamefont {Fisher}}, \bibinfo {author} {\bibfnamefont {V.}~\bibnamefont {Khemani}}, \bibinfo {author} {\bibfnamefont {A.}~\bibnamefont {Nahum}},\ and\ \bibinfo {author} {\bibfnamefont {S.}~\bibnamefont {Vijay}},\ }\bibfield  {title} {\bibinfo {title} {Random quantum circuits},\ }\href {https://doi.org/https://doi.org/10.1146/annurev-conmatphys-031720-030658} {\bibfield  {journal} {\bibinfo  {journal} {Annual Review of Condensed Matter Physics}\ }\textbf {\bibinfo {volume} {14}},\ \bibinfo {pages} {335} (\bibinfo {year} {2023})}\BibitemShut {NoStop}%
\bibitem [{\citenamefont {Lieb}\ and\ \citenamefont {Robinson}(1972)}]{Lieb1972}%
  \BibitemOpen
  \bibfield  {author} {\bibinfo {author} {\bibfnamefont {E.~H.}\ \bibnamefont {Lieb}}\ and\ \bibinfo {author} {\bibfnamefont {D.~W.}\ \bibnamefont {Robinson}},\ }\bibfield  {title} {\bibinfo {title} {The finite group velocity of quantum spin systems},\ }\href {https://doi.org/10.1007/BF01645779} {\bibfield  {journal} {\bibinfo  {journal} {Communications in Mathematical Physics}\ }\textbf {\bibinfo {volume} {28}},\ \bibinfo {pages} {251} (\bibinfo {year} {1972})}\BibitemShut {NoStop}%
\bibitem [{\citenamefont {Shenker}\ and\ \citenamefont {Stanford}(2014)}]{Shenker_2014}%
  \BibitemOpen
  \bibfield  {author} {\bibinfo {author} {\bibfnamefont {S.~H.}\ \bibnamefont {Shenker}}\ and\ \bibinfo {author} {\bibfnamefont {D.}~\bibnamefont {Stanford}},\ }\bibfield  {title} {\bibinfo {title} {Black holes and the butterfly effect},\ }\bibfield  {journal} {\bibinfo  {journal} {Journal of High Energy Physics}\ }\textbf {\bibinfo {volume} {2014}},\ \href {https://doi.org/10.1007/jhep03(2014)067} {10.1007/jhep03(2014)067} (\bibinfo {year} {2014})\BibitemShut {NoStop}%
\bibitem [{\citenamefont {Roberts}\ and\ \citenamefont {Yoshida}(2017)}]{Roberts2017}%
  \BibitemOpen
  \bibfield  {author} {\bibinfo {author} {\bibfnamefont {D.~A.}\ \bibnamefont {Roberts}}\ and\ \bibinfo {author} {\bibfnamefont {B.}~\bibnamefont {Yoshida}},\ }\bibfield  {title} {\bibinfo {title} {Chaos and complexity by design},\ }\href {https://doi.org/10.1007/JHEP04(2017)121} {\bibfield  {journal} {\bibinfo  {journal} {Journal of High Energy Physics}\ }\textbf {\bibinfo {volume} {2017}},\ \bibinfo {pages} {121} (\bibinfo {year} {2017})}\BibitemShut {NoStop}%
\bibitem [{\citenamefont {Lewis-Swan}\ \emph {et~al.}(2019)\citenamefont {Lewis-Swan}, \citenamefont {Safavi-Naini}, \citenamefont {Kaufman},\ and\ \citenamefont {Rey}}]{Lewis2019}%
  \BibitemOpen
  \bibfield  {author} {\bibinfo {author} {\bibfnamefont {R.~J.}\ \bibnamefont {Lewis-Swan}}, \bibinfo {author} {\bibfnamefont {A.}~\bibnamefont {Safavi-Naini}}, \bibinfo {author} {\bibfnamefont {A.~M.}\ \bibnamefont {Kaufman}},\ and\ \bibinfo {author} {\bibfnamefont {A.~M.}\ \bibnamefont {Rey}},\ }\bibfield  {title} {\bibinfo {title} {Dynamics of quantum information},\ }\href {https://doi.org/10.1038/s42254-019-0090-y} {\bibfield  {journal} {\bibinfo  {journal} {Nature Reviews Physics}\ }\textbf {\bibinfo {volume} {1}},\ \bibinfo {pages} {627–634} (\bibinfo {year} {2019})}\BibitemShut {NoStop}%
\bibitem [{\citenamefont {Joshi}\ \emph {et~al.}(2022)\citenamefont {Joshi}, \citenamefont {Elben}, \citenamefont {Vikram}, \citenamefont {Vermersch}, \citenamefont {Galitski},\ and\ \citenamefont {Zoller}}]{Joshi2022}%
  \BibitemOpen
  \bibfield  {author} {\bibinfo {author} {\bibfnamefont {L.~K.}\ \bibnamefont {Joshi}}, \bibinfo {author} {\bibfnamefont {A.}~\bibnamefont {Elben}}, \bibinfo {author} {\bibfnamefont {A.}~\bibnamefont {Vikram}}, \bibinfo {author} {\bibfnamefont {B.}~\bibnamefont {Vermersch}}, \bibinfo {author} {\bibfnamefont {V.}~\bibnamefont {Galitski}},\ and\ \bibinfo {author} {\bibfnamefont {P.}~\bibnamefont {Zoller}},\ }\bibfield  {title} {\bibinfo {title} {Probing many-body quantum chaos with quantum simulators},\ }\href {https://doi.org/10.1103/PhysRevX.12.011018} {\bibfield  {journal} {\bibinfo  {journal} {Phys. Rev. X}\ }\textbf {\bibinfo {volume} {12}},\ \bibinfo {pages} {011018} (\bibinfo {year} {2022})}\BibitemShut {NoStop}%
\bibitem [{\citenamefont {Preskill}(2012)}]{preskill2012}%
  \BibitemOpen
  \bibfield  {author} {\bibinfo {author} {\bibfnamefont {J.}~\bibnamefont {Preskill}},\ }\href {https://arxiv.org/abs/1203.5813} {\bibinfo {title} {Quantum computing and the entanglement frontier}} (\bibinfo {year} {2012}),\ \Eprint {https://arxiv.org/abs/1203.5813} {arXiv:1203.5813 [quant-ph]} \BibitemShut {NoStop}%
\bibitem [{\citenamefont {Boixo}\ \emph {et~al.}(2018)\citenamefont {Boixo}, \citenamefont {Isakov}, \citenamefont {Smelyanskiy}, \citenamefont {Babbush}, \citenamefont {Ding}, \citenamefont {Jiang}, \citenamefont {Bremner}, \citenamefont {Martinis},\ and\ \citenamefont {Neven}}]{Boixo_2018}%
  \BibitemOpen
  \bibfield  {author} {\bibinfo {author} {\bibfnamefont {S.}~\bibnamefont {Boixo}}, \bibinfo {author} {\bibfnamefont {S.~V.}\ \bibnamefont {Isakov}}, \bibinfo {author} {\bibfnamefont {V.~N.}\ \bibnamefont {Smelyanskiy}}, \bibinfo {author} {\bibfnamefont {R.}~\bibnamefont {Babbush}}, \bibinfo {author} {\bibfnamefont {N.}~\bibnamefont {Ding}}, \bibinfo {author} {\bibfnamefont {Z.}~\bibnamefont {Jiang}}, \bibinfo {author} {\bibfnamefont {M.~J.}\ \bibnamefont {Bremner}}, \bibinfo {author} {\bibfnamefont {J.~M.}\ \bibnamefont {Martinis}},\ and\ \bibinfo {author} {\bibfnamefont {H.}~\bibnamefont {Neven}},\ }\bibfield  {title} {\bibinfo {title} {Characterizing quantum supremacy in near-term devices},\ }\href {https://doi.org/10.1038/s41567-018-0124-x} {\bibfield  {journal} {\bibinfo  {journal} {Nature Physics}\ }\textbf {\bibinfo {volume} {14}},\ \bibinfo {pages} {595–600} (\bibinfo {year} {2018})}\BibitemShut {NoStop}%
\bibitem [{\citenamefont {Arute}\ \emph {et~al.}(2019)\citenamefont {Arute}, \citenamefont {Arya}, \citenamefont {Babbush}, \citenamefont {Bacon}, \citenamefont {Bardin}, \citenamefont {Barends}, \citenamefont {Biswas}, \citenamefont {Boixo}, \citenamefont {Brandao}, \citenamefont {Buell} \emph {et~al.}}]{arute2019}%
  \BibitemOpen
  \bibfield  {author} {\bibinfo {author} {\bibfnamefont {F.}~\bibnamefont {Arute}}, \bibinfo {author} {\bibfnamefont {K.}~\bibnamefont {Arya}}, \bibinfo {author} {\bibfnamefont {R.}~\bibnamefont {Babbush}}, \bibinfo {author} {\bibfnamefont {D.}~\bibnamefont {Bacon}}, \bibinfo {author} {\bibfnamefont {J.~C.}\ \bibnamefont {Bardin}}, \bibinfo {author} {\bibfnamefont {R.}~\bibnamefont {Barends}}, \bibinfo {author} {\bibfnamefont {R.}~\bibnamefont {Biswas}}, \bibinfo {author} {\bibfnamefont {S.}~\bibnamefont {Boixo}}, \bibinfo {author} {\bibfnamefont {F.~G.}\ \bibnamefont {Brandao}}, \bibinfo {author} {\bibfnamefont {D.~A.}\ \bibnamefont {Buell}}, \emph {et~al.},\ }\bibfield  {title} {\bibinfo {title} {Quantum supremacy using a programmable superconducting processor},\ }\href {https://doi.org/https://doi.org/10.1038/s41586-019-1666-5} {\bibfield  {journal} {\bibinfo  {journal} {Nature}\ }\textbf {\bibinfo {volume} {574}},\ \bibinfo {pages} {505} (\bibinfo {year} {2019})}\BibitemShut {NoStop}%
\bibitem [{\citenamefont {Aaronson}\ and\ \citenamefont {Gunn}(2020)}]{aaronson2020}%
  \BibitemOpen
  \bibfield  {author} {\bibinfo {author} {\bibfnamefont {S.}~\bibnamefont {Aaronson}}\ and\ \bibinfo {author} {\bibfnamefont {S.}~\bibnamefont {Gunn}},\ }\href {https://doi.org/10.4086/toc.2020.v016a011} {\bibinfo {title} {On the classical hardness of spoofing linear cross-entropy benchmarking}} (\bibinfo {year} {2020}),\ \Eprint {https://arxiv.org/abs/1910.12085} {arXiv:1910.12085 [quant-ph]} \BibitemShut {NoStop}%
\bibitem [{\citenamefont {Brandão}\ \emph {et~al.}(2016)\citenamefont {Brandão}, \citenamefont {Harrow},\ and\ \citenamefont {Horodecki}}]{Brand_o_2016}%
  \BibitemOpen
  \bibfield  {author} {\bibinfo {author} {\bibfnamefont {F.~G. S.~L.}\ \bibnamefont {Brandão}}, \bibinfo {author} {\bibfnamefont {A.~W.}\ \bibnamefont {Harrow}},\ and\ \bibinfo {author} {\bibfnamefont {M.}~\bibnamefont {Horodecki}},\ }\bibfield  {title} {\bibinfo {title} {Local random quantum circuits are approximate polynomial-designs},\ }\href {https://doi.org/10.1007/s00220-016-2706-8} {\bibfield  {journal} {\bibinfo  {journal} {Communications in Mathematical Physics}\ }\textbf {\bibinfo {volume} {346}},\ \bibinfo {pages} {397–434} (\bibinfo {year} {2016})}\BibitemShut {NoStop}%
\bibitem [{\citenamefont {Nakata}\ \emph {et~al.}(2017)\citenamefont {Nakata}, \citenamefont {Hirche}, \citenamefont {Koashi},\ and\ \citenamefont {Winter}}]{Nakata2017}%
  \BibitemOpen
  \bibfield  {author} {\bibinfo {author} {\bibfnamefont {Y.}~\bibnamefont {Nakata}}, \bibinfo {author} {\bibfnamefont {C.}~\bibnamefont {Hirche}}, \bibinfo {author} {\bibfnamefont {M.}~\bibnamefont {Koashi}},\ and\ \bibinfo {author} {\bibfnamefont {A.}~\bibnamefont {Winter}},\ }\bibfield  {title} {\bibinfo {title} {Efficient quantum pseudorandomness with nearly time-independent hamiltonian dynamics},\ }\href {https://doi.org/10.1103/PhysRevX.7.021006} {\bibfield  {journal} {\bibinfo  {journal} {Phys. Rev. X}\ }\textbf {\bibinfo {volume} {7}},\ \bibinfo {pages} {021006} (\bibinfo {year} {2017})}\BibitemShut {NoStop}%
\bibitem [{\citenamefont {Ho}\ and\ \citenamefont {Choi}(2022)}]{Ho_2022}%
  \BibitemOpen
  \bibfield  {author} {\bibinfo {author} {\bibfnamefont {W.~W.}\ \bibnamefont {Ho}}\ and\ \bibinfo {author} {\bibfnamefont {S.}~\bibnamefont {Choi}},\ }\bibfield  {title} {\bibinfo {title} {Exact emergent quantum state designs from quantum chaotic dynamics},\ }\bibfield  {journal} {\bibinfo  {journal} {Physical Review Letters}\ }\textbf {\bibinfo {volume} {128}},\ \href {https://doi.org/10.1103/physrevlett.128.060601} {10.1103/physrevlett.128.060601} (\bibinfo {year} {2022})\BibitemShut {NoStop}%
\bibitem [{\citenamefont {Cotler}\ \emph {et~al.}(2023)\citenamefont {Cotler}, \citenamefont {Mark}, \citenamefont {Huang}, \citenamefont {Hern\'andez}, \citenamefont {Choi}, \citenamefont {Shaw}, \citenamefont {Endres},\ and\ \citenamefont {Choi}}]{Cotler2023}%
  \BibitemOpen
  \bibfield  {author} {\bibinfo {author} {\bibfnamefont {J.~S.}\ \bibnamefont {Cotler}}, \bibinfo {author} {\bibfnamefont {D.~K.}\ \bibnamefont {Mark}}, \bibinfo {author} {\bibfnamefont {H.-Y.}\ \bibnamefont {Huang}}, \bibinfo {author} {\bibfnamefont {F.}~\bibnamefont {Hern\'andez}}, \bibinfo {author} {\bibfnamefont {J.}~\bibnamefont {Choi}}, \bibinfo {author} {\bibfnamefont {A.~L.}\ \bibnamefont {Shaw}}, \bibinfo {author} {\bibfnamefont {M.}~\bibnamefont {Endres}},\ and\ \bibinfo {author} {\bibfnamefont {S.}~\bibnamefont {Choi}},\ }\bibfield  {title} {\bibinfo {title} {Emergent quantum state designs from individual many-body wave functions},\ }\href {https://doi.org/10.1103/PRXQuantum.4.010311} {\bibfield  {journal} {\bibinfo  {journal} {PRX Quantum}\ }\textbf {\bibinfo {volume} {4}},\ \bibinfo {pages} {010311} (\bibinfo {year} {2023})}\BibitemShut {NoStop}%
\bibitem [{\citenamefont {Ananth}\ \emph {et~al.}(2022)\citenamefont {Ananth}, \citenamefont {Qian},\ and\ \citenamefont {Yuen}}]{ananth2022}%
  \BibitemOpen
  \bibfield  {author} {\bibinfo {author} {\bibfnamefont {P.}~\bibnamefont {Ananth}}, \bibinfo {author} {\bibfnamefont {L.}~\bibnamefont {Qian}},\ and\ \bibinfo {author} {\bibfnamefont {H.}~\bibnamefont {Yuen}},\ }\bibfield  {title} {\bibinfo {title} {Cryptography from pseudorandom quantum states},\ }in\ \href {https://doi.org/https://doi.org/10.1007/978-3-031-15802-5_8} {\emph {\bibinfo {booktitle} {Annual International Cryptology Conference}}}\ (\bibinfo {organization} {Springer},\ \bibinfo {year} {2022})\ pp.\ \bibinfo {pages} {208--236}\BibitemShut {NoStop}%
\bibitem [{\citenamefont {Kretschmer}\ \emph {et~al.}(2023)\citenamefont {Kretschmer}, \citenamefont {Qian}, \citenamefont {Sinha},\ and\ \citenamefont {Tal}}]{kretschmer2023}%
  \BibitemOpen
  \bibfield  {author} {\bibinfo {author} {\bibfnamefont {W.}~\bibnamefont {Kretschmer}}, \bibinfo {author} {\bibfnamefont {L.}~\bibnamefont {Qian}}, \bibinfo {author} {\bibfnamefont {M.}~\bibnamefont {Sinha}},\ and\ \bibinfo {author} {\bibfnamefont {A.}~\bibnamefont {Tal}},\ }\bibfield  {title} {\bibinfo {title} {Quantum cryptography in algorithmica},\ }in\ \href {https://arxiv.org/abs/2212.00879} {\emph {\bibinfo {booktitle} {Proceedings of the 55th Annual ACM Symposium on Theory of Computing}}}\ (\bibinfo {year} {2023})\ pp.\ \bibinfo {pages} {1589--1602}\BibitemShut {NoStop}%
\bibitem [{\citenamefont {Knill}\ \emph {et~al.}(2008)\citenamefont {Knill}, \citenamefont {Leibfried}, \citenamefont {Reichle}, \citenamefont {Britton}, \citenamefont {Blakestad}, \citenamefont {Jost}, \citenamefont {Langer}, \citenamefont {Ozeri}, \citenamefont {Seidelin},\ and\ \citenamefont {Wineland}}]{Knill_2008}%
  \BibitemOpen
  \bibfield  {author} {\bibinfo {author} {\bibfnamefont {E.}~\bibnamefont {Knill}}, \bibinfo {author} {\bibfnamefont {D.}~\bibnamefont {Leibfried}}, \bibinfo {author} {\bibfnamefont {R.}~\bibnamefont {Reichle}}, \bibinfo {author} {\bibfnamefont {J.}~\bibnamefont {Britton}}, \bibinfo {author} {\bibfnamefont {R.~B.}\ \bibnamefont {Blakestad}}, \bibinfo {author} {\bibfnamefont {J.~D.}\ \bibnamefont {Jost}}, \bibinfo {author} {\bibfnamefont {C.}~\bibnamefont {Langer}}, \bibinfo {author} {\bibfnamefont {R.}~\bibnamefont {Ozeri}}, \bibinfo {author} {\bibfnamefont {S.}~\bibnamefont {Seidelin}},\ and\ \bibinfo {author} {\bibfnamefont {D.~J.}\ \bibnamefont {Wineland}},\ }\bibfield  {title} {\bibinfo {title} {Randomized benchmarking of quantum gates},\ }\bibfield  {journal} {\bibinfo  {journal} {Physical Review A}\ }\textbf {\bibinfo {volume} {77}},\ \href {https://doi.org/10.1103/physreva.77.012307} {10.1103/physreva.77.012307} (\bibinfo {year} {2008})\BibitemShut {NoStop}%
\bibitem [{\citenamefont {Huang}\ \emph {et~al.}(2020)\citenamefont {Huang}, \citenamefont {Kueng},\ and\ \citenamefont {Preskill}}]{Huang_2020}%
  \BibitemOpen
  \bibfield  {author} {\bibinfo {author} {\bibfnamefont {H.-Y.}\ \bibnamefont {Huang}}, \bibinfo {author} {\bibfnamefont {R.}~\bibnamefont {Kueng}},\ and\ \bibinfo {author} {\bibfnamefont {J.}~\bibnamefont {Preskill}},\ }\bibfield  {title} {\bibinfo {title} {Predicting many properties of a quantum system from very few measurements},\ }\href {https://doi.org/10.1038/s41567-020-0932-7} {\bibfield  {journal} {\bibinfo  {journal} {Nature Physics}\ }\textbf {\bibinfo {volume} {16}},\ \bibinfo {pages} {1050–1057} (\bibinfo {year} {2020})}\BibitemShut {NoStop}%
\bibitem [{\citenamefont {Huang}(2022)}]{huang2022}%
  \BibitemOpen
  \bibfield  {author} {\bibinfo {author} {\bibfnamefont {H.-Y.}\ \bibnamefont {Huang}},\ }\bibfield  {title} {\bibinfo {title} {Learning quantum states from their classical shadows},\ }\href {https://doi.org/10.1038/s42254-021-00411-5} {\bibfield  {journal} {\bibinfo  {journal} {Nature Reviews Physics}\ }\textbf {\bibinfo {volume} {4}},\ \bibinfo {pages} {81} (\bibinfo {year} {2022})}\BibitemShut {NoStop}%
\bibitem [{\citenamefont {Bouland}\ \emph {et~al.}(2019)\citenamefont {Bouland}, \citenamefont {Fefferman}, \citenamefont {Nirkhe},\ and\ \citenamefont {Vazirani}}]{bouland2019}%
  \BibitemOpen
  \bibfield  {author} {\bibinfo {author} {\bibfnamefont {A.}~\bibnamefont {Bouland}}, \bibinfo {author} {\bibfnamefont {B.}~\bibnamefont {Fefferman}}, \bibinfo {author} {\bibfnamefont {C.}~\bibnamefont {Nirkhe}},\ and\ \bibinfo {author} {\bibfnamefont {U.}~\bibnamefont {Vazirani}},\ }\bibfield  {title} {\bibinfo {title} {On the complexity and verification of quantum random circuit sampling},\ }\href {https://doi.org/https://doi.org/10.1038/s41567-018-0318-2} {\bibfield  {journal} {\bibinfo  {journal} {Nature Physics}\ }\textbf {\bibinfo {volume} {15}},\ \bibinfo {pages} {159} (\bibinfo {year} {2019})}\BibitemShut {NoStop}%
\bibitem [{\citenamefont {Knill}(1995)}]{knill1995}%
  \BibitemOpen
  \bibfield  {author} {\bibinfo {author} {\bibfnamefont {E.}~\bibnamefont {Knill}},\ }\href {https://arxiv.org/abs/quant-ph/9508006} {\bibinfo {title} {Approximation by quantum circuits}} (\bibinfo {year} {1995}),\ \Eprint {https://arxiv.org/abs/quant-ph/9508006} {arXiv:quant-ph/9508006 [quant-ph]} \BibitemShut {NoStop}%
\bibitem [{\citenamefont {Ji}\ \emph {et~al.}(2018)\citenamefont {Ji}, \citenamefont {Liu},\ and\ \citenamefont {Song}}]{Ji2018}%
  \BibitemOpen
  \bibfield  {author} {\bibinfo {author} {\bibfnamefont {Z.}~\bibnamefont {Ji}}, \bibinfo {author} {\bibfnamefont {Y.-K.}\ \bibnamefont {Liu}},\ and\ \bibinfo {author} {\bibfnamefont {F.}~\bibnamefont {Song}},\ }\bibfield  {title} {\bibinfo {title} {Pseudorandom quantum states},\ }in\ \href {https://doi.org/https://doi.org/10.1007/978-3-319-96878-0_5} {\emph {\bibinfo {booktitle} {Advances in Cryptology -- CRYPTO 2018}}},\ \bibinfo {editor} {edited by\ \bibinfo {editor} {\bibfnamefont {H.}~\bibnamefont {Shacham}}\ and\ \bibinfo {editor} {\bibfnamefont {A.}~\bibnamefont {Boldyreva}}}\ (\bibinfo  {publisher} {Springer International Publishing},\ \bibinfo {address} {Cham},\ \bibinfo {year} {2018})\ pp.\ \bibinfo {pages} {126--152}\BibitemShut {NoStop}%
\bibitem [{\citenamefont {Aaronson}\ \emph {et~al.}(2023)\citenamefont {Aaronson}, \citenamefont {Bouland}, \citenamefont {Fefferman}, \citenamefont {Ghosh}, \citenamefont {Vazirani}, \citenamefont {Zhang},\ and\ \citenamefont {Zhou}}]{aaronson2023quantum}%
  \BibitemOpen
  \bibfield  {author} {\bibinfo {author} {\bibfnamefont {S.}~\bibnamefont {Aaronson}}, \bibinfo {author} {\bibfnamefont {A.}~\bibnamefont {Bouland}}, \bibinfo {author} {\bibfnamefont {B.}~\bibnamefont {Fefferman}}, \bibinfo {author} {\bibfnamefont {S.}~\bibnamefont {Ghosh}}, \bibinfo {author} {\bibfnamefont {U.}~\bibnamefont {Vazirani}}, \bibinfo {author} {\bibfnamefont {C.}~\bibnamefont {Zhang}},\ and\ \bibinfo {author} {\bibfnamefont {Z.}~\bibnamefont {Zhou}},\ }\href@noop {} {\bibinfo {title} {Quantum pseudoentanglement}} (\bibinfo {year} {2023}),\ \Eprint {https://arxiv.org/abs/2211.00747} {arXiv:2211.00747 [quant-ph]} \BibitemShut {NoStop}%
\bibitem [{\citenamefont {Gu}\ \emph {et~al.}(2024{\natexlab{a}})\citenamefont {Gu}, \citenamefont {Leone}, \citenamefont {Ghosh}, \citenamefont {Eisert}, \citenamefont {Yelin},\ and\ \citenamefont {Quek}}]{gu2023little}%
  \BibitemOpen
  \bibfield  {author} {\bibinfo {author} {\bibfnamefont {A.}~\bibnamefont {Gu}}, \bibinfo {author} {\bibfnamefont {L.}~\bibnamefont {Leone}}, \bibinfo {author} {\bibfnamefont {S.}~\bibnamefont {Ghosh}}, \bibinfo {author} {\bibfnamefont {J.}~\bibnamefont {Eisert}}, \bibinfo {author} {\bibfnamefont {S.~F.}\ \bibnamefont {Yelin}},\ and\ \bibinfo {author} {\bibfnamefont {Y.}~\bibnamefont {Quek}},\ }\bibfield  {title} {\bibinfo {title} {Pseudomagic quantum states},\ }\bibfield  {journal} {\bibinfo  {journal} {Physical Review Letters}\ }\textbf {\bibinfo {volume} {132}},\ \href {https://doi.org/10.1103/physrevlett.132.210602} {10.1103/physrevlett.132.210602} (\bibinfo {year} {2024}{\natexlab{a}})\BibitemShut {NoStop}%
\bibitem [{\citenamefont {Haug}\ \emph {et~al.}(2023)\citenamefont {Haug}, \citenamefont {Bharti},\ and\ \citenamefont {Koh}}]{haug2023pseudorandom}%
  \BibitemOpen
  \bibfield  {author} {\bibinfo {author} {\bibfnamefont {T.}~\bibnamefont {Haug}}, \bibinfo {author} {\bibfnamefont {K.}~\bibnamefont {Bharti}},\ and\ \bibinfo {author} {\bibfnamefont {D.~E.}\ \bibnamefont {Koh}},\ }\href@noop {} {\bibinfo {title} {Pseudorandom unitaries are neither real nor sparse nor noise-robust}} (\bibinfo {year} {2023}),\ \Eprint {https://arxiv.org/abs/2306.11677} {arXiv:2306.11677 [quant-ph]} \BibitemShut {NoStop}%
\bibitem [{\citenamefont {Kretschmer}(2021)}]{kretschmer2021}%
  \BibitemOpen
  \bibfield  {author} {\bibinfo {author} {\bibfnamefont {W.}~\bibnamefont {Kretschmer}},\ }\bibfield  {title} {\bibinfo {title} {{Quantum Pseudorandomness and Classical Complexity}},\ }in\ \href {https://doi.org/10.4230/LIPIcs.TQC.2021.2} {\emph {\bibinfo {booktitle} {16th Conference on the Theory of Quantum Computation, Communication and Cryptography (TQC 2021)}}},\ \bibinfo {series} {Leibniz International Proceedings in Informatics (LIPIcs)}, Vol.\ \bibinfo {volume} {197},\ \bibinfo {editor} {edited by\ \bibinfo {editor} {\bibfnamefont {M.-H.}\ \bibnamefont {Hsieh}}}\ (\bibinfo  {publisher} {Schloss Dagstuhl -- Leibniz-Zentrum f{\"u}r Informatik},\ \bibinfo {address} {Dagstuhl, Germany},\ \bibinfo {year} {2021})\ pp.\ \bibinfo {pages} {2:1--2:20}\BibitemShut {NoStop}%
\bibitem [{\citenamefont {Morimae}\ and\ \citenamefont {Yamakawa}(2022)}]{morimae2022}%
  \BibitemOpen
  \bibfield  {author} {\bibinfo {author} {\bibfnamefont {T.}~\bibnamefont {Morimae}}\ and\ \bibinfo {author} {\bibfnamefont {T.}~\bibnamefont {Yamakawa}},\ }\href {https://eprint.iacr.org/2022/1336} {\bibinfo {title} {One-wayness in quantum cryptography}},\ \bibinfo {howpublished} {Cryptology {ePrint} Archive, Paper 2022/1336} (\bibinfo {year} {2022})\BibitemShut {NoStop}%
\bibitem [{\citenamefont {Bostanci}\ \emph {et~al.}(2023)\citenamefont {Bostanci}, \citenamefont {Efron}, \citenamefont {Metger}, \citenamefont {Poremba}, \citenamefont {Qian},\ and\ \citenamefont {Yuen}}]{bostanci2023}%
  \BibitemOpen
  \bibfield  {author} {\bibinfo {author} {\bibfnamefont {J.}~\bibnamefont {Bostanci}}, \bibinfo {author} {\bibfnamefont {Y.}~\bibnamefont {Efron}}, \bibinfo {author} {\bibfnamefont {T.}~\bibnamefont {Metger}}, \bibinfo {author} {\bibfnamefont {A.}~\bibnamefont {Poremba}}, \bibinfo {author} {\bibfnamefont {L.}~\bibnamefont {Qian}},\ and\ \bibinfo {author} {\bibfnamefont {H.}~\bibnamefont {Yuen}},\ }\bibfield  {title} {\bibinfo {title} {Unitary complexity and the uhlmann transformation problem},\ }\href {https://arxiv.org/abs/2306.13073} {\bibfield  {journal} {\bibinfo  {journal} {arXiv preprint arXiv:2306.13073}\ } (\bibinfo {year} {2023})}\BibitemShut {NoStop}%
\bibitem [{\citenamefont {Elben}\ \emph {et~al.}(2023)\citenamefont {Elben}, \citenamefont {Flammia}, \citenamefont {Huang}, \citenamefont {Kueng}, \citenamefont {Preskill}, \citenamefont {Vermersch},\ and\ \citenamefont {Zoller}}]{elben2023}%
  \BibitemOpen
  \bibfield  {author} {\bibinfo {author} {\bibfnamefont {A.}~\bibnamefont {Elben}}, \bibinfo {author} {\bibfnamefont {S.~T.}\ \bibnamefont {Flammia}}, \bibinfo {author} {\bibfnamefont {H.-Y.}\ \bibnamefont {Huang}}, \bibinfo {author} {\bibfnamefont {R.}~\bibnamefont {Kueng}}, \bibinfo {author} {\bibfnamefont {J.}~\bibnamefont {Preskill}}, \bibinfo {author} {\bibfnamefont {B.}~\bibnamefont {Vermersch}},\ and\ \bibinfo {author} {\bibfnamefont {P.}~\bibnamefont {Zoller}},\ }\bibfield  {title} {\bibinfo {title} {The randomized measurement toolbox},\ }\href {https://doi.org/https://doi.org/10.1038/s42254-022-00535-2} {\bibfield  {journal} {\bibinfo  {journal} {Nature Reviews Physics}\ }\textbf {\bibinfo {volume} {5}},\ \bibinfo {pages} {9} (\bibinfo {year} {2023})}\BibitemShut {NoStop}%
\bibitem [{\citenamefont {Movassagh}(2023)}]{movassagh2023}%
  \BibitemOpen
  \bibfield  {author} {\bibinfo {author} {\bibfnamefont {R.}~\bibnamefont {Movassagh}},\ }\bibfield  {title} {\bibinfo {title} {The hardness of random quantum circuits},\ }\href {https://doi.org/https://doi.org/10.1038/s41567-023-02131-2} {\bibfield  {journal} {\bibinfo  {journal} {Nature Physics}\ }\textbf {\bibinfo {volume} {19}},\ \bibinfo {pages} {1719} (\bibinfo {year} {2023})}\BibitemShut {NoStop}%
\bibitem [{\citenamefont {Hashimoto}\ \emph {et~al.}(2017)\citenamefont {Hashimoto}, \citenamefont {Murata},\ and\ \citenamefont {Yoshii}}]{hashimoto2017}%
  \BibitemOpen
  \bibfield  {author} {\bibinfo {author} {\bibfnamefont {K.}~\bibnamefont {Hashimoto}}, \bibinfo {author} {\bibfnamefont {K.}~\bibnamefont {Murata}},\ and\ \bibinfo {author} {\bibfnamefont {R.}~\bibnamefont {Yoshii}},\ }\bibfield  {title} {\bibinfo {title} {Out-of-time-order correlators in quantum mechanics},\ }\href {https://doi.org/https://doi.org/10.1007/JHEP10(2017)138} {\bibfield  {journal} {\bibinfo  {journal} {Journal of High Energy Physics}\ }\textbf {\bibinfo {volume} {2017}},\ \bibinfo {pages} {1} (\bibinfo {year} {2017})}\BibitemShut {NoStop}%
\bibitem [{\citenamefont {Naor}\ and\ \citenamefont {Reingold}(2004)}]{Naor2004}%
  \BibitemOpen
  \bibfield  {author} {\bibinfo {author} {\bibfnamefont {M.}~\bibnamefont {Naor}}\ and\ \bibinfo {author} {\bibfnamefont {O.}~\bibnamefont {Reingold}},\ }\bibfield  {title} {\bibinfo {title} {Number-theoretic constructions of efficient pseudo-random functions},\ }\href {https://doi.org/10.1145/972639.972643} {\bibfield  {journal} {\bibinfo  {journal} {J. ACM}\ }\textbf {\bibinfo {volume} {51}},\ \bibinfo {pages} {231–262} (\bibinfo {year} {2004})}\BibitemShut {NoStop}%
\bibitem [{\citenamefont {Swingle}\ \emph {et~al.}(2016)\citenamefont {Swingle}, \citenamefont {Bentsen}, \citenamefont {Schleier-Smith},\ and\ \citenamefont {Hayden}}]{Brian2016}%
  \BibitemOpen
  \bibfield  {author} {\bibinfo {author} {\bibfnamefont {B.}~\bibnamefont {Swingle}}, \bibinfo {author} {\bibfnamefont {G.}~\bibnamefont {Bentsen}}, \bibinfo {author} {\bibfnamefont {M.}~\bibnamefont {Schleier-Smith}},\ and\ \bibinfo {author} {\bibfnamefont {P.}~\bibnamefont {Hayden}},\ }\bibfield  {title} {\bibinfo {title} {Measuring the scrambling of quantum information},\ }\href {https://doi.org/10.1103/PhysRevA.94.040302} {\bibfield  {journal} {\bibinfo  {journal} {Phys. Rev. A}\ }\textbf {\bibinfo {volume} {94}},\ \bibinfo {pages} {040302} (\bibinfo {year} {2016})}\BibitemShut {NoStop}%
\bibitem [{\citenamefont {Yao}\ \emph {et~al.}(2016)\citenamefont {Yao}, \citenamefont {Grusdt}, \citenamefont {Swingle}, \citenamefont {Lukin}, \citenamefont {Stamper-Kurn}, \citenamefont {Moore},\ and\ \citenamefont {Demler}}]{yao2016}%
  \BibitemOpen
  \bibfield  {author} {\bibinfo {author} {\bibfnamefont {N.~Y.}\ \bibnamefont {Yao}}, \bibinfo {author} {\bibfnamefont {F.}~\bibnamefont {Grusdt}}, \bibinfo {author} {\bibfnamefont {B.}~\bibnamefont {Swingle}}, \bibinfo {author} {\bibfnamefont {M.~D.}\ \bibnamefont {Lukin}}, \bibinfo {author} {\bibfnamefont {D.~M.}\ \bibnamefont {Stamper-Kurn}}, \bibinfo {author} {\bibfnamefont {J.~E.}\ \bibnamefont {Moore}},\ and\ \bibinfo {author} {\bibfnamefont {E.~A.}\ \bibnamefont {Demler}},\ }\href {https://arxiv.org/abs/1607.01801} {\bibinfo {title} {Interferometric approach to probing fast scrambling}} (\bibinfo {year} {2016}),\ \Eprint {https://arxiv.org/abs/1607.01801} {arXiv:1607.01801 [quant-ph]} \BibitemShut {NoStop}%
\bibitem [{\citenamefont {Schiulaz}\ \emph {et~al.}(2019)\citenamefont {Schiulaz}, \citenamefont {Torres-Herrera},\ and\ \citenamefont {Santos}}]{Schiulaz_2019}%
  \BibitemOpen
  \bibfield  {author} {\bibinfo {author} {\bibfnamefont {M.}~\bibnamefont {Schiulaz}}, \bibinfo {author} {\bibfnamefont {E.~J.}\ \bibnamefont {Torres-Herrera}},\ and\ \bibinfo {author} {\bibfnamefont {L.~F.}\ \bibnamefont {Santos}},\ }\bibfield  {title} {\bibinfo {title} {Thouless and relaxation time scales in many-body quantum systems},\ }\bibfield  {journal} {\bibinfo  {journal} {Physical Review B}\ }\textbf {\bibinfo {volume} {99}},\ \href {https://doi.org/10.1103/physrevb.99.174313} {10.1103/physrevb.99.174313} (\bibinfo {year} {2019})\BibitemShut {NoStop}%
\bibitem [{\citenamefont {Cotler}\ \emph {et~al.}(2017)\citenamefont {Cotler}, \citenamefont {Hunter-Jones}, \citenamefont {Liu},\ and\ \citenamefont {Yoshida}}]{Cotler2017Chaos}%
  \BibitemOpen
  \bibfield  {author} {\bibinfo {author} {\bibfnamefont {J.}~\bibnamefont {Cotler}}, \bibinfo {author} {\bibfnamefont {N.}~\bibnamefont {Hunter-Jones}}, \bibinfo {author} {\bibfnamefont {J.}~\bibnamefont {Liu}},\ and\ \bibinfo {author} {\bibfnamefont {B.}~\bibnamefont {Yoshida}},\ }\bibfield  {title} {\bibinfo {title} {Chaos, complexity, and random matrices},\ }\bibfield  {journal} {\bibinfo  {journal} {Journal of High Energy Physics}\ }\textbf {\bibinfo {volume} {2017}},\ \href {https://doi.org/10.1007/jhep11(2017)048} {10.1007/jhep11(2017)048} (\bibinfo {year} {2017})\BibitemShut {NoStop}%
\bibitem [{\citenamefont {Livan}\ \emph {et~al.}(2018)\citenamefont {Livan}, \citenamefont {Novaes},\ and\ \citenamefont {Vivo}}]{livan2018}%
  \BibitemOpen
  \bibfield  {author} {\bibinfo {author} {\bibfnamefont {G.}~\bibnamefont {Livan}}, \bibinfo {author} {\bibfnamefont {M.}~\bibnamefont {Novaes}},\ and\ \bibinfo {author} {\bibfnamefont {P.}~\bibnamefont {Vivo}},\ }\bibfield  {title} {\bibinfo {title} {Introduction to random matrices theory and practice},\ }\href@noop {} {\bibfield  {journal} {\bibinfo  {journal} {Monograph Award}\ }\textbf {\bibinfo {volume} {63}},\ \bibinfo {pages} {914} (\bibinfo {year} {2018})}\BibitemShut {NoStop}%
\bibitem [{\citenamefont {Garc\'\i{}a-Mata}\ \emph {et~al.}(2023)\citenamefont {Garc\'\i{}a-Mata}, \citenamefont {Jalabert},\ and\ \citenamefont {Wisniacki}}]{Garc_a_Mata_2023}%
  \BibitemOpen
  \bibfield  {author} {\bibinfo {author} {\bibfnamefont {I.}~\bibnamefont {Garc\'\i{}a-Mata}}, \bibinfo {author} {\bibfnamefont {R.~A.}\ \bibnamefont {Jalabert}},\ and\ \bibinfo {author} {\bibfnamefont {D.~A.}\ \bibnamefont {Wisniacki}},\ }\bibfield  {title} {\bibinfo {title} {{Out-of-time-order correlators and quantum chaos}},\ }\href {https://doi.org/10.4249/scholarpedia.55237} {\bibfield  {journal} {\bibinfo  {journal} {Scholarpedia}\ }\textbf {\bibinfo {volume} {18}},\ \bibinfo {pages} {55237} (\bibinfo {year} {2023})},\ \Eprint {https://arxiv.org/abs/2209.07965} {arXiv:2209.07965 [quant-ph]} \BibitemShut {NoStop}%
\bibitem [{\citenamefont {Hanada}\ \emph {et~al.}(2024)\citenamefont {Hanada}, \citenamefont {Jevicki}, \citenamefont {Liu}, \citenamefont {Rinaldi},\ and\ \citenamefont {Tezuka}}]{Hanada2024}%
  \BibitemOpen
  \bibfield  {author} {\bibinfo {author} {\bibfnamefont {M.}~\bibnamefont {Hanada}}, \bibinfo {author} {\bibfnamefont {A.}~\bibnamefont {Jevicki}}, \bibinfo {author} {\bibfnamefont {X.}~\bibnamefont {Liu}}, \bibinfo {author} {\bibfnamefont {E.}~\bibnamefont {Rinaldi}},\ and\ \bibinfo {author} {\bibfnamefont {M.}~\bibnamefont {Tezuka}},\ }\bibfield  {title} {\bibinfo {title} {A model of randomly-coupled pauli spins},\ }\bibfield  {journal} {\bibinfo  {journal} {Journal of High Energy Physics}\ }\textbf {\bibinfo {volume} {2024}},\ \href {https://doi.org/10.1007/jhep05(2024)280} {10.1007/jhep05(2024)280} (\bibinfo {year} {2024})\BibitemShut {NoStop}%
\bibitem [{\citenamefont {Baumgratz}\ \emph {et~al.}(2014)\citenamefont {Baumgratz}, \citenamefont {Cramer},\ and\ \citenamefont {Plenio}}]{Baumgratz2014}%
  \BibitemOpen
  \bibfield  {author} {\bibinfo {author} {\bibfnamefont {T.}~\bibnamefont {Baumgratz}}, \bibinfo {author} {\bibfnamefont {M.}~\bibnamefont {Cramer}},\ and\ \bibinfo {author} {\bibfnamefont {M.~B.}\ \bibnamefont {Plenio}},\ }\bibfield  {title} {\bibinfo {title} {Quantifying coherence},\ }\href {https://doi.org/10.1103/PhysRevLett.113.140401} {\bibfield  {journal} {\bibinfo  {journal} {Phys. Rev. Lett.}\ }\textbf {\bibinfo {volume} {113}},\ \bibinfo {pages} {140401} (\bibinfo {year} {2014})}\BibitemShut {NoStop}%
\bibitem [{\citenamefont {Banerjee}\ \emph {et~al.}(2012)\citenamefont {Banerjee}, \citenamefont {Peikert},\ and\ \citenamefont {Rosen}}]{Banerjee2012}%
  \BibitemOpen
  \bibfield  {author} {\bibinfo {author} {\bibfnamefont {A.}~\bibnamefont {Banerjee}}, \bibinfo {author} {\bibfnamefont {C.}~\bibnamefont {Peikert}},\ and\ \bibinfo {author} {\bibfnamefont {A.}~\bibnamefont {Rosen}},\ }\bibfield  {title} {\bibinfo {title} {Pseudorandom functions and lattices},\ }in\ \href {https://doi.org/https://doi.org/10.1007/978-3-642-29011-4_42} {\emph {\bibinfo {booktitle} {Advances in Cryptology -- EUROCRYPT 2012}}},\ \bibinfo {editor} {edited by\ \bibinfo {editor} {\bibfnamefont {D.}~\bibnamefont {Pointcheval}}\ and\ \bibinfo {editor} {\bibfnamefont {T.}~\bibnamefont {Johansson}}}\ (\bibinfo  {publisher} {Springer Berlin Heidelberg},\ \bibinfo {address} {Berlin, Heidelberg},\ \bibinfo {year} {2012})\ pp.\ \bibinfo {pages} {719--737}\BibitemShut {NoStop}%
\bibitem [{\citenamefont {Hosoyamada}\ and\ \citenamefont {Iwata}(2019)}]{Hosoyamada2019}%
  \BibitemOpen
  \bibfield  {author} {\bibinfo {author} {\bibfnamefont {A.}~\bibnamefont {Hosoyamada}}\ and\ \bibinfo {author} {\bibfnamefont {T.}~\bibnamefont {Iwata}},\ }\bibfield  {title} {\bibinfo {title} {4-round luby-rackoff construction is a qprp},\ }in\ \href {https://api.semanticscholar.org/CorpusID:147692931} {\emph {\bibinfo {booktitle} {International Conference on the Theory and Application of Cryptology and Information Security}}}\ (\bibinfo {year} {2019})\BibitemShut {NoStop}%
\bibitem [{\citenamefont {Arunachalam}\ \emph {et~al.}(2021)\citenamefont {Arunachalam}, \citenamefont {Grilo},\ and\ \citenamefont {Sundaram}}]{Arunachalam2021}%
  \BibitemOpen
  \bibfield  {author} {\bibinfo {author} {\bibfnamefont {S.}~\bibnamefont {Arunachalam}}, \bibinfo {author} {\bibfnamefont {A.~B.}\ \bibnamefont {Grilo}},\ and\ \bibinfo {author} {\bibfnamefont {A.}~\bibnamefont {Sundaram}},\ }\bibfield  {title} {\bibinfo {title} {Quantum hardness of learning shallow classical circuits},\ }\href {https://doi.org/10.1137/20M1344202} {\bibfield  {journal} {\bibinfo  {journal} {SIAM Journal on Computing}\ }\textbf {\bibinfo {volume} {50}},\ \bibinfo {pages} {972} (\bibinfo {year} {2021})}\BibitemShut {NoStop}%
\bibitem [{\citenamefont {Zhao}\ \emph {et~al.}(2024)\citenamefont {Zhao}, \citenamefont {Lewis}, \citenamefont {Kannan}, \citenamefont {Quek}, \citenamefont {Huang},\ and\ \citenamefont {Caro}}]{Zhao2024}%
  \BibitemOpen
  \bibfield  {author} {\bibinfo {author} {\bibfnamefont {H.}~\bibnamefont {Zhao}}, \bibinfo {author} {\bibfnamefont {L.}~\bibnamefont {Lewis}}, \bibinfo {author} {\bibfnamefont {I.}~\bibnamefont {Kannan}}, \bibinfo {author} {\bibfnamefont {Y.}~\bibnamefont {Quek}}, \bibinfo {author} {\bibfnamefont {H.-Y.}\ \bibnamefont {Huang}},\ and\ \bibinfo {author} {\bibfnamefont {M.~C.}\ \bibnamefont {Caro}},\ }\bibfield  {title} {\bibinfo {title} {Learning quantum states and unitaries of bounded gate complexity},\ }\href {https://doi.org/10.1103/PRXQuantum.5.040306} {\bibfield  {journal} {\bibinfo  {journal} {PRX Quantum}\ }\textbf {\bibinfo {volume} {5}},\ \bibinfo {pages} {040306} (\bibinfo {year} {2024})}\BibitemShut {NoStop}%
\bibitem [{\citenamefont {Lee}\ \emph {et~al.}(2024)\citenamefont {Lee}, \citenamefont {Kwon},\ and\ \citenamefont {Cho}}]{lee2024fast}%
  \BibitemOpen
  \bibfield  {author} {\bibinfo {author} {\bibfnamefont {W.}~\bibnamefont {Lee}}, \bibinfo {author} {\bibfnamefont {H.}~\bibnamefont {Kwon}},\ and\ \bibinfo {author} {\bibfnamefont {G.~Y.}\ \bibnamefont {Cho}},\ }\href {https://arxiv.org/abs/2411.03974} {\bibinfo {title} {Fast pseudothermalization}} (\bibinfo {year} {2024}),\ \Eprint {https://arxiv.org/abs/2411.03974} {arXiv:2411.03974 [quant-ph]} \BibitemShut {NoStop}%
\bibitem [{\citenamefont {Naor}\ and\ \citenamefont {Reingold}(1999)}]{NAOR1999}%
  \BibitemOpen
  \bibfield  {author} {\bibinfo {author} {\bibfnamefont {M.}~\bibnamefont {Naor}}\ and\ \bibinfo {author} {\bibfnamefont {O.}~\bibnamefont {Reingold}},\ }\bibfield  {title} {\bibinfo {title} {Synthesizers and their application to the parallel construction of pseudo-random functions},\ }\href {https://doi.org/https://doi.org/10.1006/jcss.1998.1618} {\bibfield  {journal} {\bibinfo  {journal} {Journal of Computer and System Sciences}\ }\textbf {\bibinfo {volume} {58}},\ \bibinfo {pages} {336} (\bibinfo {year} {1999})}\BibitemShut {NoStop}%
\bibitem [{\citenamefont {Zhandry}(2015)}]{Zhandry2015}%
  \BibitemOpen
  \bibfield  {author} {\bibinfo {author} {\bibfnamefont {M.}~\bibnamefont {Zhandry}},\ }\bibfield  {title} {\bibinfo {title} {A note on the quantum collision and set equality problems},\ }\href {https://dl.acm.org/doi/abs/10.5555/2871411.2871413} {\bibfield  {journal} {\bibinfo  {journal} {Quantum Info. Comput.}\ }\textbf {\bibinfo {volume} {15}},\ \bibinfo {pages} {557–567} (\bibinfo {year} {2015})}\BibitemShut {NoStop}%
\bibitem [{\citenamefont {Evered}\ \emph {et~al.}(2023)\citenamefont {Evered}, \citenamefont {Bluvstein}, \citenamefont {Kalinowski}, \citenamefont {Ebadi}, \citenamefont {Manovitz}, \citenamefont {Zhou}, \citenamefont {Li}, \citenamefont {Geim}, \citenamefont {Wang}, \citenamefont {Maskara}, \citenamefont {Levine}, \citenamefont {Semeghini}, \citenamefont {Greiner}, \citenamefont {Vuleti{\'c}},\ and\ \citenamefont {Lukin}}]{Evered2023}%
  \BibitemOpen
  \bibfield  {author} {\bibinfo {author} {\bibfnamefont {S.~J.}\ \bibnamefont {Evered}}, \bibinfo {author} {\bibfnamefont {D.}~\bibnamefont {Bluvstein}}, \bibinfo {author} {\bibfnamefont {M.}~\bibnamefont {Kalinowski}}, \bibinfo {author} {\bibfnamefont {S.}~\bibnamefont {Ebadi}}, \bibinfo {author} {\bibfnamefont {T.}~\bibnamefont {Manovitz}}, \bibinfo {author} {\bibfnamefont {H.}~\bibnamefont {Zhou}}, \bibinfo {author} {\bibfnamefont {S.~H.}\ \bibnamefont {Li}}, \bibinfo {author} {\bibfnamefont {A.~A.}\ \bibnamefont {Geim}}, \bibinfo {author} {\bibfnamefont {T.~T.}\ \bibnamefont {Wang}}, \bibinfo {author} {\bibfnamefont {N.}~\bibnamefont {Maskara}}, \bibinfo {author} {\bibfnamefont {H.}~\bibnamefont {Levine}}, \bibinfo {author} {\bibfnamefont {G.}~\bibnamefont {Semeghini}}, \bibinfo {author} {\bibfnamefont {M.}~\bibnamefont {Greiner}}, \bibinfo {author} {\bibfnamefont {V.}~\bibnamefont {Vuleti{\'c}}},\ and\ \bibinfo {author} {\bibfnamefont {M.~D.}\ \bibnamefont {Lukin}},\ }\bibfield  {title} {\bibinfo {title}
  {High-fidelity parallel entangling gates on a neutral-atom quantum computer},\ }\href {https://doi.org/10.1038/s41586-023-06481-y} {\bibfield  {journal} {\bibinfo  {journal} {Nature}\ }\textbf {\bibinfo {volume} {622}},\ \bibinfo {pages} {268} (\bibinfo {year} {2023})}\BibitemShut {NoStop}%
\bibitem [{\citenamefont {Bluvstein}\ \emph {et~al.}(2024)\citenamefont {Bluvstein}, \citenamefont {Evered}, \citenamefont {Geim}, \citenamefont {Li}, \citenamefont {Zhou}, \citenamefont {Manovitz}, \citenamefont {Ebadi}, \citenamefont {Cain}, \citenamefont {Kalinowski}, \citenamefont {Hangleiter}, \citenamefont {Bonilla~Ataides}, \citenamefont {Maskara}, \citenamefont {Cong}, \citenamefont {Gao}, \citenamefont {Sales~Rodriguez}, \citenamefont {Karolyshyn}, \citenamefont {Semeghini}, \citenamefont {Gullans}, \citenamefont {Greiner}, \citenamefont {Vuleti{\'c}},\ and\ \citenamefont {Lukin}}]{Bluvstein2024}%
  \BibitemOpen
  \bibfield  {author} {\bibinfo {author} {\bibfnamefont {D.}~\bibnamefont {Bluvstein}}, \bibinfo {author} {\bibfnamefont {S.~J.}\ \bibnamefont {Evered}}, \bibinfo {author} {\bibfnamefont {A.~A.}\ \bibnamefont {Geim}}, \bibinfo {author} {\bibfnamefont {S.~H.}\ \bibnamefont {Li}}, \bibinfo {author} {\bibfnamefont {H.}~\bibnamefont {Zhou}}, \bibinfo {author} {\bibfnamefont {T.}~\bibnamefont {Manovitz}}, \bibinfo {author} {\bibfnamefont {S.}~\bibnamefont {Ebadi}}, \bibinfo {author} {\bibfnamefont {M.}~\bibnamefont {Cain}}, \bibinfo {author} {\bibfnamefont {M.}~\bibnamefont {Kalinowski}}, \bibinfo {author} {\bibfnamefont {D.}~\bibnamefont {Hangleiter}}, \bibinfo {author} {\bibfnamefont {J.~P.}\ \bibnamefont {Bonilla~Ataides}}, \bibinfo {author} {\bibfnamefont {N.}~\bibnamefont {Maskara}}, \bibinfo {author} {\bibfnamefont {I.}~\bibnamefont {Cong}}, \bibinfo {author} {\bibfnamefont {X.}~\bibnamefont {Gao}}, \bibinfo {author} {\bibfnamefont {P.}~\bibnamefont {Sales~Rodriguez}}, \bibinfo {author} {\bibfnamefont
  {T.}~\bibnamefont {Karolyshyn}}, \bibinfo {author} {\bibfnamefont {G.}~\bibnamefont {Semeghini}}, \bibinfo {author} {\bibfnamefont {M.~J.}\ \bibnamefont {Gullans}}, \bibinfo {author} {\bibfnamefont {M.}~\bibnamefont {Greiner}}, \bibinfo {author} {\bibfnamefont {V.}~\bibnamefont {Vuleti{\'c}}},\ and\ \bibinfo {author} {\bibfnamefont {M.~D.}\ \bibnamefont {Lukin}},\ }\bibfield  {title} {\bibinfo {title} {Logical quantum processor based on reconfigurable atom arrays},\ }\href {https://doi.org/10.1038/s41586-023-06927-3} {\bibfield  {journal} {\bibinfo  {journal} {Nature}\ }\textbf {\bibinfo {volume} {626}},\ \bibinfo {pages} {58} (\bibinfo {year} {2024})}\BibitemShut {NoStop}%
\bibitem [{\citenamefont {Feng}\ and\ \citenamefont {Ippoliti}(2024)}]{feng2024}%
  \BibitemOpen
  \bibfield  {author} {\bibinfo {author} {\bibfnamefont {X.}~\bibnamefont {Feng}}\ and\ \bibinfo {author} {\bibfnamefont {M.}~\bibnamefont {Ippoliti}},\ }\href {https://arxiv.org/abs/2403.09619} {\bibinfo {title} {Dynamics of pseudoentanglement}} (\bibinfo {year} {2024}),\ \Eprint {https://arxiv.org/abs/2403.09619} {arXiv:2403.09619 [quant-ph]} \BibitemShut {NoStop}%
\bibitem [{\citenamefont {Haah}\ \emph {et~al.}(2024)\citenamefont {Haah}, \citenamefont {Liu},\ and\ \citenamefont {Tan}}]{haah2024}%
  \BibitemOpen
  \bibfield  {author} {\bibinfo {author} {\bibfnamefont {J.}~\bibnamefont {Haah}}, \bibinfo {author} {\bibfnamefont {Y.}~\bibnamefont {Liu}},\ and\ \bibinfo {author} {\bibfnamefont {X.}~\bibnamefont {Tan}},\ }\href {https://arxiv.org/abs/2402.05239} {\bibinfo {title} {Efficient approximate unitary designs from random pauli rotations}} (\bibinfo {year} {2024}),\ \Eprint {https://arxiv.org/abs/2402.05239} {arXiv:2402.05239 [quant-ph]} \BibitemShut {NoStop}%
\bibitem [{\citenamefont {Schuster}\ \emph {et~al.}(2024)\citenamefont {Schuster}, \citenamefont {Haferkamp},\ and\ \citenamefont {Huang}}]{schuster2024}%
  \BibitemOpen
  \bibfield  {author} {\bibinfo {author} {\bibfnamefont {T.}~\bibnamefont {Schuster}}, \bibinfo {author} {\bibfnamefont {J.}~\bibnamefont {Haferkamp}},\ and\ \bibinfo {author} {\bibfnamefont {H.-Y.}\ \bibnamefont {Huang}},\ }\href {https://arxiv.org/abs/2407.07754} {\bibinfo {title} {Random unitaries in extremely low depth}} (\bibinfo {year} {2024}),\ \Eprint {https://arxiv.org/abs/2407.07754} {arXiv:2407.07754 [quant-ph]} \BibitemShut {NoStop}%
\bibitem [{\citenamefont {Helsen}\ \emph {et~al.}(2022)\citenamefont {Helsen}, \citenamefont {Roth}, \citenamefont {Onorati}, \citenamefont {Werner},\ and\ \citenamefont {Eisert}}]{Helsen2022}%
  \BibitemOpen
  \bibfield  {author} {\bibinfo {author} {\bibfnamefont {J.}~\bibnamefont {Helsen}}, \bibinfo {author} {\bibfnamefont {I.}~\bibnamefont {Roth}}, \bibinfo {author} {\bibfnamefont {E.}~\bibnamefont {Onorati}}, \bibinfo {author} {\bibfnamefont {A.}~\bibnamefont {Werner}},\ and\ \bibinfo {author} {\bibfnamefont {J.}~\bibnamefont {Eisert}},\ }\bibfield  {title} {\bibinfo {title} {General framework for randomized benchmarking},\ }\href {https://doi.org/10.1103/PRXQuantum.3.020357} {\bibfield  {journal} {\bibinfo  {journal} {PRX Quantum}\ }\textbf {\bibinfo {volume} {3}},\ \bibinfo {pages} {020357} (\bibinfo {year} {2022})}\BibitemShut {NoStop}%
\bibitem [{\citenamefont {Yoshida}\ and\ \citenamefont {Kitaev}(2017)}]{yoshida2017}%
  \BibitemOpen
  \bibfield  {author} {\bibinfo {author} {\bibfnamefont {B.}~\bibnamefont {Yoshida}}\ and\ \bibinfo {author} {\bibfnamefont {A.}~\bibnamefont {Kitaev}},\ }\href {https://arxiv.org/abs/1710.03363} {\bibinfo {title} {Efficient decoding for the hayden-preskill protocol}} (\bibinfo {year} {2017}),\ \Eprint {https://arxiv.org/abs/1710.03363} {arXiv:1710.03363 [hep-th]} \BibitemShut {NoStop}%
\bibitem [{\citenamefont {Yoshida}\ and\ \citenamefont {Yao}(2019)}]{Yoshida2019}%
  \BibitemOpen
  \bibfield  {author} {\bibinfo {author} {\bibfnamefont {B.}~\bibnamefont {Yoshida}}\ and\ \bibinfo {author} {\bibfnamefont {N.~Y.}\ \bibnamefont {Yao}},\ }\bibfield  {title} {\bibinfo {title} {Disentangling scrambling and decoherence via quantum teleportation},\ }\href {https://doi.org/10.1103/PhysRevX.9.011006} {\bibfield  {journal} {\bibinfo  {journal} {Phys. Rev. X}\ }\textbf {\bibinfo {volume} {9}},\ \bibinfo {pages} {011006} (\bibinfo {year} {2019})}\BibitemShut {NoStop}%
\bibitem [{\citenamefont {Bouland}\ \emph {et~al.}(2020)\citenamefont {Bouland}, \citenamefont {Fefferman},\ and\ \citenamefont {Vazirani}}]{bouland2019computational}%
  \BibitemOpen
  \bibfield  {author} {\bibinfo {author} {\bibfnamefont {A.}~\bibnamefont {Bouland}}, \bibinfo {author} {\bibfnamefont {B.}~\bibnamefont {Fefferman}},\ and\ \bibinfo {author} {\bibfnamefont {U.}~\bibnamefont {Vazirani}},\ }\bibfield  {title} {\bibinfo {title} {{Computational Pseudorandomness, the Wormhole Growth Paradox, and Constraints on the AdS/CFT Duality}},\ }in\ \href {https://doi.org/10.4230/LIPIcs.ITCS.2020.63} {\emph {\bibinfo {booktitle} {11th Innovations in Theoretical Computer Science Conference (ITCS 2020)}}},\ \bibinfo {series} {Leibniz International Proceedings in Informatics (LIPIcs)}, Vol.\ \bibinfo {volume} {151},\ \bibinfo {editor} {edited by\ \bibinfo {editor} {\bibfnamefont {T.}~\bibnamefont {Vidick}}}\ (\bibinfo  {publisher} {Schloss Dagstuhl -- Leibniz-Zentrum f{\"u}r Informatik},\ \bibinfo {address} {Dagstuhl, Germany},\ \bibinfo {year} {2020})\ pp.\ \bibinfo {pages} {63:1--63:2}\BibitemShut {NoStop}%
\bibitem [{\citenamefont {Piroli}\ \emph {et~al.}(2020)\citenamefont {Piroli}, \citenamefont {Sünderhauf},\ and\ \citenamefont {Qi}}]{Piroli_2020}%
  \BibitemOpen
  \bibfield  {author} {\bibinfo {author} {\bibfnamefont {L.}~\bibnamefont {Piroli}}, \bibinfo {author} {\bibfnamefont {C.}~\bibnamefont {Sünderhauf}},\ and\ \bibinfo {author} {\bibfnamefont {X.-L.}\ \bibnamefont {Qi}},\ }\bibfield  {title} {\bibinfo {title} {A random unitary circuit model for black hole evaporation},\ }\bibfield  {journal} {\bibinfo  {journal} {Journal of High Energy Physics}\ }\textbf {\bibinfo {volume} {2020}},\ \href {https://doi.org/10.1007/jhep04(2020)063} {10.1007/jhep04(2020)063} (\bibinfo {year} {2020})\BibitemShut {NoStop}%
\bibitem [{\citenamefont {Chamon}\ \emph {et~al.}(2022)\citenamefont {Chamon}, \citenamefont {Mucciolo},\ and\ \citenamefont {Ruckenstein}}]{CHAMON2022169086}%
  \BibitemOpen
  \bibfield  {author} {\bibinfo {author} {\bibfnamefont {C.}~\bibnamefont {Chamon}}, \bibinfo {author} {\bibfnamefont {E.~R.}\ \bibnamefont {Mucciolo}},\ and\ \bibinfo {author} {\bibfnamefont {A.~E.}\ \bibnamefont {Ruckenstein}},\ }\bibfield  {title} {\bibinfo {title} {Quantum statistical mechanics of encryption: Reaching the speed limit of classical block ciphers},\ }\href {https://doi.org/https://doi.org/10.1016/j.aop.2022.169086} {\bibfield  {journal} {\bibinfo  {journal} {Annals of Physics}\ }\textbf {\bibinfo {volume} {446}},\ \bibinfo {pages} {169086} (\bibinfo {year} {2022})}\BibitemShut {NoStop}%
\bibitem [{\citenamefont {Chamon}\ \emph {et~al.}(2024)\citenamefont {Chamon}, \citenamefont {Mucciolo}, \citenamefont {Ruckenstein},\ and\ \citenamefont {Yang}}]{Chamon2024}%
  \BibitemOpen
  \bibfield  {author} {\bibinfo {author} {\bibfnamefont {C.}~\bibnamefont {Chamon}}, \bibinfo {author} {\bibfnamefont {E.~R.}\ \bibnamefont {Mucciolo}}, \bibinfo {author} {\bibfnamefont {A.~E.}\ \bibnamefont {Ruckenstein}},\ and\ \bibinfo {author} {\bibfnamefont {Z.-C.}\ \bibnamefont {Yang}},\ }\bibfield  {title} {\bibinfo {title} {Fast pseudorandom quantum state generators via inflationary quantum gates},\ }\href {https://doi.org/https://doi.org/10.1038/s41534-024-00831-y} {\bibfield  {journal} {\bibinfo  {journal} {npj Quantum Information}\ }\textbf {\bibinfo {volume} {10}},\ \bibinfo {pages} {37} (\bibinfo {year} {2024})}\BibitemShut {NoStop}%
\bibitem [{\citenamefont {Gu}\ \emph {et~al.}(2024{\natexlab{b}})\citenamefont {Gu}, \citenamefont {Quek}, \citenamefont {Yelin}, \citenamefont {Eisert},\ and\ \citenamefont {Leone}}]{gu2024chaos}%
  \BibitemOpen
  \bibfield  {author} {\bibinfo {author} {\bibfnamefont {A.}~\bibnamefont {Gu}}, \bibinfo {author} {\bibfnamefont {Y.}~\bibnamefont {Quek}}, \bibinfo {author} {\bibfnamefont {S.}~\bibnamefont {Yelin}}, \bibinfo {author} {\bibfnamefont {J.}~\bibnamefont {Eisert}},\ and\ \bibinfo {author} {\bibfnamefont {L.}~\bibnamefont {Leone}},\ }\href {https://arxiv.org/abs/2410.18196} {\bibinfo {title} {Simulating quantum chaos without chaos}} (\bibinfo {year} {2024}{\natexlab{b}}),\ \Eprint {https://arxiv.org/abs/2410.18196} {arXiv:2410.18196 [quant-ph]} \BibitemShut {NoStop}%
\bibitem [{\citenamefont {Bremner}\ \emph {et~al.}(2016)\citenamefont {Bremner}, \citenamefont {Montanaro},\ and\ \citenamefont {Shepherd}}]{Bremner2016}%
  \BibitemOpen
  \bibfield  {author} {\bibinfo {author} {\bibfnamefont {M.~J.}\ \bibnamefont {Bremner}}, \bibinfo {author} {\bibfnamefont {A.}~\bibnamefont {Montanaro}},\ and\ \bibinfo {author} {\bibfnamefont {D.~J.}\ \bibnamefont {Shepherd}},\ }\bibfield  {title} {\bibinfo {title} {Average-case complexity versus approximate simulation of commuting quantum computations},\ }\href {https://doi.org/10.1103/PhysRevLett.117.080501} {\bibfield  {journal} {\bibinfo  {journal} {Phys. Rev. Lett.}\ }\textbf {\bibinfo {volume} {117}},\ \bibinfo {pages} {080501} (\bibinfo {year} {2016})}\BibitemShut {NoStop}%
\bibitem [{\citenamefont {Harrow}\ and\ \citenamefont {Mehraban}(2023)}]{Harrow2023}%
  \BibitemOpen
  \bibfield  {author} {\bibinfo {author} {\bibfnamefont {A.~W.}\ \bibnamefont {Harrow}}\ and\ \bibinfo {author} {\bibfnamefont {S.}~\bibnamefont {Mehraban}},\ }\bibfield  {title} {\bibinfo {title} {Approximate unitary t-designs by short random quantum circuits using {Nearest-Neighbor} and {Long-Range} gates},\ }\href {https://doi.org/10.1007/s00220-023-04675-z} {\bibfield  {journal} {\bibinfo  {journal} {Communications in Mathematical Physics}\ }\textbf {\bibinfo {volume} {401}},\ \bibinfo {pages} {1531} (\bibinfo {year} {2023})}\BibitemShut {NoStop}%
\bibitem [{\citenamefont {Lee}\ \emph {et~al.}(2025)\citenamefont {Lee}, \citenamefont {Kwon},\ and\ \citenamefont {Cho}}]{Lee2025pseudodata}%
  \BibitemOpen
  \bibfield  {author} {\bibinfo {author} {\bibfnamefont {W.}~\bibnamefont {Lee}}, \bibinfo {author} {\bibfnamefont {H.}~\bibnamefont {Kwon}},\ and\ \bibinfo {author} {\bibfnamefont {G.~Y.}\ \bibnamefont {Cho}},\ }\href {https://doi.org/10.17632/h7gsjtv27p.1} {\bibinfo {title} {Pseudochaotic many-body dynamics as a pseudorandom state generator}},\ \bibinfo {howpublished} {Mendeley Data, V1, doi: 10.17632/h7gsjtv27p.1} (\bibinfo {year} {2025})\BibitemShut {NoStop}%
\end{thebibliography}

\vspace{1em}  
\noindent {\bf Acknowledgements}\\
We thank Changhun Oh for helpful discussions. W.L. and G.Y.C. are supported by Samsung Science and Technology Foundation under Project Number SSTF-BA2002-05 and SSTF-BA2401-03, the NRF of Korea (Grants No.~RS-2023-00208291, No.~2023M3K5A1094810, No.~2023M3K5A1094813, No.~RS-2024-00410027, No.~RS-2024-00444725) funded by the Korean Government (MSIT), the Air Force Office of Scientific Research under Award No.~FA2386-22-1-4061, No. FA23862514026., and Institute of Basic Science under project code IBS-R014-D1. H.K. is supported by the KIAS Individual Grant No. CG085302 at Korea Institute for Advanced Study and National Research Foundation of Korea (Grants No. RS-2023-NR119931, No. RS-2024-00413957, and No. RS-2024-00438415) funded by the Korean Government (MSIT).

\vspace{1em}  
\noindent {\bf Author Contributions}\\
W.L., H.K., and G.Y.C. conceived and designed the project. W.L. performed theoretical analyses and calculations under the supervision of H.K. and G.Y.C. All authors contributed to discussions and to writing and revising the manuscript.

\vspace{1em}  
\noindent {\bf Competing Interests}\\
The authors declare no competing interests.

\end{document}


\begin{center}
    {\color{gray}\large Supplementary Information for}
\end{center}

\title{Pseudochaotic Many-Body Dynamics as a Pseudorandom State Generator}
\author{Wonjun Lee}
\email{wonjun1998@postech.ac.kr}
\affiliation{Department of Physics, Pohang University of Science and Technology, Pohang, 37673, Republic of Korea}
\affiliation{Center for Artificial Low Dimensional Electronic Systems, Institute for Basic Science, Pohang, 37673, Republic of Korea}

\author{Hyukjoon Kwon}
\email{hjkwon@kias.re.kr}
\affiliation{School of Computational Sciences, Korea Institute for Advanced Study, Seoul 02455, South Korea}

\author{Gil Young Cho}
\email{gilyoungcho@kaist.ac.kr}
\affiliation{Department of Physics, Korea Advanced Institute of Science and Technology, Daejeon 34141, Korea}
\affiliation{Center for Artificial Low Dimensional Electronic Systems, Institute for Basic Science, Pohang, 37673, Republic of Korea}
\affiliation{Asia-Pacific Center for Theoretical Physics, Pohang, 37673, Republic of Korea}

\maketitle
\tableofcontents

\newpage
\clearpage
\setcounter{page}{1}
\section*{Supplementary Figures}
\hypertarget{fig:Hada}{}
\subsection*{Supplementary Figure 1: Time dependence of OTOCs with embedded Hadamard gate}
\begin{figure}[h!]
    \includegraphics[width=16cm]{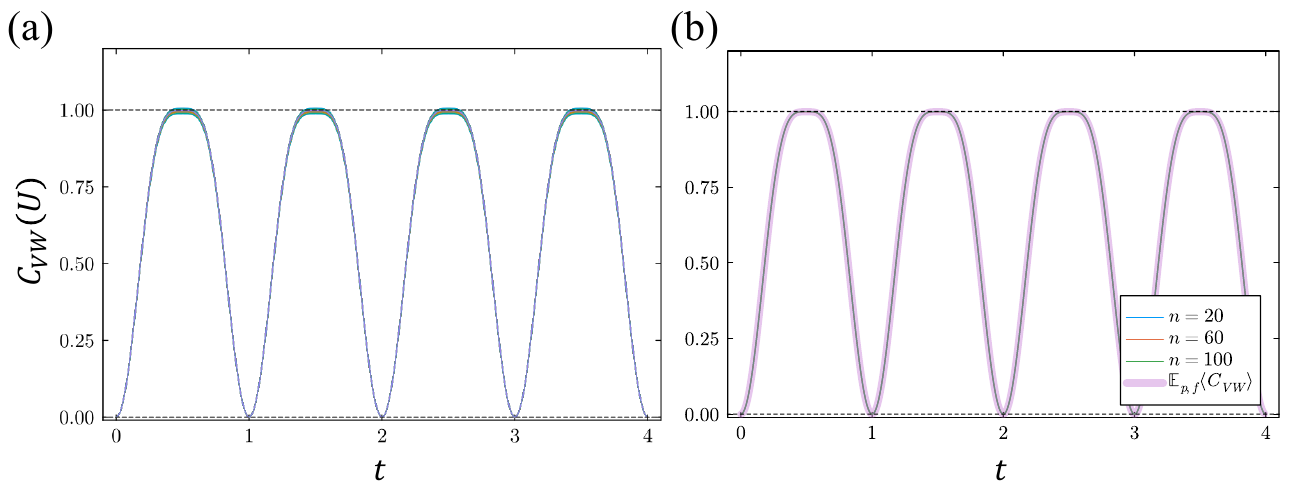}
    \begin{justify}
        Time dependence of Poisson bracket OTOCs of $u^\tau = H^\tau$ with the Hadamard gate $H$ and local operators $V=Z_i$ and $W=Z_j$ with $i\neq j$. (a) Poisson bracket OTOCs of $10^2$ independent realizations of RSED. (b) Poisson bracket OTOCs averaged over random subset isometries and Hamiltonian for various system sizes $n$. $\mathbb{E}_{p,f}\langle C_{VW}\rangle$ in the captions means OTOCs are evaluated using the analytic formula in Eq.~\eqref{eq:ZZ-OTOC-AVGRSP}. The subspace size is set to be $k=10$.
    \end{justify}
\end{figure}

\newpage
\hypertarget{fig:HadaP-stat}{}
\subsection*{Supplementary Figure 2: Level statistics of the RSED}
\begin{figure}[h]
    \centering
    \begin{subfigure}[]
        \centering
        \includegraphics[width=0.4\linewidth]{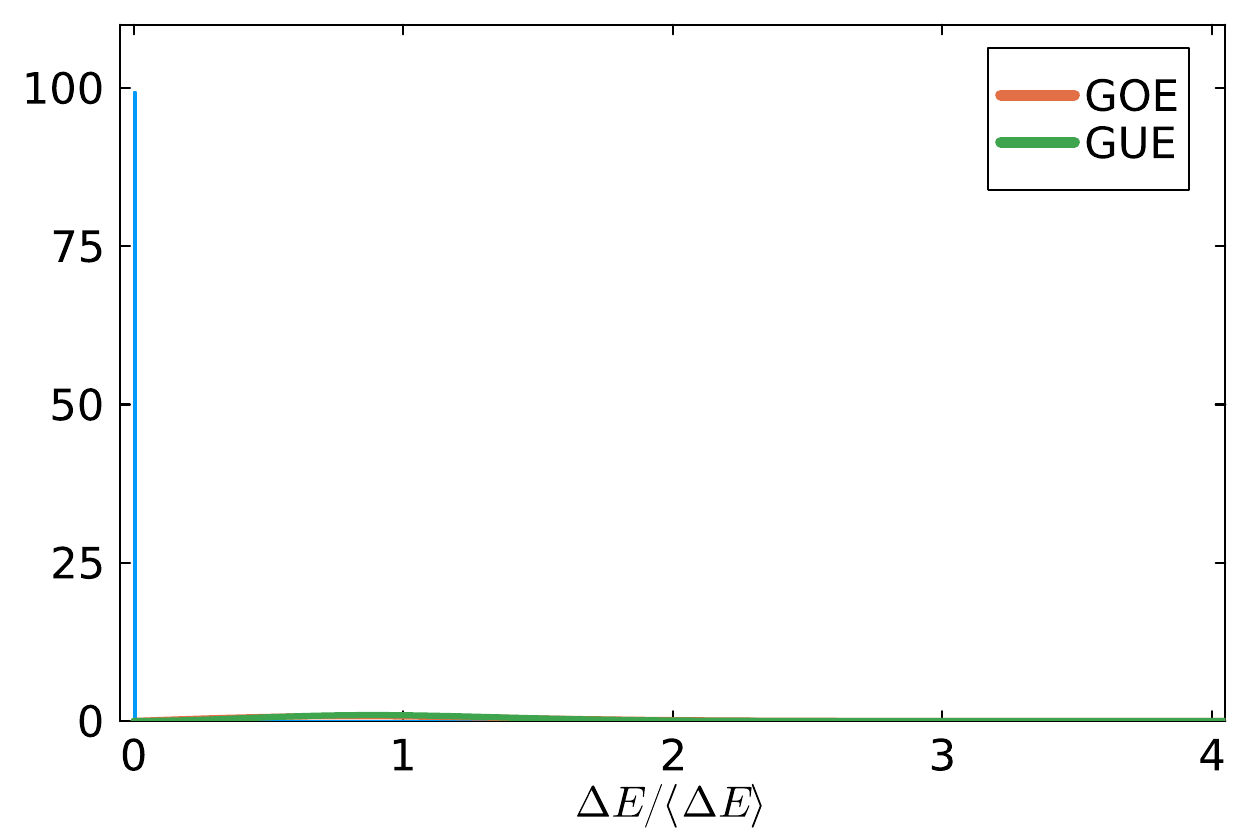}
    \end{subfigure}
    \begin{subfigure}[]
        \centering
        \includegraphics[width=0.4\linewidth]{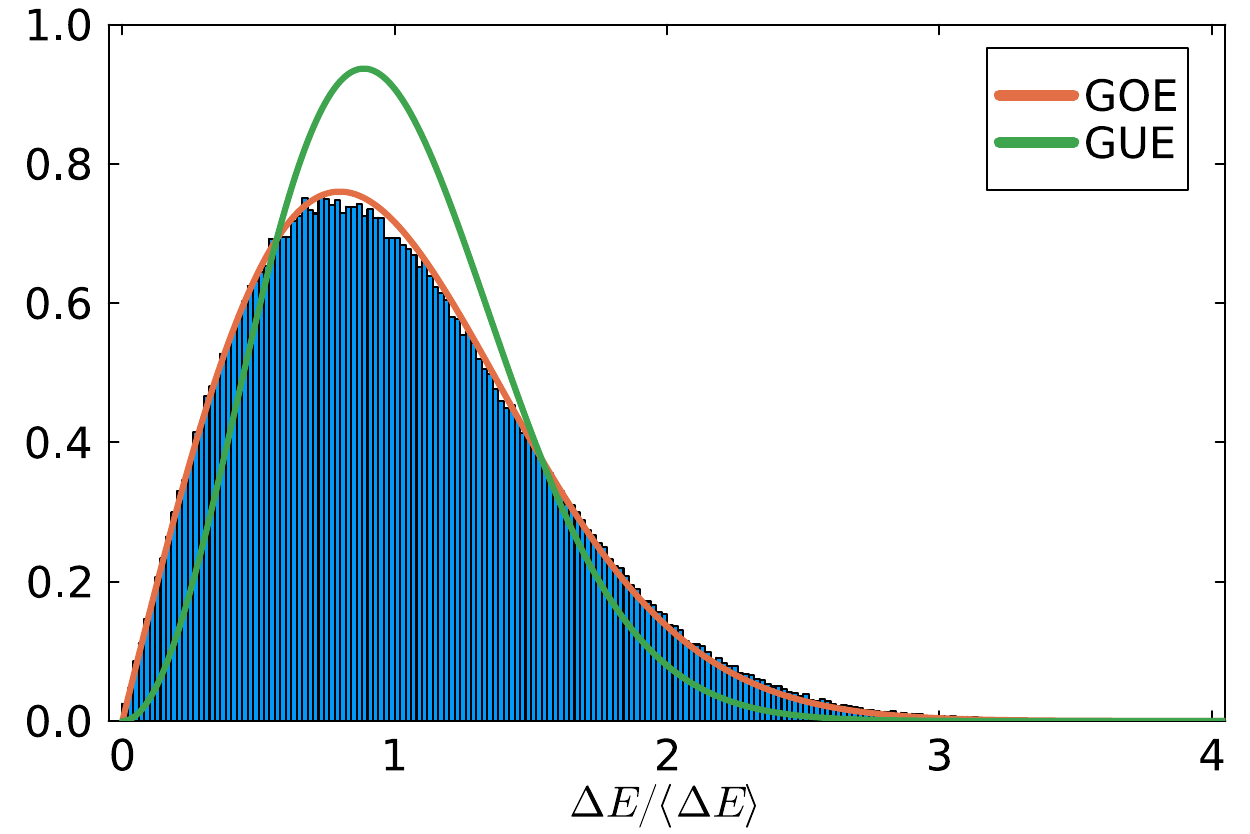}
    \end{subfigure}
    \begin{justify}
        (a) The level statistics of the RSED with the embedded random phase Hadamard gate with the Wigner-Dyson distributions for the Gaussian orthogonal ensemble (GOE) and the Gaussian unitary ensemble (GUE). The system and subsystem sizes are taken as $n=13$ and $k=10$, respectively. (b) The right figure statistics excludes the peak at the zero due to the exponential energy degeneracies of the RSED to show the level statistics of the embedded Hamiltonian separately and explicitly.
    \end{justify}
\end{figure}

\newpage
\hypertarget{fig:HadaP}{}
\subsection*{Supplementary Figure 3: Time dependence of OTOCs with embedded random phase Hadamard gate}
\begin{figure}[ht]
    \includegraphics[width=8cm]{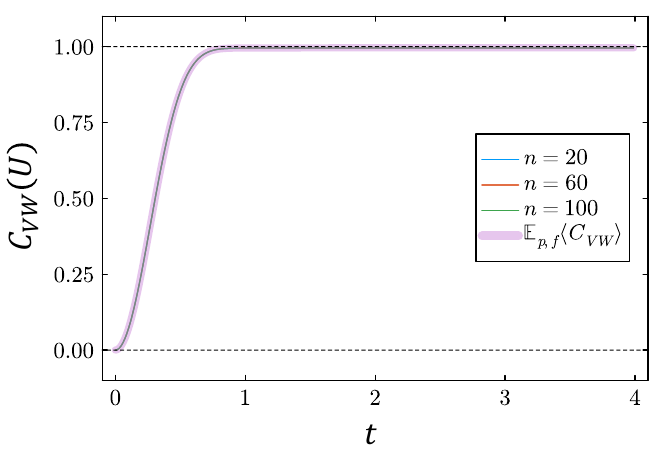}
    \centering
    \begin{justify}
        Time dependence of Poisson bracket OTOCs of $u^\tau = (H^{\otimes k}P)^\tau$ with the random sign gate $P$, the Hadamard gate $H$, and local operators $V=Z_i$ and $W=Z_j$ with $i\neq j$ averaged over random subset isometries and Hamiltonian for various system sizes $n$. $\mathbb{E}_{p,f}\langle C_{VW}\rangle$ in the captions means OTOCs are evaluated using the analytic formula in Eq.~\eqref{eq:ZZ-OTOC-AVGRSP}. In both cases, the subspace size is set to be $k=10$.
    \end{justify}
\end{figure}

\newpage
\hypertarget{fig:Pauli-SYK}{}
\subsection*{Supplementary Figure 4: Time dependence of OTOCs with embedded Pauli-SYK model}
\begin{figure}[ht]
    \includegraphics[width=16cm]{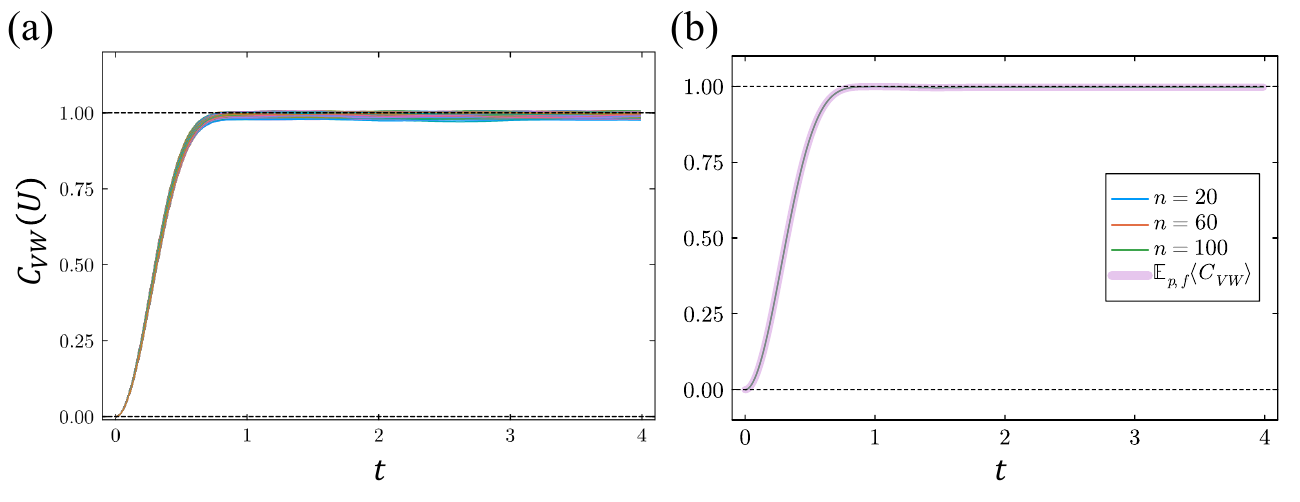}
    \centering
    \begin{justify}
        Time dependence of Poisson bracket OTOCs of the Pauli SYK model with local operators $V=Z_i$ and $W=Z_j$ with $i\neq j$. (a) OTOCs of $10^3$ independent realizations of RSED and Hamiltonian. (b) OTOCs averaged over random subset isometries and Hamiltonian for various system sizes $n$. $\mathbb{E}_{p,f}\langle C_{VW}\rangle$ in the captions means OTOCs are evaluated using the analytic formula in Eq.~\eqref{eq:ZZ-OTOC-AVGRSP}. The subspace size is set to be $k=10$.
    \end{justify}
\end{figure}

\newpage
\hypertarget{fig:HadaS-finite-temp}{}
\subsection*{Supplementary Figure 5: Time dependence of finite temperature OTOCs}
\begin{figure}[ht]
    \centering
    \includegraphics[width=0.8\linewidth]{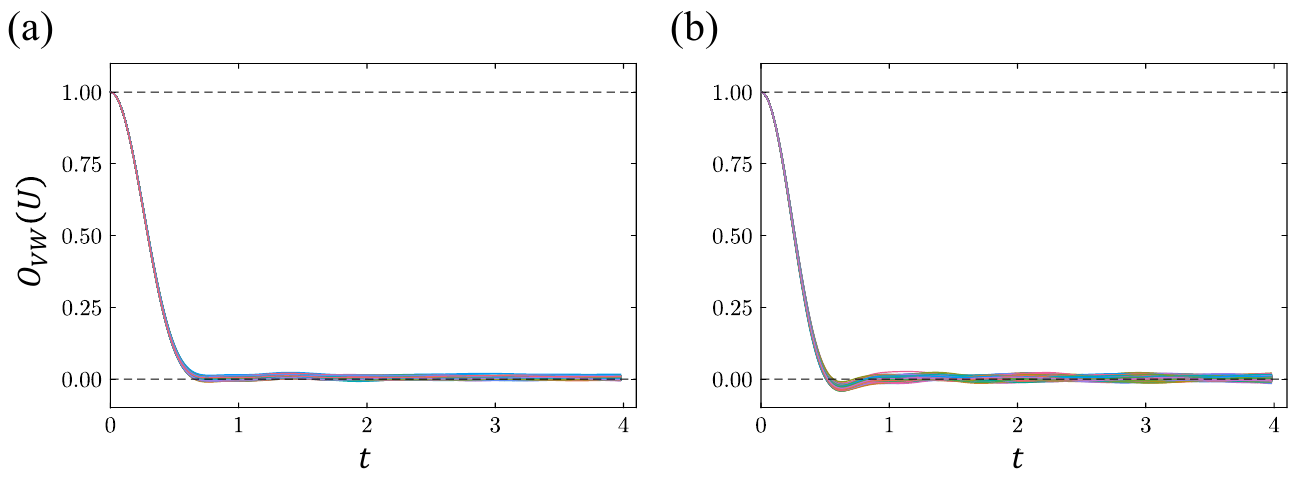}
    \begin{justify}
        Numerics of finite temperature OTOCs of random phase Hadamard gates with $V=Z_i$, $W=Z_j$, and (a) $\beta=1$ (b) $\beta=10$. Each line corresponds to $O_{VW}$ of independent realizations of RSEDs.
    \end{justify}
\end{figure}

\newpage
\section*{Supplementary Notes}
\subsection*{1. Random subsystem Embedded Dynamics}
Let us consider a $n$-qubit system $\mathcal{H}$ and its a random subspace $S_{p,a}=\{p(b,a)|b\in\{0,1\}^k\}$ constructed from an invertible random permutation $p:\{0,1\}^k\times\{0,1\}^{n-k}\rightarrow\{0,1\}^n$ for some constant natural number $k\in\mathbb{N}$ less than $n$. In the main text, we consider the case of $k=\omega (\log n)$ for the realization of pseudochaotic dynamics. However, in principle $k$ can take any positive integer. Here, the subscript `$a$' of $S_{p,a}$ is in $\{0,1\}^{n-k}$ and plays the role of a seed of the random permutation. In addition, we introduce a random function $f(b,a):\{0,1\}^k\times\{0,1\}^{n-k}\rightarrow\{0,1\}$. This introduces extra phase (or sign) randomness on the dynamics. We note that in practice, the random permutation and function are realized by a pseudorandom number generator which is indistinguishable from the true random with limited computational resource. The subspace can be accessed from $\mathcal{H}$ by the isometry $O_a$ defined as
\begin{equation}
    O_a \equiv \sum_{b \in\{0,1\}^k}(-1)^{f(b,a)}\ketbra{p(b,a)}{b,a},
\end{equation}
We will say the subspace $S_{p,a}$ is embedded in the entire space $\mathcal{H}$ by the isometry $O_a$.

Now, let us consider subspace dynamics embedded in $\mathcal{H}$ by $O_a$. This can be characterized by a $k$-qubit hamiltonian $h_a$ represented by a $2^k\times 2^k$ matrix. Such subsystem Hamiltonian can be embedded into $\mathcal{H}$ using $O_a$
\begin{equation}
    H_a=O_a h_a O_a^{\dagger},
\end{equation}
where $H_a$ is a $2^n\times 2^n$ matrix. We call dynamics generated by such a Hamiltonian as a random subspace embedded dynamics (RSED). Its continuous time evolution operator $U_a(t)$ is given by
\begin{equation}
    U_a(t) = e^{-i H_a t} = O e^{-i h_a t} O^{\dagger}+(1-\Pi_a)
\end{equation}
with the projector $\Pi_a=O_a O_a^\dagger$. This can be generalized to embed multiple Hamiltonians $\{h_{a_i}\}$ into $\mathcal{H}$ using isometries $\{O_{a_i}\}$ with a set of seeds $\{a_i\}\subseteq \{0,1\}^{n-k}$. We note that each seed generates non-overlapping subspace $\{S_{p,a_i}\}$. Let $\chi$ be the total number of such subspaces, then the total Hamiltonian for the entire Hilbert space $\mathcal{H}$ is given by
\begin{equation}
    H_{\{a_i\}} = \sum_i^\chi O_{a_i} h_i O_{a_i}^{\dagger}.
\end{equation}
The corresponding unitary evolution $U_{\{a_i\}}(t)$ is given by
\begin{equation}
    U_{\{b_i\}}(t)=\sum_{i=1}^\chi O_{a_i} e^{-i h_i t} O_{a_i}^{\dagger}+\left(1-\sum_{i=1}^\chi \Pi_{a_i}\right), 
\end{equation}
because of the orthogonlity $O_{a_i} O_{a_j}^{\dagger} = \delta_{a_i, a_j}$. In this work, we focus on the case of the maximal embedding of $\chi=2^{n-k}$ with a constant subspace Hamiltonian $h_a=h$.

The trajectory generated by dynamics of each $h_a$ can be densely embedded into $\mathcal{H}$ so that its dynamical properties, specifically regarding its non-periodicity, are preserved in $\mathcal{H}$, if the dimension of the subspace, $2^k$, is large enough. To clearly see this, let us suppose that $h_a$ is a chaotic Hamiltonian in $S_{p,a}$. It is obvious that $H_a$ cannot be chaotic as the space spanned by its time evolution unitary operators is exponentially smaller than the set of unitary operators in $\mathcal{H}$. At the same time, one can imagine a regime that any polynomially many copies are insufficient to reveal whether it is chaotic or not for sufficiently large $k$. If $H_a$ is in such a regime, then we call its dynamics is pseudochaotic. In the following sections, we study when this regime emerges. 

{
When all $\{u_a\}$ are given by the same unitary $u$, independent of $a$ --- which is precisely the case in our RSED construction in the main text --- then the collection $\{O_a u O_a^\dag\}$ can be executed simultaneously by $PF(u\otimes I^{\otimes(n-k)})FP^\dag$ with pseudorandom function $F$ and permutation $P$. This implementation avoids the need for any control gates conditioned on bits outside the subsystem. This can be understood by expanding $I^{\otimes(n-k)}$ into $\sum_{a\in\{0,1\}^{n-k}}\ketbra{a}$:
\begin{equation}
    \begin{split}
        PF(u\otimes I^{\otimes(n-k)})FP^\dag 
        &= \left(\sum_{\substack{b_1\in\{0,1\}^k\\a_1\in\{0,1\}^{n-k}}}(-1)^{f(b_1,a_1)}\ketbra{p(b_1,a_1)}{b_1,a_1}\right)\left(u\otimes\sum_{a\in\{0,1\}^{n-k}}\ketbra{a}{a}\right)\\
        &\quad\times \left(\sum_{\substack{b_2\in\{0,1\}^k\\a_2\in\{0,1\}^{n-k}}}(-1)^{f(b_2,a_2)}\ketbra{b_2,a_2}{p(b_2,a_2)}\right)\\
        &= \sum_{a\in\{0,1\}^{n-k}}\sum_{b_1,b_2\in\{0,1\}^k}(-1)^{f(b_1,a)+f(b_2,a)}\ketbra{p(b_1,a)}{b_1,a}u\ketbra{b_2,a}{p(b_2,a)}\\
        &= \sum_{a\in\{0,1\}^{n-k}}O_a u O_a^\dag.
    \end{split}
\end{equation}
This is illustrated in Figure 4 of the main text.
}

Before finishing this section, we introduce two functions that are useful to calculate dynamical properties of $H_a$. For a given permutation $p$, these two functions are $x_j:\{0,1\}^k\times\{0,1\}^{n-k}\rightarrow\{0,1\}^k$ and $y_j:\{0,1\}^k\times\{0,1\}^{n-k}\rightarrow\{0,1\}^{n-k}$ for some position $j\in[1,n]$ such that $p(x_j(b,a),y_j(b,a))=p(b,a)+\hat{e}_j$ for all $b\in\{0,1\}^k$ and $a\in\{0,1\}^{n-k}$ with the unit vector $\hat{e}_j$ at the site $j$. In other words, $x_j(b,a)$ and $y_j(b,a)$ are the state and seed, respectively, of the permutation that gives $p(b,a)$ with a bit-flip at $j$. They will be extensively used for computing OTOCs in later sections.

\newpage
\subsection*{2. OTOCs with $V^2=I$ and $W^2=I$}
The out-of-time ordered correlator for operators $V$ and $W$ is defined as 
\begin{equation}
    C_{VW}(U) = \frac{1}{2}\left\langle [UWU^\dagger,V]^\dagger[UWU^\dagger,V] \right\rangle_\beta,
\end{equation}
where $U$ is a time evolution operator. For Pauli like operators with $V^2=W^2=I$, it becomes
\begin{equation}
    C_{VW}(U) = 1 - \real \left\langle V U W U^\dagger V U W U^\dagger \right\rangle_\beta.
\end{equation}
The expectation value $\real \left\langle V U W U^\dagger V U W U^\dagger \right\rangle_\beta$ can be estimated, for example, by interferometric measurements~\cite{Brian2016,yao2016}. For pseudochaotic dynamics, it has a negligible value. Thus, distinguishing pseudochaotic dynamics from chaotic dynamics requires $\omega(\operatorname{poly}(n))$ copies. For later use, we define this four-point correlator at the infinite temperature as
\begin{equation}\label{eq:inf-otoc}
    O_{VW}(U) \equiv \frac{1}{2^n}\operatorname{tr}\left( V U W U^\dagger V U W U^\dagger \right).
\end{equation}
In what follows, we explicitly compute the ensemble average of $\abs{O_{VW}}^2$ over random isometries with $V$ and $W$ being on-site Pauli-like traceless matrices. We find that it is negligible in $n$ if $k$ is set by $\omega(\log n)$. This implies that the magnitude of $O_{VW}$ with each isometry is also negligible in $n$. Thus, any polynomial number of measurements cannot determine whether $O_{VW}$ is vanishing (chaotic) or not.

\newpage
\hypertarget{sec:local-OTOC}{}
\subsection*{3. OTOCs with traceless on-site operators}
\begin{theorem}\label{thm:negl-otoc-supp}
    OTOCs with local operators are negligible with probabilities higher than $1-\operatorname{negl(n)}$ in the system size $n$ for RSED with an embedded unitary operator of the dimension $2^k$ with $k=\omega(\log n)$ (sampled from an ensemble) if the maximum (averaged) magnitude of elements of the embedded operator is $O(2^{-k/2})$.
\end{theorem}
\begin{proof}

Let $V$ and $W$ be random traceless hermitian on-site operators positioned at different sites $i$ and $j$, respectively, {and satisfying $V^2=W^2=I$}. 
If the ensemble average of $\abs{O_{VW}(U)}^2$ {defined in Eq.~\eqref{eq:inf-otoc}} is $O(2^{-k})$, then Markov's inequality gives
\begin{equation}
    \operatorname{Pr}\left[\abs{O_{VW}(U)}^2\geq 2^{-(1-\epsilon)k}\right] \leq O(2^{-\epsilon k})
\end{equation}
for some constant $\epsilon>0$, meaning that any OTOC is negligible with the failure probability negligible in $n$. Then, let us compute the second moments of $O_{VW}(U)$ for random on-site traceless unitary operators $V$ and $W$.  

\begin{equation}
    \begin{split}
    \abs{O_{VW}(U)}^2
    &= \frac{1}{4^n} \sum_{\left\{b_i\right\}_{i=1}^{8}} \sum_{\left\{a_i\right\}_{i=1}^4}\sum_{\left\{b'_i\right\}_{i=1}^{8}} \sum_{\left\{a'_i\right\}_{i=1}^4}(-1)^{\sum_{i=1}^{8} \left[ f\left(b_i ,a_{\lfloor(i+1)/2\rfloor}\right)+f\left(b'_i ,a'_{\lfloor(i+1)/2\rfloor}\right)\right]} \\
    &\qquad\times [V]_{p\left(b_8 , a_4\right), p\left(b_1 , a_1\right)}[V]_{p\left(b_4 , a_2\right), p\left(b_5 , a_3\right)}\left[V^*\right]_{p\left(b'_8 , a'_4\right), p\left(b'_1 , a'_1\right)}\left[V^*\right]_{p\left(b'_4 , a'_2\right), p\left(b'_5 , a'_3\right)}\\
    &\qquad\times [W]_{p\left(b_2 , a_1\right), p\left(b_3 , a_2\right)}[W]_{p\left(b_6 , a_3\right), p\left(b_7 , a_4\right)}\left[W^*\right]_{p\left(b'_2 , a'_1\right), p\left(b'_3 , a'_2\right)}\left[W^*\right]_{p\left(b'_6 , a'_3\right), p\left(b'_7 , a'_4\right)} \\
    &\qquad\times \left[U\right]_{b_1, b_2}\left[U^*\right]_{b_4,b_3} \left[U\right]_{b_5, b_6}\left[U^*\right]_{b_8,b_7}\left[U^*\right]_{b'_1, b'_2}\left[U\right]_{b'_4,b'_3} \left[U^*\right]_{b'_5, b'_6}\left[U\right]_{b'_8,b'_7}\\
    \end{split}
\end{equation}
Now, let us consider the following local operators
\begin{equation}
    V = X_i \cos \alpha + Z_i \sin \alpha
    \quad\text{and}\quad
    W = X_j \cos \beta + Z_j \sin \beta
\end{equation}
with $\alpha,\beta\in[0,2\pi)$. While trivial to do so, we do not include the Pauli $Y$ operator in calculation, as its role is the same as the Pauli $X$ operator here. Let us consider that we have expanded Eq.~\eqref{eq:local-OTOC-square} by replacing $V$ and $W$ in terms of the Pauli $X$ or $Z$ operators. Then, $\abs{O_{VW}(U)}^2$ is given by the sum of expectation values of length-eight Pauli strings comprised of $X$ and $Z$. Each element of a Pauli operator is non-vanishing when $c$ is set by $0$ or $1$ for $Z$ and $X$ operators. Then, we have
\begin{equation}
    \begin{split}
    \abs{O_{VW}(U)}^2
    &= \frac{1}{4^n} \sum_{\left\{b_i\right\}_{i=1}^{8}} \sum_{\left\{a_i\right\}_{i=1}^4}\sum_{\left\{b'_i\right\}_{i=1}^{8}} \sum_{\left\{a'_i\right\}_{i=1}^4}\sum_{\substack{c_1,c_2,c_4,c_6\\c'_1,c'_2,c'_4,c'_6}} \Phi(\{a\},\{b\}) \Phi(\{a'\},\{b'\}) \\
    &\qquad\times [V]_{[p\left(b_1 , a_1\right)]_i+c_1, [p\left(b_1 , a_1\right)]_i}[V]_{[p\left(b_4 , a_2\right)]_i, [p\left(b_4 , a_2\right)]_i+c_4}\\
    &\qquad\times\left[V^*\right]_{[p\left(b'_1 , a'_1\right)]_i+c'_1, [p\left(b'_1 , a'_1\right)]_i}\left[V^*\right]_{[p\left(b'_4 , a'_2\right)]_i, [p\left(b'_4 , a'_2\right)]_i+c'_4}\\
    &\qquad\times [W]_{[p\left(b_2 , a_1\right)]_j, [p\left(b_2 , a_1\right)]_j+c_2}[W]_{[p\left(b_6 , a_3\right)]_j, [p\left(b_6 , a_3\right)]_j+c_6}\\
    &\qquad\times \left[W^*\right]_{[p\left(b'_2 , a'_1\right)]_j, [p\left(b'_2 , a'_1\right)]_j+c'_2}\left[W^*\right]_{[p\left(b'_6, a'_3\right)]_j, [p\left(b'_6, a'_3\right)]_j+c'_6} \\
    &\qquad\times \left[U\right]_{b_1, b_2} \left[U\right]_{b_5, b_6}\left[U\right]_{b'_4,b'_3} \left[U\right]_{b'_8,b'_7}\\
    &\qquad\times \left[U^*\right]_{b'_1, b'_2}\left[U^*\right]_{b'_5, b'_6}\left[U^*\right]_{b_4,b_3}\left[U^*\right]_{b_8,b_7}\\
    &\qquad\times\delta_{p(b_8,a_4),p(b_1,a_1)+c_1\hat{e}_i}\delta_{p(b_5,a_3),p(b_4,a_2)+c_4\hat{e}_i}\delta_{p(b_3,a_2),p(b_2,a_1)+c_2\hat{e}_j}\delta_{p(b_7,a_4),p(b_6,a_3)+c_6\hat{e}_j}\\
    &\qquad\times\delta_{p(b'_8,a'_4),p(b'_1,a'_1)+c'_1\hat{e}_i}\delta_{p(b'_5,a'_3),p(b'_4,a'_2)+c'_4\hat{e}_i}\delta_{p(b'_3,a'_2),p(b'_2,a'_1)+c'_2\hat{e}_j}\delta_{p(b'_7,a'_4),p(b'_6,a'_3)+c'_6\hat{e}_j}\\
    \end{split}
\end{equation}
with
\begin{equation}
    \Phi(\{a\},\{b\})
    = (-1)^{f(b_1,a_1)+f(b_2,a_1)+f(b_8,a_4)+f(b_3,a_2)+f(b_4,a_2)+f(b_6,a_3)+f(b_5,a_3)+f(b_7,a_4)}\\
\end{equation}
For given $\{c_1,c_2,c_4,c_6\}$, the product of delta functions,
\begin{equation}
    \delta(\{a\},\{b\},\{c\})=\delta_{p(b_8,a_4),p(b_1,a_1)+c_1\hat{e}_i}\delta_{p(b_3,a_2),p(b_2,a_1)+c_2\hat{e}_j}\delta_{p(b_5,a_3),p(b_4,a_2)+c_4\hat{e}_i}\delta_{p(b_7,a_4),p(b_6,a_3)+c_6\hat{e}_j},
\end{equation}
is non-vanishing for unique values of $(b_3,a_2)$, $(b_5,a_3)$, $(b_7,a_4)$, and $b_8$ if $a_4$ obtained from $\delta_{p(b_8,a_4),p(b_1,a_1)+c_1\hat{e}_i}$,
\begin{equation}
    a_4 = y_i^{c_1}(b_1,a_1) = 
    \begin{cases}
        a_1 & c_1=0\\
        y_i(b_1,a_1) & c_1=1
    \end{cases}
\end{equation}
agrees with that obtained from $\delta_{p(b_7,a_4),p(b_6,a_3)+c_6\hat{e}_j}$,
\begin{equation}
    a_4 = y_j^{c_6}(b_6,a_3) =
    \begin{cases}
        a_3 & c_6=0\\
        y_j(b_6,a_3) & c_6=1
    \end{cases},
\end{equation}
with 
\begin{equation}
    a_2 = y_j^{c_2}(b_2,a_1) = 
    \begin{cases}
        a_1 & c_2=0\\
        y_j(b_2,a_1) & c_2=1
    \end{cases}
    \qquad\text{and}\qquad
    a_3 = y_i^{c_4}(b_4,a_2) = 
    \begin{cases}
        a_2 & c_4=0\\
        y_i(b_4,a_2) & c_4=1
    \end{cases}.
\end{equation}
In other words, values of indices $\{b_1,b_2,b_4,b_6,b_8,a_1\}$ satisfying the following equation can only give non-vanishing contribution to $\abs{O_{VW}(U)}^2$:
\begin{equation}
    p(b_1,a_1) = p(b_8,y_j^{c_5}(b_7,y_i^{c_4}(b_4,y_j^{c_2}(b_2,a_1)))) + c_1\hat{e}_i.
\end{equation}
This equation is simplified when $c_1=c_2=c_4=c_6=0$, which gives $a_1=a_2=a_3=a_4$, $b_8=b_1$, and no constraints on $\{b_1,b_2,b_4,b_6\}$. However, in general, for given $\{b_1,b_2,b_4,b_6,a_1\}$, there may not exist $b_8$ satisfying this. Let this constraint to be $\delta(\{a\},\{b\},\{c\})$, and let indices $\{b_i\}$ and $\{a_i\}$ satisfying this constraint to be $b_i(\{c\})$ and $a_i(\{c\})$, respectively, then we have
\begin{equation}\label{eq:local-OTOC-square}
    \begin{split}
    \abs{O_{VW}(U)}^2 &= \frac{1}{4^n} \sum_{\substack{b_1,b_2,b_4,b_6\\b'_1,b'_2,b'_4,b'_6}} \sum_{a,a'} \sum_{\substack{c_1,c_2,c_4,c_6\\c'_1,c'_2,c'_4,c'_6}} \Phi(\{a(\{c\})\},\{b(\{c\})\}) \Phi(\{a'(\{c'\})\},\{b'(\{c'\})\}) \\
    &\qquad\qquad\qquad\times [V]_{[p\left(b_1 , a\right)]_i+c_1, [p\left(b_1 , a\right)]_i}[V]_{[p\left(b_4 , a_2(\{c\})\right)]_i, [p\left(b_4 , a_2(\{c\})\right)]_i+c_4}\\
    &\qquad\qquad\qquad\times\left[V^*\right]_{[p\left(b'_1 , a'\right)]_i+c'_1, [p\left(b'_1 , a'\right)]_i}\left[V^*\right]_{[p\left(b'_4 , a'_2(\{c'\})\right)]_i, [p\left(b'_4 , a'_2(\{c'\})\right)]_i+c'_4}\\
    &\qquad\qquad\qquad\times [W]_{[p\left(b_2 , a\right)]_j, [p\left(b_2 , a\right)]_j+c_2}[W]_{[p\left(b_6 , a_3(\{c\})\right)]_j, [p\left(b_6 , a_3(\{c\})\right)]_j+c_6}\\
    &\qquad\qquad\qquad\times \left[W^*\right]_{[p\left(b'_2 , a'\right)]_j, [p\left(b'_2 , a'\right)]_j+c'_2}\left[W^*\right]_{[p\left(b'_6, a'_3(\{c'\})\right)]_j, [p\left(b'_6, a'_3(\{c'\})\right)]_j+c'_6} \\
    &\qquad\qquad\qquad\times \left[U\right]_{b_1, b_2} \left[U\right]_{b_5(\{c\}), b_6}\left[U\right]_{b'_4,b'_3(\{c'\})} \left[U\right]_{b'_8(\{c'\}),b'_7(\{c'\})}\\
    &\qquad\qquad\qquad\times \left[U^*\right]_{b'_1, b'_2}\left[U^*\right]_{b'_5(\{c'\}), b'_6}\left[U^*\right]_{b_4,b_3(\{c\})}\left[U^*\right]_{b_8(\{c\}),b_7(\{c\})}\\
    &\qquad\qquad\qquad\times \delta(\{a\},\{b\},\{c\})\delta(\{a'\},\{b'\},\{c'\}).
    \end{split}
\end{equation}
If all elements of $U$ have magnitudes of $O(2^{-k/2})$ with high probability, then we have
\begin{equation}\label{eq:O_VW_square_upper_bound}
    \begin{split}
    \abs{O_{VW}(U)}^2
    &\leq \frac{A}{2^{2n+4k}} \sum_{\substack{b_1,b_2,b_4,b_6\\b'_1,b'_2,b'_4,b'_6}} \sum_{a,a'} \sum_{\substack{c_1,c_2,c_4,c_6\\c'_1,c'_2,c'_4,c'_6}}\Phi(a,b,c)\Phi(a',b',c')\\
    &\qquad\times [V]_{[p\left(b_1 , a\right)]_i+c_1, [p\left(b_1 , a\right)]_i}[V]_{[p\left(b_4 , a_2(a,b,c)\right)]_i, [p\left(b_4 , a_2(a,b,c)\right)]_i+c_4}\\
    &\qquad\times\left[V^*\right]_{[p\left(b'_1 , a'\right)]_i+c'_1, [p\left(b'_1 , a'\right)]_i}\left[V^*\right]_{[p\left(b'_4 , a'_2(a',b',c')\right)]_i, [p\left(b'_4 , a'_2(a',b',c')\right)]_i+c'_4}\\
    &\qquad\times [W]_{[p\left(b_2 , a\right)]_j, [p\left(b_2 , a\right)]_j+c_2}[W]_{[p\left(b_6 , a_3(a,b,c)\right)]_j, [p\left(b_6 , a_3(a,b,c)\right)]_j+c_6}\\
    &\qquad\times \left[W^*\right]_{[p\left(b'_2 , a'\right)]_j, [p\left(b'_2 , a'\right)]_j+c'_2}\left[W^*\right]_{[p\left(b'_6, a'_3(a',b',c')\right)]_j, [p\left(b'_6, a'_3(a',b',c')\right)]_j+c'_6} \\
    \end{split}
\end{equation}
for some constant $A$. $\abs{O_{VW}(U)}^2$ is upper bounded by a negligible number if, for example, at least three of $\{b_i\}$ and $\{b'_i\}$ are canceled by delta functions arising from the ensemble average over $f$ and $p$. In this case, we have
\begin{equation}
    \abs{O_{VW}(U)}^2 \leq \frac{A}{2^k}\times 2^{8}.
\end{equation}
Below, we study how many delta functions can arise due to the ensemble average over $f$ and $p$ for each of $\{c_i\}$ and $\{c'_i\}$. 

\noindent
\textbf{Case:} $c_1=c_2=c_4=c_6=1$,
\begin{equation}
    \begin{split}
    \Phi(\{a(\{c\})\},\{b(\{c\})\})
    &= (-1)^{f(b_1,a_1)+f(b_2,a_1)+f(x_i(b_1,a_1),y_i(b_1,a_1))+f(x_j(b_2,a_1),y_j(b_2,a_1))}\\
    &\quad\times(-1)^{f(b_4,a_2(\{c\}))+f(b_6,a_3(\{c\}))+f(x_i(b_4,a_2(\{c\})),y_i(b_4,a_2(\{c\})))+f(x_j(b_6,a_3(\{c\})),y_j(b_6,a_3(\{c\})))}.
    \end{split}
\end{equation}

In this case, there are eight random functions in total. They can be contracted with each other or other random functions in $\Phi(\{a'(\{c'\})\},\{b'(\{c'\})\})$, and one may see that these contractions can be organized into three different types. The first type is the contraction of $f(b_1,a_1)$ with $f(b_2,a_1)$, which gives
\begin{equation}
    \begin{split}
    \Phi(\{a(\{c\})\},\{b(\{c\})\})
    &= (-1)^{f(b_4,a_2(\{c\}))+f(b_6,a_3(\{c\}))+f(x_i(b_4,a_2(\{c\})),y_i(b_4,a_2(\{c\})))+f(x_j(b_6,a_3(\{c\})),y_j(b_6,a_3(\{c\})))}\\
    &\quad\times\delta_{b_1,b_2}.
    \end{split}
\end{equation}
This contraction introduces a single delta function $\delta_{b_1,b_2}$ by removing four random functions. The second one is the contraction of $f(b_4,a_2(\{c\}))$ with $f(b_6,a_3(\{c\}))$. After the contraction, we get 
\begin{equation}
    \begin{split}
    \Phi(\{a(\{c\})\},\{b(\{c\})\})
    &= (-1)^{f(b_1,a_1)+f(b_2,a_1)+f(x_i(b_1,a_1),y_i(b_1,a_1))+f(x_j(b_2,a_1),y_j(b_2,a_1))}\\
    &\quad\times\delta_{b_4,b_6}\delta_{a_2,y_i(b_4,a_2)}\delta_{a_2,y_j(b_1,a_1)}.
    \end{split}
\end{equation}
The contraction produces $\delta_{b_4,b_6}\delta_{a_2,y_i(b_4,a_2)}\delta_{a_2,y_j(b_1,a_1)}$ in return for removing four random functions. The third one is the contraction of $f(b_1,a_1)$ with $f(x_j(b_2,a_1),y_j(b_2,a_1))$, giving
\begin{equation}
    \begin{split}
    \Phi(\{a(\{c\})\},\{b(\{c\})\})
    &= (-1)^{f(b_4,a_2(\{c\}))+f(b_6,a_3(\{c\}))+f(x_i(b_4,a_2(\{c\})),y_i(b_4,a_2(\{c\})))+f(x_j(b_6,a_3(\{c\})),y_j(b_6,a_3(\{c\})))}\\ &\quad\times(-1)^{f(b_2,a_1)+f(x_i(b_1,a_1),y_i(b_1,a_1))}\delta_{b_1,x_j(b_2,a_1)}\delta_{a_1,y_j(b_2,a_1)}.
    \end{split}
\end{equation}
This contraction generates $\delta_{b_1,x_j(b_2,a_1)}$ and $\delta_{a_1,y_j(b_2,a_1)}$ by removing two random functions. We skip other possible contractions between a random function in $\Phi(\{a(\{c\})\},\{b(\{c\})\})$ and another one in $\Phi(\{a'(\{c'\})\},\{b'(\{c'\})\})$ as they belong to the three types. All types of contractions produce at least a single delta function on $\{b_i\}$ using at most four random functions. Thus, the full contraction of $\Phi(\{a(\{c\})\},\{b(\{c\})\})\Phi(\{a'(\{c'\})\},\{b'(\{c'\})\})$ gives at least four delta functions on $\{b_i\}$. Therefore, $\abs{O_{VW}(U)}^2$ is upper bounded by $O(2^{-k})$.

\noindent
\textbf{Case:} $c_1=c_4=0$ and $c_2=c_6=1$,
\begin{equation}
    \begin{split}
    \Phi(\{a(\{c\})\},\{b(\{c\})\})
    &= (-1)^{f(b_2,a_1)+f(x_j(b_2,a_1),y_j(b_2,a_1))+f(b_6,a_3(\{c\}))+f(x_j(b_6,a_3(\{c\})),y_j(b_6,a_3(\{c\})))}\\
    &\quad\times (-1)^{[p(b_1,a_1)]_i+[p(b_4,a_2(\{c\}))]_i}.
    \end{split}
\end{equation}
This case $\Phi(\{a(\{c\})\},\{b(\{c\})\})$ has four random functions and two random permutation terms in total. Their contractions can be classified into two types. The first type is the contraction of $f(b_2,a_1)$ with $f(b_6,a_3(\{c\}))$. This gives $\delta_{b_2,b_6}\delta_{a_1,y_j(b_2,a_1)}$ by removing the four random functions,
\begin{equation}
    \Phi(\{a(\{c\})\},\{b(\{c\})\})
    = (-1)^{[p(b_1,a_1)]_i+[p(b_4,a_2(\{c\}))]_i}\delta_{b_2,b_6}\delta_{a_1,y_j(b_2,a_1)}.
\end{equation}
The second type is the contraction of $[p(b_1,a_1)]_i$ with $[p(b_4,a_2(\{c\}))]_i$. This introduces $\times\delta_{b_1,b_4}\delta_{a_1,y_j(b_2,a_1)}$ by removing the two random permutation terms,
\begin{equation}
    \begin{split}
        \Phi(\{a(\{c\})\},\{b(\{c\})\})
        &= (-1)^{f(b_2,a_1)+f(x_j(b_2,a_1),y_j(b_2,a_1))+f(b_6,a_3(\{c\}))+f(x_j(b_6,a_3(\{c\})),y_j(b_6,a_3(\{c\})))}\\
        &\quad \times\delta_{b_1,b_4}\delta_{a_1,y_j(b_2,a_1)}.
    \end{split}
\end{equation}
Therefore, contractions of four random functions or two random permutation terms all give delta functions on $\{b_i\}$. In addition, they produces $\delta_{a_1,y_j(b_2,a_1)}$, whose summation over $a_1$ gives $2^{O(k)}$ with a negligible failure probability due to \textbf{Lemma}~\ref{thm:average-bit-flip}. This further suppresses the magnitude of the upper bound of $\abs{O_{VW}(U)}^2$ in [Eq.~\ref{eq:O_VW_square_upper_bound}]. Other possibilities are contractions between functions in $\Phi(\{a(\{c\})\},\{b(\{c\})\})$ and $\Phi(\{a'(\{c'\})\},\{b'(\{c'\})\})$. One may see that each of them again introduces a delta function on $\{b_i\}$ by replacing at most four random functions or two random permutation terms. Therefore, in this case too, $\abs{O_{VW}(U)}^2$ is upper bounded by $O(2^{-k})$.

\noindent
\textbf{Case:} $c_1=c_2=0$ and $c_4=c_6=1$,
\begin{equation}
    \begin{split}
    \Phi(\{a(\{c\})\},\{b(\{c\})\})
    &= (-1)^{f(b_4,a_2)+f(x_i(b_4,a_2(\{c\})),y_i(b_4,a_2(\{c\})))+f(b_6,a_3(\{c\}))+f(x_j(b_6,a_3(\{c\})),y_j(b_6,a_3(\{c\})))}\\
    &\quad\times (-1)^{[p(b_1,a_1)]_i+[p(b_2,a_1)]_j}.
    \end{split}
\end{equation}
Similar to the previous case, this case $\Phi(\{a(\{c\})\},\{b(\{c\})\})$ has four random functions and two random permutation terms. The contraction of random functions gives $\delta_{b_4,b_6}\delta_{a_1,y_i(b_4,a_1)}$,
\begin{equation}
    \Phi(\{a(\{c\})\},\{b(\{c\})\})
    = (-1)^{[p(b_1,a_1)]_i+[p(b_2,a_1)]_j}\delta_{b_4,b_6}\delta_{a_1,y_i(b_4,a_1)}.
\end{equation}
On the other hand, the contraction of $[p(b_1,a_1)]_i$ with $[p(b_2,a_1)]_j$ makes $\Phi(\{a(\{c\})\},\{b(\{c\})\})$ vanishing as $i\neq j$. Contractions between random permutation terms in $\Phi(\{a(\{c\})\},\{b(\{c\})\})$ and $\Phi(\{a'(\{c'\})\},\{b'(\{c'\})\})$ can generate non-vanishing delta functions. For example, the contraction of $[p(b_1,a_1)]_i$ and $[p(b'_1,a'_1)]_i$ in $(-1)^{[p(b_1,a_1)]_i+p(b'_1,a'_1)]_i}$ introduces $\delta_{b_1,b'_1}$ and $\delta_{a_1,a'_1}$. Again, contractions of four random functions or two random permutation terms give delta functions on $\{b_i\}$. Consequently, the full contraction introduces four such delta functions and make $\abs{O_{VW}(U)}^2$ be upper bounded by $O(2^{-k})$.

\noindent
\textbf{Case:} $c_1=c_2=c_4=c_6=0$,
\begin{equation}
    \Phi(\{a(\{c\})\},\{b(\{c\})\})=(-1)^{[p(b_1,a_1)]_i+[p(b_4,a_2(\{c\}))]_i+[p(b_2,a_1)]_j+[p(b_6,a_3(\{c\}))]_j}.
\end{equation}
This is the last case. In this case, $\Phi(\{a(\{c\})\},\{b(\{c\})\})$ has four random permutation terms. The contraction of $[p(b_1,a_1)]_i$ with $[p(b_4,a_2(\{c\}))]_i$ gives $\delta_{b_1,b_4}$,
\begin{equation}
    \Phi(\{a(\{c\})\},\{b(\{c\})\})=(-1)^{[p(b_2,a_1)]_j+[p(b_6,a_3(\{c\}))]_j}\delta_{b_1,b_4},
\end{equation}
and the contraction of $[p(b_2,a_1)]_j$ with $[p(b_6,a_3(\{c\}))]_j$ gives $\delta_{b_1,b_4}$,
\begin{equation}
    \Phi(\{a(\{c\})\},\{b(\{c\})\})=(-1)^{[p(b_1,a_1)]_i+[p(b_4,a_2(\{c\}))]_i}\delta_{b_2,b_6}.
\end{equation}
Thus, the full contraction again introduces four delta functions on $\{b_i\}$, and $\abs{O_{VW}(U)}^2$ should be $O(2^{-k})$.

\end{proof}

\newpage
\hypertarget{sec:non-local-otoc}{}
\subsection*{4. Non-local Pauli OTOCs}
Here, we consider the case when $V$ is given by arbitrary sequences of $\{X_i\}$ and $\{Z_i\}$ operators, e.g., $V = X_4 Z_5 X_6 X_7$, and $W$ is given by $Z_j$ or $X_j$. To show these OTOCs are negligible for $k=\omega(\log n)$, we prove that OTOCs with $(V,W)=(Z_i,Z_j)$, $(X_i,Z_j)$, and $(X_i,X_j)$ are negligible for an RSED with $u = H^{\otimes k} P$ ($H$ is the Hadamard operator, $P$ is a random sign operator) when $k=\omega(\log n)$ in a way that can be generalized to the case with $V$ being a non-local Pauli string. We will not consider cases with $V=Y_i$ or $W=Y_j$ since they give the same OTOC formulas of $V=X_i$ and $V=X_j$ up to phase factors that do not affect the proof. 

Our proofs rely on the validity of \textbf{Lemma}~\ref{thm:average-bit-flip} and \textbf{Lemma}~\ref{thm:parity-sum} introduced below. These lemmas can be applied for a non-local $V$ as well. First, \textbf{Lemma}~\ref{thm:average-bit-flip} holds for multiple bit-flips, so appearances of $\{X_i\}$ multiple times in $V$ are allowed. Next, \textbf{Lemma}~\ref{thm:parity-sum} holds as long as the exponent is given by a random function, which is true for the sum of random phase functions arising from multiple $\{Z_i\}$ terms in $V$. Therefore, the proofs for \hyperlink{sec:XZ-OTOCs}{XZ-OTOCs} as well as \hyperlink{sec:XX-OTOCs}{XX-OTOCs} can be applied to those cases. Combining altogether, The OTOC with a non-local Pauli string $V$ and an onsite Pauli matrix $W$ are negligible for $k=\omega(\log n)$.\\

\noindent
\textbf{Bit-flip distribution}\\
\indent Here, we study how many configurations in a subspace remains under a set of bit flips conjugated by a random permutation. Specifically, we are interested in the probability distribution $P_z(n,k,l)$ of the number $l$ of such configurations with the bit flips in a $k$-bit subspace $b\in\{0,1\}^k$ of the space $\{0,1\}^n$ with a label $a\in\{0,1\}^{n-k}$. We refer such a configuration as a bit-flip allowed configuration and call it with its image under the conjugation by the permutation after the bit-flips as a bit-flip pair. We note that the probability distribution of the number of the bit-flip pairs can be use to upper bound $\sum_a \delta_{a,y_z(b,a)}$ with $b\in\{0,1\}^k$, which is required to upper bound non-local Pauli OTOCs. Below, we find an upper bound of $P_z(n,k,l)$ which is tight enough to argue OTOCs of a certain dynamics are negligible. Although the bit-flip pair and the lemma introduced below are well fitted for multiple bit-flips at once, we will consider each bit-flip pair is made by a single bit-flip for simplicity.

Let us consider a random random permutation $p:\{0,1\}^k\times\{0,1\}^{n-k}\rightarrow \{0,1\}^n$. The image of a subspace $b\in\{0,1\}^k$ with a label $a\in\{0,1\}^{n-k}$ under the permutation can be successively constructed by sampling random bit-strings in $\{0,1\}^n$ without replacement. Let us first find an upper bound of $P_z(n,k,0)$, \textit{i.e.}, the case with no bit-flip pair. Suppose we have sampled $j$ distinct bitstrings $\{b_i\}_{i=1}^j$ having no bit-flip pair. Then, the probability $\tilde{P}_z(\{b_i\}_{i=1}^j)$ of sampling a bitstring $b_{j+1}$ that forms a bit-flip pair with one of $\{b_i\}_{i=1}^j$ is given by
\begin{equation}
    \tilde{P}_z(\{b_i\}_{i=1}^j) = \frac{1}{2^{n}-j}\left|B_z\left(\left\{b_i\right\}_{i=1}^j\right)\right|,
\end{equation}
where $B_z\left(\left\{b_i\right\}_{i=1}^j\right)$ is the set of bitstrings whose elements can form a bit-flip pair with one of $\{b_i\}_{i=1}^j$. Importantly, $B_z\left(\left\{b_i\right\}_{i=1}^j\right)$ is nothing but $\{b_i\}_{i=1}^j$ with bit-flips at the position $z$. This implies that $P_z(n, k, 0)$ is given by
\begin{equation}
P_z(n, k, 0) = \prod_{j=1}^{2^k}\left(1-\frac{j-1}{2^n-j+1}\right).
\end{equation}
This can be generalized to arbitrary $l$ with a caution that bit-flip pairs can be formed during the construction process. Let $\{i_j\}_{j=0}^l$ be the steps in which bit-flip pairs are formed with $i_0=0$. Then, one may see that $P_x(n,k,l)$ is upper bounded by
\begin{equation}
    \begin{split}
    P_z(n, k, l) &= \sum_{i_1=1}^{2^k-l+1} \sum_{i_2=i_1+1}^{2^k-l+2} \cdots \sum_{i_l=i_{l-1}+1}^{2^k} \prod_{m=1}^{l}\left[\frac{i_m-2m+1}{2^n-i_m+1} \prod_{j=i_{m-1}+1}^{i_{m}-1}\left(1-\frac{j-2m+1}{2^n-j}\right)\right]\\ 
    &\qquad\qquad\qquad\qquad\qquad\qquad\times\prod_{j=i_l+1}^{2^k}\left(1-\frac{j-2 l-1}{2^n-j}\right) \\
    & \leq \sum_{i_1=1}^{2^k-l+1} \sum_{i_2=i_1+1}^{2^k-l+2} \cdots \sum_{i_l=i_{l-1}+1}^{2^k} \left(\frac{2^k}{2^n-2^k}\right)^l\left(1-\frac{1}{2^n}\right)^{2^k-l} \\
    & \leq \binom{2^k}{l} \left(\frac{2^{k}}{2^n-2^k}\right)^l\left(1-\frac{1}{2^n}\right)^{2^k-l}\\
    & \leq \left(\frac{2^{2k}}{2^n-2^k}\right)^l\exp(-\frac{1}{2^{n-k+l+1}}).
    \end{split}
\end{equation}
For large $n\gg k$ and $k=\operatorname{polylog}(n)$, the exponential term in the last line converges to unity, so the upper bound asymptotically follows $\Theta(2^{l(2k-n)})$, which decays exponentially in $n$.

Now, let us compute how many bit-flip pairs are formed, or equivalently, estimate how large $\sum_{a}\delta_{a,y_z(b,a)}$ is with high probability. 

\begin{lemma}\label{thm:average-bit-flip}
    With a negligible failure probability, $\sum_{a}\delta_{a,y_z(b,a)}$ is $2^{(1+\epsilon)k}$ for any positive number $\epsilon > 0$.
\end{lemma}
\begin{proof}
    Let us first compute its expectation value as
    \begin{equation}
        \mathbb{E}_S\sum_{a}\delta_{a,y_z(b,a)} = 2^{n-k} \sum_{l=1}^{2^k} P_z(n,k,l) \lesssim  \sum_{l=1}^{2^k} 2^{l(2k-n)+n-k} = 2^{k}\times\frac{1-2^{2^k(2k-n)}}{1-2^{2k-n}}\approx 2^k.
    \end{equation}
    The last approximation holds for $k=\operatorname{polylog}(n)$. Then, Markov's inequality gives
    \begin{equation}
        \operatorname{Pr}\left[\sum_{a}\delta_{a,y_z(b,a)}\geq 2^{(1+\epsilon)k}\right] \leq 2^{-(1+\epsilon)k}\mathbb{E}_S\sum_{a}\delta_{a,y_z(b,a)}\lesssim 2^{-\epsilon k}=\operatorname{negl}(n)
    \end{equation}
    for some constant $\epsilon>0$. Thus, the sum is $2^{O(k)}$ with a probability negligible in $n$.
\end{proof}

\noindent
\textbf{ZZ-OTOCs}\\
\indent Here, we consider the case of $V=Z_i$ and $W=Z_j$ with $i\neq j$. Then, $O_{VW}(U)$ is given by
\begin{equation}
    \begin{split}
        O_{VW}(U)
        &=\frac{1}{2^n} \sum_{\left\{b_i\right\}_{i=1}^8} \sum_{\left\{a_i\right\}_{i=1}^4}(-1)^{\sum_{i=1}^8 f\left(b_i , a_i\right)}[V]_{p\left(b_8 , a_4\right), p\left(b_1 , a_1\right)}[U]_{b_1, b_2}[W]_{p\left(b_2 , a_1\right), p\left(b_3 , a_2\right)}\left[U^{\dagger}\right]_{b_3, b_4}  \\
        &\qquad\qquad\qquad\qquad \times[V]_{p\left(b_4 , a_2\right), p\left(b_5 , a_3\right)}[U]_{b_5, b_6}[W]_{p\left(b_6 , a_3\right), p\left(b_7 , a_4\right)}\left[U^{\dagger}\right]_{b_7, b_8}\\
        &=\frac{1}{2^n} \sum_{\left\{b_i\right\}_{i=1}^4} \sum_{a} (-1)^{[p(b_1,a)]_i+[p(b_2,a)]_j+[p(b_3,a)]_i+[p(b_4,a)]_j} [U]_{b_1, b_2}\left[U^{\dagger}\right]_{b_2, b_3}[U]_{b_3, b_4}\left[U^{\dagger}\right]_{b_4, b_1}.
    \end{split}
\end{equation}
By averaging over random subsets and phases, we get
\begin{equation}\label{eq:ZZ-OTOC-AVGRSP}
    \begin{split}
        \mathbb{E}_{S, f} O_{VW}(U) 
        & =\frac{1}{2^n} \sum_{\left\{b_i\right\}_{i=1}^4} \sum_a \delta_{b_1, b_3}\delta_{b_2, b_4}[U]_{b_1, b_2}\left[U^{\dagger}\right]_{b_2, b_3}[U]_{b_3, b_4}\left[U^{\dagger}\right]_{b_4, b_1} + E_{Z_iZ_j}(U) \\
        & =\frac{1}{2^k} \operatorname{tr}\left(\left(U . * U . * U^*\right) U^{\dagger}\right).
    \end{split}
\end{equation}
where `$.*$' is the element-wise multiplication. Now, let us compute the variance of $O_{VW}$. To this end, we compute the average of the square of $O_{VW}$:
\begin{equation}
    \begin{split}
        \mathbb{E}_{S, f} \left[ \left(O_{VW}(U)\right)^2 \right]
        & = \mathbb{E}_{S, f} \frac{1}{4^n} \sum_{\left\{b_i\right\}_{i=1}^8} \sum_{a_1,a_2} (-1)^{[p(b_1,a_1)]_i+[p(b_2,a_1)]_j+[p(b_3,a_1)]_i+[p(b_4,a_1)]_j} [U]_{b_1, b_2}\left[U^{\dagger}\right]_{b_2, b_3}[U]_{b_3, b_4}\left[U^{\dagger}\right]_{b_4, b_1} \\
        &\qquad\qquad\qquad\qquad\times (-1)^{[p(b_5,a_2)]_i+[p(b_6,a_2)]_j+[p(b_7,a_2)]_i+[p(b_8,a_2)]_j} [U]_{b_5, b_6}\left[U^{\dagger}\right]_{b_6, b_7}[U]_{b_7, b_8}\left[U^{\dagger}\right]_{b_8, b_5}\\
        & = \frac{1}{4^n} \sum_{\left\{b_i\right\}_{i=1}^8} \sum_{a_1,a_2} [U]_{b_1, b_2}\left[U^{\dagger}\right]_{b_2, b_3}[U]_{b_3, b_4}\left[U^{\dagger}\right]_{b_4, b_1}[U]_{b_5, b_6}\left[U^{\dagger}\right]_{b_6, b_7}[U]_{b_7, b_8}\left[U^{\dagger}\right]_{b_8, b_5} \\
        &\qquad\qquad\qquad\quad\times \left(\delta_{b_1,b_3}\delta_{b_5,b_7}+\delta_{b_1,b_5}\delta_{b_3,b_7}\delta_{a_1,a_2}+\delta_{b_1,b_7}\delta_{b_3,b_5}\delta_{a_1,a_2}-2\delta_{b_1,b_3,b_5,b_7}\delta_{a_1,a_2}\right)\\
        &\qquad\qquad\qquad\quad\times \left(\delta_{b_2,b_4}\delta_{b_6,b_8}+\delta_{b_2,b_6}\delta_{b_4,b_8}\delta_{a_1,a_2}+\delta_{b_2,b_8}\delta_{b_4,b_6}\delta_{a_1,a_2}-2\delta_{b_2,b_4,b_6,b_8}\delta_{a_1,a_2}\right).
    \end{split}
\end{equation}

Now, let us consider when $U$ is the product of Hadamard gates. In this case, $\mathbb{E}_{S, f} \left[ O_{VW}(U) \right]$ is given by
\begin{equation}
    \begin{split}
        \mathbb{E}_{S, f} \left[ O_{VW}(U) \right]
        &= 2^{-k} \operatorname{tr} \left[\left[\left(U . * U . * U^*\right) U^{\dagger}\right]_{i j}\right]\\
        &= 2^{-k} \operatorname{tr} \left[\left[\left(H . * H . * H^*\right) H^{\dagger}\right]_{i j}\right]\\
        &= 2^{-2k} \operatorname{tr} \left[\left[H H^\dagger\right]_{i j}\right]\\
        &= 2^{-k}.
    \end{split}
\end{equation}
Thus, if we set $k=\omega(\log n)$, then the ensemble average of the correlator becomes negligible. Next, let us compute its second moments.
\begin{equation}
    \begin{split}
        \mathbb{E}_{S, f} \left[ \left(O_{VW}(U)\right)^2 \right]
        & =\frac{1}{4^{n+2k}} \sum_{\left\{b_i\right\}_{i=1}^8} \sum_{a_1,a_2} (-1)^{b_1\cdot b_2 + b_2\cdot b_3 + b_3\cdot b_4 + b_4\cdot b_1 + b_5\cdot b_6 + b_6\cdot b_7 + b_7\cdot b_8 + b_8\cdot b_5} \\
        &\qquad\qquad\qquad\quad\times \left(\delta_{b_1,b_3}\delta_{b_5,b_7}+\delta_{b_1,b_5}\delta_{b_3,b_7}\delta_{a_1,a_2}+\delta_{b_1,b_7}\delta_{b_3,b_5}\delta_{a_1,a_2}-2\delta_{b_1,b_3,b_5,b_7}\delta_{a_1,a_2}\right)\\
        &\qquad\qquad\qquad\quad\times \left(\delta_{b_2,b_4}\delta_{b_6,b_8}+\delta_{b_2,b_6}\delta_{b_4,b_8}\delta_{a_1,a_2}+\delta_{b_2,b_8}\delta_{b_4,b_6}\delta_{a_1,a_2}-2\delta_{b_2,b_4,b_6,b_8}\delta_{a_1,a_2}\right)\\
        &=\frac{1}{4^{n+2k}}\left[\sum_{\left\{b_i\right\}_{b=1}^4} \sum_{a_1,a_2} 1+8\times\sum_{\left\{b_i\right\}_{b=1}^4} \sum_{a} 1-6\times\sum_{\left\{b_i\right\}_{b=1}^3} \sum_{a} 1+\sum_{\left\{b_i\right\}_{b=1}^2} \sum_{a} 1\right]\\
        &=\frac{1}{2^{2n+4k}}\left[2^{2n+2k}+8\times 2^{n+3k}-6\times 2^{n+2k} + 2^{n+k}\right]\\
        &=\frac{1}{2^{2k}}+\frac{8}{2^{n+k}}-\frac{6}{2^{n+2k}}+\frac{1}{2^{n+3k}}
    \end{split}
\end{equation}
Thus, the variance of the correlator is given by
\begin{equation}
    \operatorname{Var}_{S,f} \left[O_{VW}(U)\right] = \mathbb{E}_{S, f} \left[ \left(O_{VW}(U)\right)^2 \right] - \left[\mathbb{E}_{S, f} O_{VW}(U) \right]^2 = \frac{8}{2^{n+k}}-\frac{6}{2^{n+2k}}+\frac{1}{2^{n+3k}},
\end{equation}
which is exponentially vanishing in $n$. Thus, the ensemble of $\{O_{VW}(U)\}$ over random functions and permutations is highly concentrated near the averaged value.\\

\noindent
\hypertarget{sec:XZ-OTOCs}{\textbf{XZ-OTOCs}}\\
\indent Here, we consider the case of $V=X_i$ and $W=Z_j$ with $i\neq j$. Then, $O_{VW}(U)$ is given by
\begin{equation*}
    \begin{split}
    O_{VW}(U)
    &=\frac{1}{2^n} \sum_{\left\{b_i\right\}_{i=1}^{8}} \sum_{\left\{a_i\right\}_{i=1}^4}(-1)^{\sum_{i=1}^{8} f\left(b_i ,a_{\lfloor(i+1)/2\rfloor}\right)}[V]_{p\left(b_8 , a_4\right), p\left(b_1 , a_1\right)}\left[U\right]_{b_1, b_2}[W]_{p\left(b_2 , a_1\right), p\left(b_3 , a_2\right)} \\
    &\qquad\qquad\qquad\qquad \times \left[U^\dagger\right]_{b_3, b_4} [V]_{p\left(b_4 , a_2\right), p\left(b_5 , a_3\right)}\left[U\right]_{b_5, b_6}[W]_{p\left(b_6 , a_3\right), p\left(b_7 , a_4\right)}\left[U^\dagger\right]_{b_7, b_8} \\
    & =\frac{1}{2^n} \sum_{\left\{b_i\right\}_{i=1}^8} \sum_{\left\{a_i\right\}_{i=1}^4}(-1)^{\sum_{i=1}^{8} f\left(b_i ,a_{\lfloor(i+1)/2\rfloor}\right)+\left[p\left(b_2 , a_1\right)\right]_j+\left[p\left(b_6 , a_3\right)\right]_j}\left[U\right]_{b_1, b_2}\left[U^\dagger\right]_{b_3, b_4}\left[U\right]_{b_5, b_6}\left[U^\dagger\right]_{b_7, b_8} \\
    & \qquad\qquad\qquad\qquad\times \delta_{p\left(b_8 , a_4\right)+\hat{e}_i, p\left(b_1 , a_1\right)} \delta_{p\left(b_4 , a_2\right)+\hat{e}_i, p\left(b_5 , a_3\right)} \delta_{p\left(b_2 , a_1\right), p\left(b_3 , a_2\right)} \delta_{p\left(b_6 , a_3\right), p\left(b_7 , a_4\right)} \\
    & =\frac{1}{2^n} \sum_{\left\{b_i\right\}_{i=1}^6} \sum_{\left\{a_i\right\}_{i=1}^2}(-1)^{f\left(b_1 , a_1\right)+f\left(b_3 , a_1\right)+f\left(b_4 , a_2\right)+f\left(b_6 , a_2\right)+\left[p\left(b_2 , a_1\right)\right]_j+\left[p\left(b_5 , a_2\right)\right]_j} \\
    & \qquad\qquad\qquad\qquad\times  \left[U\right]_{b_1, b_2}\left[U^\dagger\right]_{b_2, b_3}\left[U\right]_{b_4, b_5}\left[U^\dagger\right]_{b_5, b_6} \delta_{p\left(b_6 , a_2\right)+\hat{e}_i, p\left(b_1 , a_1\right)} \delta_{p\left(b_3 , a_1\right)+\hat{e}_i, p\left(b_4 , a_2\right)} \\
    \end{split}
\end{equation*}
The terms in the last line vanish upon averaging over the random function $f$ unless $b_1=b_3$ or $p(b_1,a)=p(b_3,a)+\hat{e}_i$ hold, which are exclusive. Thus, the ensemble average of the four-point correlator is given by
\begin{equation}
    \begin{split}
        \mathbb{E}_{S,f}\left[O_{VW}(U)\right]
        &=\mathbb{E}_{S}\frac{1}{2^n} \sum_{\left\{b_i\right\}_{i=1}^6} \sum_{\left\{a_i\right\}_{i=1}^2}(-1)^{\left[p\left(b_2 , a_1\right)\right]_j+\left[p\left(b_5 , a_2\right)\right]_j}\\
        & \qquad\qquad\qquad\qquad\times(\delta_{b_1,b_3}\delta_{b_4,b_6}+\delta_{a_1,a_2}(\delta_{b_1,b_4}\delta_{b_3,b_6}+\delta_{b_1,b_6}\delta_{b_3,b_4}-2\delta_{b_1,b_3,b_4,b_6})) \\
        & \qquad\qquad\qquad\qquad\times  \left[U\right]_{b_1, b_2}\left[U^\dagger\right]_{b_2, b_3}\left[U\right]_{b_4, b_5}\left[U^\dagger\right]_{b_5, b_6} \delta_{p\left(b_6 , a_2\right)+\hat{e}_i, p\left(b_1 , a_1\right)} \delta_{p\left(b_3 , a_1\right)+\hat{e}_i, p\left(b_4 , a_2\right)}.
    \end{split}
\end{equation}
Similarly, the second moment of the correlator is given by
\begin{equation}
    \begin{split}
        \mathbb{E}_{S, f} \left[ \left(O_{VW}(U)\right)^2 \right]
        &= \mathbb{E}_{S, f} \frac{1}{4^n} \sum_{\left\{b_i\right\}_{i=1}^{12}} \sum_{\{a_i\}_{i=1}^4} (-1)^{f\left(b_1 , a_1\right)+f\left(b_3 , a_1\right)+f\left(b_4 , a_2\right)+f\left(b_6 , a_2\right)+\left[p\left(b_2 , a_1\right)\right]_j+\left[p\left(b_5 , a_2\right)\right]_j} \\
        & \qquad\qquad\qquad\times (-1)^{f\left(b_7 , a_3\right)+f\left(b_9 , a_3\right)+f\left(b_{10} , a_4\right)+f\left(b_{12} , a_4\right)+\left[p\left(b_8 , a_3\right)\right]_j+\left[p\left(b_{11} , a_4\right)\right]_j}\\
        & \qquad\qquad\qquad\times  \left[U\right]_{b_1, b_2}\left[U^\dagger\right]_{b_2, b_3}\left[U\right]_{b_4, b_5}\left[U^\dagger\right]_{b_5, b_6} \delta_{p\left(b_6 , a_2\right)+\hat{e}_i, p\left(b_1 , a_1\right)} \delta_{p\left(b_3 , a_1\right)+\hat{e}_i, p\left(b_4 , a_2\right)} \\
        & \qquad\qquad\qquad\times  \left[U\right]_{b_7, b_8}\left[U^\dagger\right]_{b_8, b_9}\left[U\right]_{b_{10}, b_{11}}\left[U^\dagger\right]_{b_{11}, b_{12}} \delta_{p\left(b_{12} , a_4\right)+\hat{e}_i, p\left(b_7 , a_3\right)} \delta_{p\left(b_9 , a_3\right)+\hat{e}_i, p\left(b_{10} , a_4\right)}.
    \end{split}
\end{equation}

Now, let us replace $U$ by the product of the Hadamard gates. The following calculation shows that the ensemble average of $O_{VW}(U)$ has a negligible magnitude. 
\begin{equation}\label{eq:XZ-OTOC-Hada}
    \begin{split}
        \mathbb{E}_{S,f}\left[O_{VW}(U)\right]
        &=\mathbb{E}_{S}\frac{1}{2^{n+2k}} \sum_{\left\{b_i\right\}_{i=1}^6} \sum_{\left\{a_i\right\}_{i=1}^2}(-1)^{\left[p\left(b_2 , a_1\right)\right]_j+\left[p\left(b_5 , a_2\right)\right]_j+b_1\cdot b_2 + b_2\cdot b_3 + b_4\cdot b_5 + b_5\cdot b_6} \\
        & \qquad\qquad\qquad\qquad\times(\delta_{b_1,b_3}\delta_{b_4,b_6}+\delta_{a_1,a_2}(\delta_{b_1,b_4}\delta_{b_3,b_6}+\delta_{b_1,b_6}\delta_{b_3,b_4}-2\delta_{b_1,b_3,b_4,b_6})) \\
        & \qquad\qquad\qquad\qquad\times \delta_{p\left(b_6 , a_2\right)+\hat{e}_i, p\left(b_1 , a_1\right)} \delta_{p\left(b_3 , a_1\right)+\hat{e}_i, p\left(b_4 , a_2\right)} \\
        &=\mathbb{E}_{S}\frac{1}{2^{n+2k}} \sum_{b_1,b_2,b_4,b_5} \sum_{a_1,a_2}(-1)^{\left[p\left(b_2 , a_1\right)\right]_j+\left[p\left(b_5 , a_2\right)\right]_j}\delta_{p\left(b_4 , a_2\right)+\hat{e}_i, p\left(b_1 , a_1\right)} \\
        &\quad+\mathbb{E}_{S}\frac{1}{2^{n+2k}} \sum_{\left\{b_i\right\}_{i=1}^6} \sum_{a}(-1)^{\left[p\left(b_2 , a\right)\right]_j+\left[p\left(b_5 , a\right)\right]_j+b_1\cdot b_2 + b_2\cdot b_3 + b_4\cdot b_5 + b_5\cdot b_6}\\
        & \qquad\qquad\qquad\qquad\times(\delta_{b_1,b_4}\delta_{b_3,b_6}+\delta_{b_1,b_6}\delta_{b_3,b_4}-2\delta_{b_1,b_3,b_4,b_6}) \\
        & \qquad\qquad\qquad\qquad\times \delta_{p\left(b_6 , a\right)+\hat{e}_i, p\left(b_1 , a\right)} \delta_{p\left(b_3 , a\right)+\hat{e}_i, p\left(b_4 , a\right)} \\
        &=\mathbb{E}_{S}\frac{1}{2^{n+2k}} \sum_{b_1,b_2,b_5} \sum_{a}(-1)^{\left[p\left(b_2 , a\right)\right]_j+\left[p\left(b_5 , y_i(b_1,a)\right)\right]_j}\left(1-\delta_{a,y_i(b_1,a)}\right) \\
        &\quad+\mathbb{E}_{S}\frac{1}{2^{n+2k}} \sum_{b_1,b_2,b_5} \sum_{a}(-1)^{\left[p\left(b_2 , a\right)\right]_j+\left[p\left(b_5 , a\right)\right]_j+(b_1\oplus x_i(b_1,a))\cdot (b_2\oplus b_5)}\delta_{a,y_i(b_1,a)}\left(1-\delta_{b_1,x_i(b_1,a)}\right)\\
        &\leq \mathbb{E}_{S}\frac{1}{2^{n+2k}} \sum_{b_1,b_2,b_5} \sum_{a}(-1)^{\left[p\left(b_2 , a\right)\right]_j+\left[p\left(b_5 , y_i(b_1,a)\right)\right]_j}\left(1-\delta_{a,y_i(b_1,a)}\right) \\
        &\quad+\mathbb{E}_{S} \frac{1}{2^{n}} \sum_{b_1} \left(\sum_{a}\delta_{a,y_i(b_1,a)}\right)\\
        &\leq \frac{1}{2^{n/4}} + \frac{1}{2^{n-2k}}.
    \end{split}
\end{equation}
The first term in the second last line of Eq.~\eqref{eq:XZ-OTOC-Hada} can be upper bounded using the fact that $[p(b_2,a)]_j$ and $[p(b_5,y_i(b_1,a)]_j$ are independent binary random variable with a probability higher than $1-\operatorname{negl}(n)$. Their parity summation follows the binomial distribution with the number of trials given by the total number of summand $m$. Importantly, the summation is upper bounded by $\sqrt{m}e^{\epsilon n}$ for some $\epsilon>0$ with a negligible failure probability due to the following lemma, and thus, by setting $\epsilon = 1/4$, we get the upper bound $2^{-1/4}$ of the summation.

\begin{lemma}\label{thm:parity-sum}
    Let $X$ be a binomial variable $\sum_{i=1}^m (-1)^{x_i}$ with bits $\{x_i\}$ uniformly randomly sampled from $\{0,1\}$. Then, the probability of $|X|$ exceeds $\sqrt{m}\cdot\omega(\operatorname{poly}(n))$ is negligible in $n$.
\end{lemma}
\begin{proof}
    For a binomial variable $X=\sum_{i=1}^m (-1)^{x_i}$ with random bits $x_i\in\{0,1\}$, its distribution can be approximated by the normal distribution $\mathcal{N}(\mu=0,\sigma=1)$ due to the central limit theorem. Then, we need to find a number $a$ such that
    \begin{equation}
        \operatorname{Pr}\left[ \abs{X} > a \right] \leq \operatorname{negl}(n).
    \end{equation}
    Due to the normal approximation, the left hand side becomes 
    \begin{equation}
        \operatorname{Pr}\left[ \abs{X} > a \right] \approx \frac{2\sqrt{m}}{a}e^{-a^2/2m}
    \end{equation}
    for sufficiently large $a^2/m$. This implies that
    \begin{equation}
        a = \sqrt{m}\cdot \omega(\operatorname{poly}(n)).
    \end{equation}    
\end{proof}  

The second term in the same line can be upper bounded using \textbf{Lemma}~\ref{thm:average-bit-flip}. The sum of $\delta_{a,y_i(b,a)}$ over $a$ is upper bounded by $2^{2k}$, so the second term is upper bounded by $2^{2k-n}$. Combining these gives the last inequality, and thus, the ensemble average value of $O_{VW}(U)$ is negligible in $n$ for the embedded Hadamard gate.

Next, let us compute the variance of $O_{VW}(U)$. To this end, we need to compute its second moment first.
\begin{equation}
    \begin{split}
        \mathbb{E}_{S, f} \left[ \left(O_{VW}(U)\right)^2 \right]
        &= \mathbb{E}_{S, f} \frac{1}{4^{n+2k}} \sum_{\left\{b_i\right\}_{i=1}^{12}} \sum_{\{a_i\}_{i=1}^4} (-1)^{f\left(b_1 , a_1\right)+f\left(b_3 , a_1\right)+f\left(b_4 , a_2\right)+f\left(b_6 , a_2\right)+\left[p\left(b_2 , a_1\right)\right]_j+\left[p\left(b_5 , a_2\right)\right]_j} \\
        & \qquad\qquad\qquad\times (-1)^{f\left(b_7 , a_3\right)+f\left(b_9 , a_3\right)+f\left(b_{10} , a_4\right)+f\left(b_{12} , a_4\right)+\left[p\left(b_8 , a_3\right)\right]_j+\left[p\left(b_{11} , a_4\right)\right]_j}\\
        & \qquad\qquad\qquad\times (-1)^{b_1\cdot b_2 + b_2\cdot b_3 + b_4\cdot b_5 + b_5\cdot b_6} \delta_{p\left(b_6 , a_2\right)+\hat{e}_i, p\left(b_1 , a_1\right)} \delta_{p\left(b_3 , a_1\right)+\hat{e}_i, p\left(b_4 , a_2\right)} \\
        & \qquad\qquad\qquad\times (-1)^{b_7\cdot b_8 + b_8\cdot b_9 + b_{10}\cdot b_{11} + b_{11}\cdot b_{12}} \delta_{p\left(b_{12} , a_4\right)+\hat{e}_i, p\left(b_7 , a_3\right)} \delta_{p\left(b_9 , a_3\right)+\hat{e}_i, p\left(b_{10} , a_4\right)} \\
        &= \mathbb{E}_{S} \frac{1}{4^{n+2k}} \sum_{\left\{b_i\right\}_{i=1}^{12}} \sum_{\{a_i\}_{i=1}^4} (-1)^{\left[p\left(b_2 , a_1\right)\right]_j+\left[p\left(b_5 , a_2\right)\right]_j+\left[p\left(b_8 , a_3\right)\right]_j+\left[p\left(b_{11} , a_4\right)\right]_j} \\
        & \qquad\qquad\qquad\times (-1)^{b_1\cdot b_2 + b_2\cdot b_3 + b_4\cdot b_5 + b_5\cdot b_6} \delta_{p\left(b_6 , a_2\right)+\hat{e}_i, p\left(b_1 , a_1\right)} \delta_{p\left(b_3 , a_1\right)+\hat{e}_i, p\left(b_4 , a_2\right)} \\
        & \qquad\qquad\qquad\times (-1)^{b_7\cdot b_8 + b_8\cdot b_9 + b_{10}\cdot b_{11} + b_{11}\cdot b_{12}} \delta_{p\left(b_{12} , a_4\right)+\hat{e}_i, p\left(b_7 , a_3\right)} \delta_{p\left(b_9 , a_3\right)+\hat{e}_i, p\left(b_{10} , a_4\right)} \\
        & \qquad\qquad\qquad\times [\delta_{b_1,b_3}\delta_{b_4,b_6}\delta_{b_7,b_9}\delta_{b_{10},b_{12}}\\
        & \qquad\qquad\qquad\quad + \delta_{a_1,a_2}\delta_{b_7,b_9}\delta_{b_{10},b_{12}}(\delta_{b_1,b_4}\delta_{b_3,b_6}+\delta_{b_1,b_6}\delta_{b_3,b_4})\\
        & \qquad\qquad\qquad\quad + \delta_{a_1,a_3}\delta_{b_4,b_6}\delta_{b_{10},b_{12}}(\delta_{b_1,b_7}\delta_{b_3,b_9}+\delta_{b_1,b_9}\delta_{b_3,b_7})\\
        & \qquad\qquad\qquad\quad + \delta_{a_1,a_4}\delta_{b_4,b_6}\delta_{b_7,b_9}(\delta_{b_1,b_{10}}\delta_{b_3,b_{12}}+\delta_{b_1,b_{12}}\delta_{b_3,b_{10}})\\
        & \qquad\qquad\qquad\quad + \delta_{a_2,a_3}\delta_{b_1,b_3}\delta_{b_{10},b_{12}}(\delta_{b_4,b_7}\delta_{b_6,b_9}+\delta_{b_4,b_9}\delta_{b_6,b_7})\\
        & \qquad\qquad\qquad\quad + \delta_{a_2,a_4}\delta_{b_1,b_3}\delta_{b_7,b_9}(\delta_{b_4,b_{10}}\delta_{b_6,b_{12}}+\delta_{b_4,b_{12}}\delta_{b_6,b_{10}})\\
        & \qquad\qquad\qquad\quad + \delta_{a_3,a_4}\delta_{b_1,b_3}\delta_{b_4,b_6}(\delta_{b_7,b_{10}}\delta_{b_9,b_{12}}+\delta_{b_7,b_{12}}\delta_{b_9,b_{10}})\\
        & \qquad\qquad\qquad\quad + \delta_{a_1,a_2}\delta_{a_3,a_4}(\delta_{b_1,b_4}\delta_{b_3,b_6}+\delta_{b_1,b_6}\delta_{b_3,b_4})(\delta_{b_7,b_{10}}\delta_{b_9,b_{12}}+\delta_{b_7,b_{12}}\delta_{b_9,b_{10}})\\
        & \qquad\qquad\qquad\quad + \delta_{a_1,a_3}\delta_{a_2,a_4}(\delta_{b_1,b_7}\delta_{b_3,b_9}+\delta_{b_1,b_9}\delta_{b_3,b_7})(\delta_{b_4,b_{10}}\delta_{b_6,b_{12}}+\delta_{b_4,b_{12}}\delta_{b_6,b_{10}})\\
        & \qquad\qquad\qquad\quad + \delta_{a_1,a_4}\delta_{a_2,a_3}(\delta_{b_7,b_4}\delta_{b_9,b_6}+\delta_{b_7,b_6}\delta_{b_9,b_4})(\delta_{b_1,b_{10}}\delta_{b_3,b_{12}}+\delta_{b_3,b_{12}}\delta_{b_1,b_{10}})\\
        & \qquad\qquad\qquad\quad + \delta_{a_1,a_2,a_3}\delta_{b_{10},b_{12}}(\delta_{b_1,b_4}\delta_{b_3,b_7}\delta_{b_6,b_9}+\delta_{b_1,b_4}\delta_{b_3,b_9}\delta_{b_6,b_7}+\cdots)\\
        & \qquad\qquad\qquad\quad + (\text{other three-point delta function terms of $\{a_i\}$})\\
        & \qquad\qquad\qquad\quad + (\text{four-point delta function terms of $\{a_i\}$})]\\
    \end{split}
\end{equation}

\begin{equation}
    \begin{split}
        &= \mathbb{E}_{S} \frac{1}{4^{n+2k}} \sum_{b_1,b_2,b_4,b_5,b_7,b_8,b_{10},b_{11}} \sum_{\{a_i\}_{i=1}^4} (-1)^{\left[p\left(b_2 , a_1\right)\right]_j+\left[p\left(b_5 , a_2\right)\right]_j+\left[p\left(b_8 , a_3\right)\right]_j+\left[p\left(b_{11} , a_4\right)\right]_j}\\
        & \qquad\qquad\qquad\qquad\qquad\qquad\qquad\qquad\times \delta_{p\left(b_4 , a_2\right)+\hat{e}_i, p\left(b_1 , a_1\right)}\delta_{p\left(b_{10} , a_4\right)+\hat{e}_i, p\left(b_7 , a_3\right)} \\
        &\quad + \mathbb{E}_{S} \frac{1}{4^{n+2k}} \sum_{b_1,b_2,b_3,b_5,b_7,b_8,b_{10},b_{11}} \sum_{a_1,a_3,a_4} (-1)^{\left[p\left(b_2 , a_1\right)\right]_j+\left[p\left(b_5 , a_1\right)\right]_j+\left[p\left(b_8 , a_3\right)\right]_j+\left[p\left(b_{11} , a_4\right)\right]_j}\\
        & \qquad\qquad\qquad\times (-1)^{b_1\cdot b_2 + b_2\cdot b_3 + b_1\cdot b_5 + b_5\cdot b_3}\delta_{p\left(b_3 , a_1\right)+\hat{e}_i, p\left(b_1 , a_1\right)}\delta_{p\left(b_{10} , a_4\right)+\hat{e}_i, p\left(b_7 , a_3\right)} \\
        &\quad + \mathbb{E}_{S} \frac{1}{4^{n+2k}} \sum_{b_1,b_2,b_3,b_5,b_7,b_8,b_{10},b_{11}} \sum_{a_1,a_2,a_4} (-1)^{\left[p\left(b_2 , a_1\right)\right]_j+\left[p\left(b_5 , a_2\right)\right]_j+\left[p\left(b_8 , a_1\right)\right]_j+\left[p\left(b_{11} , a_4\right)\right]_j}\\
        & \qquad\qquad\qquad\qquad\qquad\times (-1)^{b_1\cdot b_2 + b_2\cdot b_3 + b_3\cdot b_5 + b_5\cdot b_1}\\
        & \qquad\qquad\qquad\qquad\qquad\times\delta_{p\left(b_1 , a_2\right)+\hat{e}_i, p\left(b_1 , a_1\right)}\delta_{p\left(b_3 , a_1\right)+\hat{e}_i, p\left(b_3 , a_2\right)}\delta_{p\left(b_{10} , a_4\right)+\hat{e}_i, p\left(b_7 , a_1\right)} \\
        &\quad + (\text{other $\delta_{a_i,a_j}$ terms}) \\
        &\quad + \mathbb{E}_{S} \frac{1}{4^{n+2k}} \sum_{b_1,b_2,b_3,b_5,b_7,b_8,b_{11}} \sum_{a_1,a_3} (-1)^{\left[p\left(b_2 , a_1\right)\right]_j+\left[p\left(b_5 , a_1\right)\right]_j+\left[p\left(b_8 , a_3\right)\right]_j+\left[p\left(b_{11} , a_3\right)\right]_j}\\
        & \qquad\qquad\qquad\qquad\qquad\times (-1)^{b_1\cdot b_2 + b_2\cdot b_3 + b_1\cdot b_5 + b_5\cdot b_3 + b_7\cdot b_8 + b_8\cdot b_9 + b_7\cdot b_{11} + b_{11}\cdot b_9} \\
        & \qquad\qquad\qquad\qquad\qquad\times\delta_{p\left(b_3 , a_1\right)+\hat{e}_i, p\left(b_1 , a_1\right)}\delta_{p\left(b_9 , a_3\right)+\hat{e}_i, p\left(b_7 , a_3\right)} \\
        &\quad + \mathbb{E}_{S} \frac{1}{4^{n+2k}} \sum_{b_1,b_2,b_3,b_4,b_5,b_6,b_8,b_{11}} \sum_{a_1,a_2} (-1)^{\left[p\left(b_2 , a_1\right)\right]_j+\left[p\left(b_5 , a_2\right)\right]_j+\left[p\left(b_8 , a_1\right)\right]_j+\left[p\left(b_{11} , a_2\right)\right]_j}\\
        & \qquad\qquad\qquad\qquad\qquad\times (-1)^{b_1\cdot b_2 + b_2\cdot b_3 + b_4\cdot b_5 + b_5\cdot b_6 + b_1\cdot b_8 + b_8\cdot b_3 + b_4\cdot b_{11} + b_{11}\cdot b_6} \\
        & \qquad\qquad\qquad\qquad\qquad\times\delta_{p\left(b_6 , a_2\right)+\hat{e}_i, p\left(b_1 , a_1\right)}\delta_{p\left(b_3 , a_1\right)+\hat{e}_i, p\left(b_4 , a_2\right)} \\
        &\quad + \mathbb{E}_{S} \frac{1}{4^{n+2k}} \sum_{b_1,b_2,b_3,b_4,b_5,b_6,b_8,b_{11}} \sum_{a_1,a_2} (-1)^{\left[p\left(b_2 , a_1\right)\right]_j+\left[p\left(b_5 , a_2\right)\right]_j+\left[p\left(b_8 , a_1\right)\right]_j+\left[p\left(b_{11} , a_2\right)\right]_j}\\
        & \qquad\qquad\qquad\times (-1)^{b_1\cdot b_2 + b_2\cdot b_3 + b_4\cdot b_5 + b_5\cdot b_6 + b_1\cdot b_8 + b_8\cdot b_3 + b_6\cdot b_{11} + b_{11}\cdot b_4} \\
        & \qquad\qquad\qquad\times\delta_{p\left(b_6 , a_2\right)+\hat{e}_i, p\left(b_1 , a_1\right)}\delta_{p\left(b_3 , a_1\right)+\hat{e}_i, p\left(b_4 , a_2\right)}\delta_{p\left(b_4 , a_2\right)+\hat{e}_i, p\left(b_1 , a_1\right)}\delta_{p\left(b_3 , a_1\right)+\hat{e}_i, p\left(b_6 , a_2\right)} \\
        &\quad + (\text{other $\delta_{a_i,a_j}\delta_{a_k,a_l}$ terms}) \\
        &\quad + \mathbb{E}_{S} \frac{1}{4^{n+2k}} \sum_{b_1,b_2,b_3,b_5,b_6,b_8,b_{10},b_{11}} \sum_{a_1,a_4} (-1)^{\left[p\left(b_2 , a_1\right)\right]_j+\left[p\left(b_5 , a_1\right)\right]_j+\left[p\left(b_8 , a_1\right)\right]_j+\left[p\left(b_{11} , a_4\right)\right]_j}\\
        & \qquad\qquad\qquad\times (-1)^{b_1\cdot b_2 + b_2\cdot b_3 + b_1\cdot b_5 + b_5\cdot b_6 + b_3\cdot b_8 + b_8\cdot b_6}  \\
        & \qquad\qquad\qquad\times \delta_{p\left(b_6 , a_1\right)+\hat{e}_i, p\left(b_1 , a_1\right)} \delta_{p\left(b_3 , a_1\right)+\hat{e}_i, p\left(b_1 , a_1\right)}\delta_{p\left(b_{10} , a_4\right)+\hat{e}_i, p\left(b_3 , a_1\right)} \delta_{p\left(b_6 , a_1\right)+\hat{e}_i, p\left(b_{10} , a_4\right)} \\
        & \quad + (\text{other three-point delta function terms of $\{a_i\}$})\\
        & \quad + (\text{four-point delta function terms of $\{a_i\}$})\\
    \end{split}
\end{equation}

\begin{equation}\label{eq:XZ-OTOC-Hada-Var}
    \begin{split}
        &= \mathbb{E}_{S} \frac{1}{4^{n+k}} \sum_{b_2,b_5,b_8,b_{11}} \sum_{a_1,a_3} (-1)^{\left[p\left(b_2 , a_1\right)\right]_j+\left[p\left(b_5 , y_i(b_1,a_1)\right)\right]_j+\left[p\left(b_8 , a_3\right)\right]_j+\left[p\left(b_{11} , y_i(b_7,a_3)\right)\right]_j}\\
        &\quad + \mathbb{E}_{S} \frac{1}{4^{n+2k}} \sum_{b_1,b_2,b_5,b_7,b_8,b_{11}} \sum_{a_3} \left(\sum_{a_1} \delta_{a_1,y_i(b_1,a_1)}\right)(-1)^{b_1\cdot b_2 + b_2\cdot x_i(b_1,a_1) + b_1\cdot b_5 + b_5\cdot x_i(b_1,a_1)} \\    
        &\qquad\qquad\qquad\qquad\qquad\qquad\qquad\times (-1)^{\left[p\left(b_2 , a_1\right)\right]_j+\left[p\left(b_5 , a_1\right)\right]_j+\left[p\left(b_8 , a_1\right)\right]_j+\left[p\left(b_{11} , y_i(b_7,a_3)\right)\right]_j} \\
        &\quad + \mathbb{E}_{S} \frac{1}{4^{n+2k}} \sum_{b_1,b_2,b_3,b_5,b_7,b_8,b_{11}} \sum_{a} (-1)^{b_1\cdot b_2 + b_2\cdot b_3 + b_3\cdot b_5 + b_5\cdot b_1}\delta_{b_1,x_i(b_1,a)}\delta_{b_3,x_i(b_3,a)}\delta_{y_i(b_1,a),y_i(b_3,a)}\\
        &\qquad\qquad\qquad\qquad\qquad\qquad\qquad\times (-1)^{\left[p\left(b_2 , a\right)\right]_j+\left[p\left(b_5 , y_i(b_1,a)\right)\right]_j+\left[p\left(b_8 , a\right)\right]_j+\left[p\left(b_{11} , y_i(b_7,a)\right)\right]_j}  \\
        &\quad + (\text{other $\delta_{a_i,a_j}$ terms}) \\
        &\quad + \mathbb{E}_{S} \frac{1}{4^{n+2k}} \sum_{b_1,b_2,b_5,b_7,b_8,b_{11}} \sum_{a_1,a_3} (-1)^{b_1\cdot b_2 + b_2\cdot x_i(b_1,a_1) + b_1\cdot b_5 + b_5\cdot x_i(b_1,a_1) + b_7\cdot b_8 + b_8\cdot x_i(b_7,a_3) + b_7\cdot b_{11} + b_{11}\cdot x_i(b_7,a_3)} \\
        &\qquad\qquad\qquad\qquad\qquad\qquad\qquad\times (-1)^{\left[p\left(b_2 , a_1\right)\right]_j + \left[p\left(b_5 , a_1\right)\right]_j + \left[p\left(b_8 , a_3\right)\right]_j + \left[p\left(b_{11} , a_3\right)\right]_j}\delta_{a_1,y_i(b_1,a_1)}\delta_{a_3,y_i(b_7,a_3)}\\
        &\quad + \mathbb{E}_{S} \frac{1}{4^{n+2k}} \sum_{b_1,b_2,b_3,b_5,b_8,b_{11}} \sum_{a} (-1)^{b_1\cdot b_2 + b_2\cdot b_3 + x_i(b_3,a)\cdot b_5 + b_5\cdot x_i(b_1,a) + b_1\cdot b_8 + b_8\cdot b_3 + x_i(b_3,a)\cdot b_{11} + b_{11}\cdot x_i(b_1,a)} \\
        & \qquad\qquad\qquad\qquad\qquad\qquad\qquad\times (-1)^{\left[p\left(b_2 , a_1\right)\right]_j+\left[p\left(b_5 , a_2\right)\right]_j+\left[p\left(b_8 , a_1\right)\right]_j+\left[p\left(b_{11} , a_2\right)\right]_j}\delta_{y_i(b_1,a),y_i(b_3,a)}\\
        &\quad + \mathbb{E}_{S} \frac{1}{4^{n+2k}} \sum_{b_1,b_2,b_5,b_8,b_{11}} \sum_{a} (-1)^{\left[p\left(b_2 , a\right)\right]_j+\left[p\left(b_5 , y_i(b_1,a)\right)\right]_j+\left[p\left(b_8 , a\right)\right]_j+\left[p\left(b_{11} , y_i(b_1,a)\right)\right]_j}\\
        &\quad + (\text{other $\delta_{a_i,a_j}\delta_{a_k,a_l}$ terms}) \\
        &\quad + \mathbb{E}_{S} \frac{1}{4^{n+2k}} \sum_{b_1,b_2,b_3,b_5,b_8,b_{11}} \sum_{a} (-1)^{\left[p\left(b_2 , a\right)\right]_j+\left[p\left(b_5 , a\right)\right]_j+\left[p\left(b_8 , a\right)\right]_j+\left[p\left(b_{11} , a\right)\right]_j+b_1\cdot b_2 + b_2\cdot b_3 + b_1\cdot b_5 + b_5\cdot b_3}\\
        & \quad + (\text{other three-point delta function terms of $\{a_i\}$})\\
        & \quad + (\text{four-point delta function terms of $\{a_i\}$})\\
        &\leq \frac{1}{2^{n/2-k}} + \frac{1}{2^{n-O(k)}} + \frac{1}{2^{n-O(k)}}+\cdots+\frac{1}{2^{2n-O(k)}}+\frac{1}{2^{n-O(k)}}+\frac{1}{2^{n-O(k)}}+\cdots+\frac{1}{2^{n-O(k)}}+\cdots.
    \end{split}
\end{equation}
The terms in the last line of Eq.~\eqref{eq:XZ-OTOC-Hada-Var} are upper bounds of each of terms of the previous lines. We again use \textbf{Lemma}~\ref{thm:parity-sum} to upper bound parity summations of $\{[p(b_l,a_m)]_j\}$ and \textbf{Lemma}~\ref{thm:average-bit-flip} to upper bound the sum of $\delta_{a,y_i(b,a)}$ over $a$. All terms in the last line decay exponentially in $n$. Therefore, again, the ensemble of $\{O_{VW}(U)\}$ over random functions and permutations is concentrated near its ensemble average value.\\

\noindent
\hypertarget{sec:XX-OTOCs}{\textbf{XX-OTOCs}}\\
\indent Here, we consider the case of $V=X_i$ and $W=X_j$ with $i\neq j$. Then, $O_{VW}(U)$ is given by
\begin{equation}
    \begin{split}
    O_{VW}(U)
    &=\frac{1}{2^n} \sum_{\left\{b_i\right\}_{i=1}^{8}} \sum_{\left\{a_i\right\}_{i=1}^4}(-1)^{\sum_{i=1}^{8} f\left(b_i ,a_{\lfloor(i+1)/2\rfloor}\right)}[V]_{p\left(b_8 , a_4\right), p\left(b_1 , a_1\right)}\left[U\right]_{b_1, b_2}[W]_{p\left(b_2 , a_1\right), p\left(b_3 , a_2\right)} \\
    &\qquad\qquad\qquad\qquad \times \left[U^\dagger\right]_{b_3, b_4} [V]_{p\left(b_4 , a_2\right), p\left(b_5 , a_3\right)}\left[U\right]_{b_5, b_6}[W]_{p\left(b_6 , a_3\right), p\left(b_7 , a_4\right)}\left[U^\dagger\right]_{b_7, b_8}\\
    & =\frac{1}{2^n} \sum_{\left\{b_i\right\}_{i=1}^8} \sum_{\left\{a_i\right\}_{i=1}^4}(-1)^{f\left(b_1 , a_1\right)+f\left(b_2 , a_1\right)+f\left(b_3 , a_2\right)+f\left(b_4 , a_2\right)+f\left(b_5 , a_3\right)+f\left(b_6 , a_3\right)+f\left(b_7 , a_4\right)+f\left(b_8 , a_4\right)} \\
    &\qquad\qquad\qquad\qquad \times \delta_{p\left(b_8 , a_4\right)+\hat{e}_i, p\left(b_1 , a_1\right)} \delta_{p\left(b_4 , a_2\right)+\hat{e}_i, p\left(b_5 , a_3\right)} \delta_{p\left(b_2 , a_1\right)+\hat{e}_j, p\left(b_3 , a_2\right)} \delta_{p\left(b_6 , a_3\right)+\hat{e}_j, p\left(b_7 , a_4\right)} \\
    &\qquad\qquad\qquad\qquad \times \left[U\right]_{b_1, b_2}\left[U^\dagger\right]_{b_3, b_4}\left[U\right]_{b_5, b_6}\left[U^\dagger\right]_{b_7, b_8}\\
    & =\frac{1}{2^n} \sum_{\left\{b_i\right\}_{i=1}^4} \sum_{\left\{a_i\right\}_{i=1}^2}(-1)^{f\left(b_1 , a_1\right)+f\left(b_2 , a_1\right)+f\left(x_i\left(b_1, a_1\right) , y_i\left(b_1, a_1\right)\right)+f\left(x_j\left(b_2, a_1\right) , y_j\left(b_2, a_1\right)\right)} \\
    &\qquad\qquad\qquad\qquad \times (-1)^{f\left(b_3 , a_2\right)+f\left(b_4 , a_2\right)+f\left(x_i\left(b_3, a_2\right) , y_i\left(b_3, a_2\right)\right)+f\left(x_j\left(b_4, a_2\right) , y_j\left(b_4, a_2\right)\right)}\\
    & \qquad\qquad\qquad\qquad\times \delta_{y_j\left(b_2, a_1\right), y_i\left(b_3, a_2\right)} \delta_{y_j\left(b_4, a_2\right), y_i\left(b_1, a_1\right)}\\
    & \qquad\qquad\qquad\qquad\times\left[U\right]_{b_1, b_2}\left[U^\dagger\right]_{x_j\left(b_2, a_1\right), x_i\left(b_3, a_2\right)}\left[U\right]_{b_3, b_4}\left[U^\dagger\right]_{x_j\left(b_4, a_2\right), x_i\left(b_1, a_1\right)}
    \end{split}
\end{equation}
Its ensemble average is given by the sum of the following terms:

\hfill \break
\noindent
- $\delta_{b_1, b_2} \delta_{b_3, b_4}$
\begin{equation}\label{eq:XX-OTOC-first}
    \mathbb{E}_S \frac{1}{2^n} \sum_{b_1, b_3} \sum_{\left\{a_i\right\}_{i=1}^2} \delta_{y_j\left(b_1, a_1\right), y_i\left(b_3, a_2\right)}\delta_{y_j\left(b_3, a_2\right), y_i\left(b_1, a_1\right)}[U]_{b_1, b_1}[U]_{b_3, b_3}\left[U^*\right]_{x_i\left(b_3, a_2\right), x_j\left(b_1, a_1\right)}\left[U^*\right]_{x_i\left(b_1, a_1\right), x_j\left(b_3, a_2\right)}
\end{equation}
- $\delta_{a_1, a_2} \delta_{b_1, b_3} \delta_{b_2, b_4}$
\begin{equation}
    \mathbb{E}_S \frac{1}{2^n} \sum_{b_1, b_2} \sum_a \delta_{y_j\left(b_2, a\right), y_i\left(b_1, a\right)}[U]_{b_1, b_2}[U]_{b_1, b_2}\left[U^*\right]_{x_i\left(b_1, a\right), x_j\left(b_2, a\right)}\left[U^*\right]_{x_i\left(b_1, a\right), x_j\left(b_2, a\right)}    
\end{equation}
- $\delta_{a_1, a_2} \delta_{b_1, b_4} \delta_{b_2, b_3}$
\begin{equation}
    \mathbb{E}_S \frac{1}{2^n} \sum_{b_1, b_2} \sum_a \delta_{y_j\left(b_2, a\right), y_i\left(b_2, a\right)} \delta_{y_j\left(b_1, a\right), y_i\left(b_1, a\right)}[U]_{b_1, b_2}[U]_{b_2, b_1}\left[U^*\right]_{x_i\left(b_2, a\right), x_j\left(b_2, a\right)}\left[U^*\right]_{x_i\left(b_1, a\right), x_j\left(b_1, a\right)}
\end{equation}
- $\delta_{a_1, a_2} \delta_{b_1, b_3} \delta_{p(b_2,a_1),p(b_4,a_2)+\hat{e}_j}$
\begin{equation}
    \mathbb{E}_S\frac{1}{2^n} \sum_{b_1,b_2} \sum_{a} \delta_{a, y_j\left(b_2, a\right)} \delta_{a, y_i\left(b_1, a\right)} \left[U\right]_{b_1, b_2}\left[U\right]_{b_1, x_j(b_2,a)}\left[U^*\right]_{x_i\left(b_1, a\right), x_j\left(b_2, a\right)}\left[U^*\right]_{x_i\left(b_1, a\right),b_2}
\end{equation}
- $\delta_{a_1, a_2} \delta_{b_1, b_4} \delta_{p(b_2,a_1),p(b_3,a_2)+\hat{e}_i}$
\begin{equation}
    \mathbb{E}_S\frac{1}{2^n} \sum_{b_1,b_2} \sum_{a} \delta_{y_j\left(b_2, a\right), a} \delta_{y_j\left(b_1, a\right), y_i\left(b_1, a\right)} \left[U\right]_{b_1, b_2}\left[U\right]_{x_i(b_2,a), b_1}\left[U^*\right]_{b_2, x_j\left(b_2, a\right)}\left[U^*\right]_{x_i\left(b_1, a\right),x_j\left(b_1, a\right)}
\end{equation}
- $\delta_{a_1, a_2} \delta_{b_2, b_3} \delta_{p(b_1,a_1),p(b_4,a_2)+\hat{e}_j}$
\begin{equation}
    \mathbb{E}_S\frac{1}{2^n} \sum_{b_1,b_2} \sum_{a} \delta_{y_j\left(b_2, a\right), y_i\left(b_2, a\right)} \delta_{a, y_i\left(b_1, a\right)} \left[U\right]_{b_1, b_2}\left[U\right]_{b_2, x_j(b_1,a)}\left[U^*\right]_{x_i\left(b_2, a\right), x_j\left(b_2, a\right)}\left[U^*\right]_{x_i\left(b_1, a\right),b_1}
\end{equation}
- $\delta_{a_1, a_2} \delta_{b_2, b_4} \delta_{p(b_1,a_1),p(b_3,a_2)+\hat{e}_i}$
\begin{equation}
    \mathbb{E}_S\frac{1}{2^n} \sum_{b_1,b_2} \sum_{a} \delta_{y_j\left(b_2, a\right), a} \delta_{y_j\left(b_2, a\right), y_i\left(b_1, a\right)} \left[U\right]_{b_1, b_2}\left[U\right]_{x_i(b_1,a), b_2}\left[U^*\right]_{b_1, x_j\left(b_2, a\right)}\left[U^*\right]_{x_i\left(b_1, a\right),x_j\left(b_2, a\right)}
\end{equation}
- $\delta_{p\left(b_1, a_1\right), p\left(b_2, a_1\right)+\hat{e}_j} \delta_{p\left(b_2, a_1\right), p\left(b_3, a_2\right)+\hat{e}_i} \delta_{p\left(b_3, a_2\right), p\left(b_4, a_2\right)+\hat{e}_j} \delta_{p\left(b_4, a_2\right), p\left(b_1, a_1\right)+\hat{e}_i}$
\begin{equation}
    \mathbb{E}_S \frac{1}{2^n} \sum_b \sum_a \delta_{a, y_j(b, a)} \delta_{y_i(b, a), y_{i, j}(b, a)}[U]_{b, x_j(b, a)}[U]_{x_{i, j}(b, a), x_i(b, a)}\left[U^*\right]_{x_j(b, a), b}\left[U^*\right]_{x_i(b, a), x_{i, j}(b, a)}
\end{equation}
- $\delta_{p\left(b_1, a_1\right), p\left(b_2, a_1\right)+\hat{e}_j} \delta_{p\left(b_2, a_1\right), p\left(b_4, a_2\right)+\hat{e}_j} \delta_{p\left(b_4, a_2\right), p\left(b_3, a_2\right)+\hat{e}_i} \delta_{p\left(b_3, a_2\right), p\left(b_1, a_1\right)+\hat{e}_i}$
\begin{equation}
    \mathbb{E}_S \frac{1}{2^n} \sum_b \sum_a \delta_{a, y_i(b, a)} \delta_{a, y_j(b, a)}[U]_{b, x_j(b, a)}[U]_{x_i(b, a), b}\left[U^*\right]_{b, b}\left[U^*\right]_{x_i(b, a), x_j(b, a)}
\end{equation}
- $\delta_{p\left(b_1, a_1\right), p\left(b_4, a_2\right)+\hat{e}_j} \delta_{p\left(b_4, a_2\right), p\left(b_3, a_2\right)+\hat{e}_i} \delta_{p\left(b_3, a_2\right), p\left(b_2, a_1\right)+\hat{e}_j} \delta_{p\left(b_2, a_1\right), p\left(b_1, a_1\right)+\hat{e}_i}$
\begin{equation}
    \mathbb{E}_S \frac{1}{2^n} \sum_b \sum_a \delta_{a, y_i(b, a)} \delta_{y_j(b, a), y_{i, j}(b, a)}[U]_{b, x_i(b, a)}[U]_{x_{i, j}(b, a), x_j(b, a)}\left[U^*\right]_{x_j(b, a), x_{i, j}(b, a)}\left[U^*\right]_{x_i(b, a), b}
\end{equation}
- $\delta_{p\left(b_1, a_1\right), p\left(b_4, a_2\right)+\hat{e}_j} \delta_{p\left(b_4, a_2\right), p\left(b_2, a_1\right)+\hat{e}_j} \delta_{p\left(b_2, a_1\right), p\left(b_3, a_2\right)+\hat{e}_i} \delta_{p\left(b_3, a_2\right), p\left(b_1, a_1\right)+\hat{e}_i}$
\begin{equation}\label{eq:XX-OTOC-last}
    \mathbb{E}_S \frac{1}{2^n} \sum_b \sum_a \delta_{a, y_i(b, a)} \delta_{a, y_j(b, a)}[U]_{b, b}[U]_{x_i(b, a), x_j(b, a)}\left[U^*\right]_{b, x_j(b, a)}\left[U^*\right]_{x_i(b, a), b}
\end{equation}

Below, we compute upper bounds of each of terms in Eq.~\eqref{eq:XX-OTOC-first}-\ref{eq:XX-OTOC-last} for embedded Hadamard gate. Consequently, all terms are negligible in $n$.
\hfill \break
\noindent
- $\delta_{b_1, b_2} \delta_{b_3, b_4}$
\begin{equation}
    \begin{split}
        &\mathbb{E}_S \frac{1}{2^{n+2 k}} \sum_{b_1, b_3} \sum_{a_1, a_2}\delta_{y_j\left(b_1, a_1\right), y_i\left(b_3, a_2\right)}\delta_{y_j\left(b_3, a_2\right), y_i\left(b_1, a_1\right)}(-1)^{b_1 \cdot b_1+b_3 \cdot b_3+x_i\left(b_3, a_2\right) \cdot x_j\left(b_1, a_1\right)+x_i\left(b_1, a_1\right) \cdot x_j\left(b_3, a_2\right)} \\
        &\quad\leq \frac{1}{2^{n+2 k}} \sum_a \sum_{y_j\left(b_1, a_1\right)=a} \sum_{y_i\left(b_3, a_2\right)=a}1 = \frac{1}{2^k} \\
    \end{split}
\end{equation}
The last line comes from the fact that the size of $\{b_1,a_1 | y_j(b_1,a_1)=a\}$ is exactly $2^k$.

\noindent
- $\delta_{a_1, a_2} \delta_{b_1, b_3} \delta_{b_2, b_4}$
\begin{equation}
    \mathbb{E}_S \frac{1}{2^{n+2 k}} \sum_{b_1, b_2} \sum_a \delta_{y_j\left(b_2, a\right), y_i\left(b_1, a\right)}\leq \frac{1}{2^{n+2 k}} \sum_{b_1, b_2} \sum_a 1 = \frac{1}{2^k}
\end{equation}
- $\delta_{a_1, a_2} \delta_{b_1, b_4} \delta_{b_2, b_3}$
\begin{equation}
    \mathbb{E}_S \frac{1}{2^{n+2 k}} \sum_{b_1, b_2} \sum_a \delta_{y_j\left(b_2, a\right), y_i\left(b_2, a\right)} \delta_{y_j\left(b_1, a\right), y_i\left(b_1, a\right)}(-1)^{x_i\left(b_2, a\right) \cdot x_j\left(b_2, a\right)+x_i\left(b_1, a\right) \cdot x_j\left(b_1, a\right)} \leq \frac{1}{2^k}
\end{equation}
- $\delta_{a_1, a_2} \delta_{b_1, b_3} \delta_{p(b_2,a_1),p(b_4,a_2)+\hat{e}_j}$
\begin{equation}
    \begin{split}
        &\mathbb{E}_S\frac{1}{2^{n+2k}} \sum_{b_1,b_2} \sum_{a} \delta_{a, y_j\left(b_2, a\right)} \delta_{a, y_i\left(b_1, a\right)} (-1)^{b_1\cdot b_2 + b_1\cdot x_j(b_2,a) + x_i\left(b_1, a\right)\cdot x_j\left(b_2, a\right)+x_i\left(b_1, a\right)\cdot x_j\left(b_2, a\right)}\\
        &\quad\leq \frac{1}{2^{n+2k}} \sum_{b_1,b_2} \sum_{a} 1 = \frac{1}{2^k}        
    \end{split}
\end{equation}
- $\delta_{p\left(b_1 a_1\right), p\left(b_2 a_1\right)+\hat{e}_j} \delta_{p\left(b_2 a_1\right), p\left(b_3 a_2\right)+\hat{e}_i} \delta_{p\left(b_3 a_2\right), p\left(b_4 a_2\right)+\hat{e}_j} \delta_{p\left(b_4 a_2\right), p\left(b_1 a_1\right)+\hat{e}_i}$
\begin{equation}
    \mathbb{E}_S \frac{1}{2^{n+2 k}} \sum_b \sum_a \delta_{a, y_j(b, a)} \delta_{y_i(b, a), y_{i, j}(b, a)} \leq \frac{1}{2^{n+2 k}} \sum_b \sum_a 1 = \frac{1}{2^{2k}}
\end{equation}
- $\delta_{p\left(b_1 a_1\right), p\left(b_2 a_1\right)+\hat{e}_j} \delta_{p\left(b_2 a_1\right), p\left(b_4 a_2\right)+\hat{e}_j} \delta_{p\left(b_4 a_2\right), p\left(b_3 a_2\right)+\hat{e}_i} \delta_{p\left(b_3 a_2\right), p\left(b_1 a_1\right)+\hat{e}_i}$
\begin{equation}
    \mathbb{E}_S \frac{1}{2^{n+2k}} \sum_b \sum_a \delta_{a, y_i(b, a)} \delta_{a, y_j(b, a)}\leq \frac{1}{2^{n+2k}} \sum_b \sum_a = \frac{1}{2^{2k}}
\end{equation}

Next, we compute the ensemble average of the square of $O_{VW}(U)$. 
\begin{equation}
    \begin{split}
    \mathbb{E}_{S,f}\left[\left(O_{VW}(U)\right)^2\right]
    & =\frac{1}{4^n} \sum_{\left\{b_i\right\}_{i=1}^8} \sum_{\left\{a_i\right\}_{i=1}^4}(-1)^{f\left(b_1 , a_1\right)+f\left(b_2 , a_1\right)+f\left(x_i\left(b_1, a_1\right) , y_i\left(b_1, a_1\right)\right)+f\left(x_j\left(b_2, a_1\right) , y_j\left(b_2, a_1\right)\right)} \\
    &\qquad\qquad\qquad\quad \times (-1)^{f\left(b_3 , a_2\right)+f\left(b_4 , a_2\right)+f\left(x_i\left(b_3, a_2\right) , y_i\left(b_3, a_2\right)\right)+f\left(x_j\left(b_4, a_2\right) , y_j\left(b_4, a_2\right)\right)}\\
    &\qquad\qquad\qquad\quad\times (-1)^{f\left(b_5 , a_3\right)+f\left(b_6 , a_3\right)+f\left(x_i\left(b_5, a_3\right) , y_i\left(b_5, a_3\right)\right)+f\left(x_j\left(b_6, a_3\right) , y_j\left(b_6, a_3\right)\right)} \\
    &\qquad\qquad\qquad\quad \times (-1)^{f\left(b_7 , a_4\right)+f\left(b_8 , a_4\right)+f\left(x_i\left(b_7, a_4\right) , y_i\left(b_7, a_4\right)\right)+f\left(x_j\left(b_8, a_4\right) , y_j\left(b_8, a_4\right)\right)}\\
    & \qquad\qquad\qquad\quad\times \delta_{y_j\left(b_2, a_1\right), y_i\left(b_3, a_2\right)} \delta_{y_j\left(b_4, a_2\right), y_i\left(b_1, a_1\right)}\\
    & \qquad\qquad\qquad\quad\times \delta_{y_j\left(b_6, a_3\right), y_i\left(b_7, a_4\right)} \delta_{y_j\left(b_8, a_4\right), y_i\left(b_5, a_3\right)}\\
    & \qquad\qquad\qquad\quad\times\left[U\right]_{b_1, b_2}\left[U^\dagger\right]_{x_j\left(b_2, a_1\right), x_i\left(b_3, a_2\right)}\left[U\right]_{b_3, b_4}\left[U^\dagger\right]_{x_j\left(b_4, a_2\right), x_i\left(b_1, a_1\right)}\\
    & \qquad\qquad\qquad\quad\times\left[U\right]_{b_5, b_6}\left[U^\dagger\right]_{x_j\left(b_6, a_3\right), x_i\left(b_7, a_4\right)}\left[U\right]_{b_7, b_8}\left[U^\dagger\right]_{x_j\left(b_8, a_4\right), x_i\left(b_5, a_3\right)}
    \end{split}
\end{equation}

If this is negligible in $n$, then the variance of $O_{VW}(U)$ over the ensemble is also negligible in $n$. Now, let us compute this explicitly.

\begin{equation}
    \begin{split}
    \mathbb{E}_{S,f}\left[O_{VW}^2\right]
    & =\mathbb{E}_{S,f}\frac{1}{4^{n+2k}} \sum_{\left\{b_i\right\}_{i=1}^8} \sum_{\left\{a_i\right\}_{i=1}^4}(-1)^{f\left(b_1 , a_1\right)+f\left(b_2 , a_1\right)+f\left(x_i\left(b_1, a_1\right) , y_i\left(b_1, a_1\right)\right)+f\left(x_j\left(b_2, a_1\right) , y_j\left(b_2, a_1\right)\right)} \\
    &\qquad\qquad\qquad\qquad\quad \times (-1)^{f\left(b_3 , a_2\right)+f\left(b_4 , a_2\right)+f\left(x_i\left(b_3, a_2\right) , y_i\left(b_3, a_2\right)\right)+f\left(x_j\left(b_4, a_2\right) , y_j\left(b_4, a_2\right)\right)}\\
    &\qquad\qquad\qquad\qquad\quad\times (-1)^{f\left(b_5 , a_3\right)+f\left(b_6 , a_3\right)+f\left(x_i\left(b_5, a_3\right) , y_i\left(b_5, a_3\right)\right)+f\left(x_j\left(b_6, a_3\right) , y_j\left(b_6, a_3\right)\right)} \\
    &\qquad\qquad\qquad\qquad\quad \times (-1)^{f\left(b_7 , a_4\right)+f\left(b_8 , a_4\right)+f\left(x_i\left(b_7, a_4\right) , y_i\left(b_7, a_4\right)\right)+f\left(x_j\left(b_8, a_4\right) , y_j\left(b_8, a_4\right)\right)}\\
    & \qquad\qquad\qquad\qquad\quad\times \delta_{y_j\left(b_2, a_1\right), y_i\left(b_3, a_2\right)} \delta_{y_j\left(b_4, a_2\right), y_i\left(b_1, a_1\right)}\delta_{y_j\left(b_6, a_3\right), y_i\left(b_7, a_4\right)} \delta_{y_j\left(b_8, a_4\right), y_i\left(b_5, a_3\right)}\\
    & \qquad\qquad\qquad\qquad\quad\times (-1)^{b_1\cdot b_2 + x_j\left(b_2, a_1\right)\cdot x_i\left(b_3, a_2\right)+b_3\cdot b_4+x_j\left(b_4, a_2\right)\cdot x_i\left(b_1, a_1\right)}\\
    & \qquad\qquad\qquad\qquad\quad\times (-1)^{b_5\cdot b_6+x_j\left(b_6, a_3\right)\cdot x_i\left(b_7, a_4\right)+b_7\cdot b_8+x_j\left(b_8, a_4\right)\cdot x_i\left(b_5, a_3\right)}\\
    & =\mathbb{E}_{S}\frac{1}{4^{n+2k}} \sum_{\left\{b_i\right\}_{i=1}^8} \sum_{\left\{a_i\right\}_{i=1}^4} \delta_{y_j\left(b_2, a_1\right), y_i\left(b_3, a_2\right)} \delta_{y_j\left(b_4, a_2\right), y_i\left(b_1, a_1\right)}\delta_{y_j\left(b_6, a_3\right), y_i\left(b_7, a_4\right)} \delta_{y_j\left(b_8, a_4\right), y_i\left(b_5, a_3\right)}\\
    &\qquad\qquad\times (-1)^{b_1\cdot b_2 + x_j\left(b_2, a_1\right)\cdot x_i\left(b_3, a_2\right)+b_3\cdot b_4+x_j\left(b_4, a_2\right)\cdot x_i\left(b_1, a_1\right)}\\
    &\qquad\qquad\times (-1)^{b_5\cdot b_6+x_j\left(b_6, a_3\right)\cdot x_i\left(b_7, a_4\right)+b_7\cdot b_8+x_j\left(b_8, a_4\right)\cdot x_i\left(b_5, a_3\right)}\\
    &\qquad\qquad\times [(\delta_{b_1,b_2}\delta_{b_3,b_4}+ \delta_{a_1,a_2}(\delta_{b_1,b_3}\delta_{b_2,b_4}+\delta_{b_1,b_4}\delta_{b_2,b_3}+\delta_{b_1, b_3}\delta_{p(b_2,a_1),p(b_4,a_2)+\hat{e}_j}\\
    &\qquad\qquad\qquad +\delta_{b_1, b_4}\delta_{p(b_2,a_1),p(b_3,a_2)+\hat{e}_j}+\delta_{b_2, b_3}\delta_{p(b_1,a_1),p(b_4,a_2)+\hat{e}_j}+\delta_{b_2, b_4}\delta_{p(b_1,a_1),p(b_3,a_2)+\hat{e}_j})\\
    &\qquad\qquad\qquad+\delta_{p\left(b_1 a_1\right), p\left(b_2 a_1\right)+\hat{e}_j} \delta_{p\left(b_2 a_1\right), p\left(b_3 a_2\right)+\hat{e}_i} \delta_{p\left(b_3 a_2\right), p\left(b_4 a_2\right)+\hat{e}_j} \delta_{p\left(b_4 a_2\right), p\left(b_1 a_1\right)+\hat{e}_i}+\cdots)\\
    &\qquad\qquad\qquad\times\text{(cross contraction terms of $\{a_3,a_4\}$)}\\
    &\qquad\qquad\quad+\text{(cross contraction terms of other $\{a_k,a_l\}\times\{a_q,a_r\}$)}\\
    &\qquad\qquad\quad+\delta_{a_1,a_2,a_3}\delta_{b_7,b_8}(\delta_{b_1,b_3}\delta_{b_4,b_5}\delta_{b_6,b_2}+\text{(permutations)})\\
    &\qquad\qquad\quad+\delta_{a_1,a_2}\delta_{b_7,b_8}(\delta_{b_1,b_3}\delta_{p(b_2,a_1),p(b_5,a_3)+\hat{e}_i}\delta_{p(b_4,a_2),p(b_6,a_3)+\hat{e}_j}+\text{(other choices of $\{b_i,b_j\}$)})\\
    &\qquad\qquad\quad+\delta_{b_7,b_8}(\delta_{p(b_1,a_1),p(b_3,a_2)+\hat{e}_i}\delta_{p(b_2,a_1),p(b_6,a_3)+\hat{e}_j}\delta_{p(b_3,a_2),p(b_5,a_3)+\hat{e}_i}\delta_{p(b_4,a_2),p(b_2,a_1)+\hat{e}_j}\\
    &\qquad\qquad\qquad\times\delta_{p(b_5,a_3),p(b_1,a_1)+\hat{e}_i}\delta_{p(b_6,a_3),p(b_4,a_2)+\hat{e}_j}+\text{(permutations)})\\
    &\qquad\qquad\quad+\text{(cross contraction terms of other $\{a_k,a_l,a_q\}\times\{a_r\}$)}\\
    &\qquad\qquad\quad + (\text{fourth-order cross contraction terms})]\\
    \end{split}
\end{equation}

Here, the contraction means the constraint on two random functions $f(b_1,a_1)$ and $f(b_2,a_2)$ having the same arguments, $b_1=b_2$ and $a_1=a_2$, such that the average of $(-1)^{f(b_1,a_1)+f(b_2,a_2)}$ over random functions $f$ is non-vanishing. In addition, a cross contraction of $\{a_k,a_l\}$ means a contraction of $f(b_1,a_k)$ and bit-flipped one $f(x_i(b_2,a_l),y_i(b_2,a_l))$. A cross contraction of $\{a_k,a_l\}\times\{a_q,a_r\}$ means a pair of contractions of $f(b_1,a_k)$ with $f(x_i(b_2,a_l),y_i(b_2,a_l))$ and $f(b_3,a_q)$ with $f(x_i(b_4,a_r),y_i(b_4,a_r))$. Finally, a cross contraction of $\{a_k,a_l,a_q\}\times\{a_r\}$ is a simultaneous contraction of three random functions with at least one random function with and without bit-flipping. Cross contraction terms of $\{a_i,a_j\}\times\{a_k,a_l\}$ are products of terms in $\mathbb{E}_{S,f}\left[O_{VW}\right]$. Thus, we only need to consider higher order cross contraction terms such as of $\{a_1,a_2,a_3\}\times\{a_4\}$ to compute its variance. Below, we neglect sign factors of due to Hadamard gate since it does not affect the magnitude of each term. One may see that all terms decay exponentially in $n$, and thus, the variance also decays exponentially.

\noindent
- $\delta_{a_1,a_2,a_3}\delta_{b_7, b_8}\delta_{b_1, b_3}\delta_{b_4, b_5}\delta_{b_2,b_6}$
\begin{equation}
    \begin{split}
    & \mathbb{E}_S \frac{1}{4^{n+2k}} \sum_{b_1,b_2,b_4,b_7} \sum_{a_1,a_4} \delta_{y_j\left(b_2, a_1\right), y_i\left(b_3, a_1\right)} \delta_{y_j\left(b_4, a_1\right), y_i\left(b_1, a_1\right)}\delta_{y_j\left(b_2, a_1\right), y_i\left(b_7, a_4\right)} \delta_{y_j\left(b_7, a_4\right), y_i\left(b_4, a_1\right)}\\
    &\leq \mathbb{E}_S \frac{1}{4^{n+2k}} \sum_{b_1,b_2,b_4,b_7} \sum_{a_1,a_4} \delta_{y_j\left(b_7, a_4\right), y_i\left(b_4, a_1\right)}\\
    &= \frac{1}{2^{2n+3k}} \sum_{b_1,b_2,b_4} \sum_{a_1} 1 \\
    &= \frac{1}{2^{n+k}}\\
    \end{split}
\end{equation}
- $\delta_{a_1,a_2}\delta_{b_7, b_8}\delta_{b_1, b_3}\delta_{p(b_2,a_1),p(b_5,a_3)+\hat{e}_i}\delta_{p(b_4,a_2),p(b_6,a_3)+\hat{e}_j}$
\begin{equation}
    \begin{split}
    & \mathbb{E}_S \frac{1}{4^{n+2k}} \sum_{b_1,b_2,b_4,b_7} \sum_{a_1,a_4} \delta_{y_j\left(b_2, a_1\right), y_i\left(b_1, a_1\right)} \delta_{y_j\left(b_4, a_1\right), y_i\left(b_1, a_1\right)}\delta_{a_1, y_i\left(b_7, a_4\right)} \delta_{y_j\left(b_7, a_4\right), a_1}\delta_{y_i(b_2,a_1),y_j(b_4,a_1)}\\
    &\leq \mathbb{E}_S\frac{1}{4^{n+2k}} \sum_{b_1,b_2,b_4,b_7} \sum_{a,a_1,a_4} \delta_{a_1, y_i\left(b_7, a_4\right)} \delta_{a,y_i(b_2,a_1)}\\
    &= \mathbb{E}_S\frac{1}{2^{2n+3k}} \sum_{b_1,b_2,b_4} \sum_{a,a_1} \delta_{a,y_i(b_2,a_1)}\\
    &= \frac{1}{2^{2n+3k}} \sum_{b_1,b_4} \sum_{a} 1\\
    &= \frac{1}{2^{n+2k}}
    \end{split}
\end{equation}
- $\delta_{b_7,b_8}\delta_{p(b_1,a_1),p(b_3,a_2)+\hat{e}_i}\delta_{p(b_2,a_1),p(b_6,a_3)+\hat{e}_j}\delta_{p(b_3,a_2),p(b_5,a_3)+\hat{e}_i}\delta_{p(b_4,a_2),p(b_2,a_1)+\hat{e}_j}\delta_{p(b_5,a_3),p(b_1,a_1)+\hat{e}_i}\delta_{p(b_6,a_3),p(b_4,a_2)+\hat{e}_j}$
\begin{equation}
    \begin{split}
    & \mathbb{E}_S \frac{1}{4^{n+2k}} \sum_{\left\{b_i\right\}_{i=1}^7} \sum_{a_1,a_4} \delta_{y_j\left(b_2, a_1\right), a_1} \delta_{a_1, y_i\left(b_1, a_1\right)}\delta_{y_i(b_1,a_1), y_i\left(b_7, a_4\right)} \delta_{y_j\left(b_7, a_4\right), a_1}\\
    &\leq \mathbb{E}_S \frac{1}{4^{n+2k}} \sum_{\left\{b_i\right\}_{i=1}^7} \sum_{a,a_1,a_4}\delta_{a, y_i(b_1,a_1)} \delta_{a, y_i\left(b_7, a_4\right)} \delta_{a_1, y_j\left(b_7, a_4\right)}\\
    &\leq \frac{1}{2^{2n+2k}} \sum_{\left\{b_i\right\}_{i=2}^6} \sum_{a}1\\
    &\leq \frac{1}{2^{n-2k}}
    \end{split}
\end{equation}

\newpage
\hypertarget{sec:micro-time}{}
\subsection*{5. {Microscopic time} behavior of OTOCs}
{In this subsection, we study behavior of $\mathbb{E}_{S,f}\left[C_{VW}(U(t))\right]$ of an RSED in the microscopic time scale. Here, the microscopic time scale refers a time window allowing the perturbation expansion of a time evolution operator.} For simplicity, we focus on the case with $V=Z_i$ and $W=Z_j$ with $i\neq j$. Let $h$ be the embedded Hamiltonian. Then, at early time $t=\epsilon\ll 1$, $\mathbb{E}_{S,f}\left[C_{VW}(U(t))\right]$ is given by
\begin{equation}\label{eq:OTOC-early-slope}
    \begin{split}
        \mathbb{E}_{S,f}\left[C_{VW}(U(t))\right]
        &= 1 - \real \left[ \tr\left(\left(U . * U . * U^*\right) U^{\dagger}\right) \right]\\
        &= 1 - \real \left[ \tr\left(\left[e^{-i \epsilon h} . * e^{-i \epsilon h} . * e^{i \epsilon h^*}\right] e^{i \epsilon h}\right) \right]\\
        &= \frac{\epsilon^2}{2}\tr\left(3\operatorname{diag}\left(h^2\right)+ h^2-2\operatorname{diag}\left(h . * h^*\right)-2\operatorname{diag}(h) h\right) + O(\epsilon^3).
    \end{split}
\end{equation}
Here, $\operatorname{diag}(h)$ is $h$ with vanishing off-diagonal terms. At the third line, we use the Taylor expansion of $e^{\pm i\epsilon  H}$. Let us compute the maximum value of the initial slope over all Hermitian matrices having energies in $[-1,1]$. The first two terms of the last line of Eq.~\eqref{eq:OTOC-early-slope} is maximized by setting $H^2=I$. Let $h_D$ and $h_O$ be the diagonal and off-diagonal parts of $h$, respectively, then the remaining terms become
\begin{equation}
    \tr\left(\operatorname{diag}(h . * h)+\operatorname{diag}(h) h\right)=2\tr\left( h_D^2+h_D h_O\right).
\end{equation}
This can be minimized by setting $h_D=0$. These indicate that a hamiltonian with off-diagonal terms larger than diagonal terms grow OTOCs quickly at early time. An example of $h$ that maximizes the initial growth rate is $h=X$. In this case, $\mathbb{E}_S\left[C_{VW}(U(t))\right]$ is given by $2 t^2$ at early time. 

\newpage
{
\subsection*{6. Early-time behavior of OTOCs and Lyanpnov exponent}
In this subsection, we discuss early time behavior of $\mathbb{E}_{S,f}[C_{VW}(U(t)]$ of an RSED, specifically regarding its Lyanpnov exponent. 

The Lyapunov exponent of a quantum system is typically defined within a specific time window—commonly referred to in the literature as the ``early-time regime.'' This lies between the microscopic timescale, where the OTOC grows as a power law as discussed in \hyperlink{sec:micro-time}{Supplementary Note 5}, and the intermediate regime, where the OTOC begins to saturate~\cite{Garc_a_Mata_2023}. The onset of the intermediate regime is characterized by a timescale \( t_* \), which corresponds to the time at which the commutator \( [e^{-iHt} W e^{iHt}, H] \) between the Hamiltonian \( H \) and a local operator \( W \) becomes significantly large~\cite{Maldacena2016chaos}. In systems governed by local Hamiltonians, this timescale is given by \( t_* = |\mathbf{x}| / v_B \), where \( |\mathbf{x}| \) is the spatial separation between the local operators in the OTOC, and \( v_B \) is the butterfly velocity of the system~\cite{Khemani2018}.

However, in the case of a \textit{nonlocal} Hamiltonian—as considered in our work—the butterfly velocity \( v_B \) diverges, leading to a vanishing timescale \( t_* \). In other words, the commutator \( [e^{-iHt} W e^{iHt}, H] \) can become large within an \( o(1) \) timescale, during which the perturbative expansion of \( e^{-iHt} W e^{iHt} \) remains valid and the OTOC still grows as a power law, as demonstrated in the previous section. As a result, there is no well-defined time window in which the OTOC exhibits exponential growth, rendering the notion of a Lyapunov exponent inapplicable in this setting.
}

\newpage
\subsection*{7. Numerical calculation of OTOCs of many-body system}

In this section, we numerically study OTOCs of three different embedded hamiltonians: a parent Hamiltonian of the Hadamard gate, parent Hamiltonians of random sign Hadamard gates, and the Sachdev–Ye–Kitaev (SYK) model in terms of Pauli operators~\cite{Hanada2024}. For a given gate $U$, we take its paranent hamiltonian $h$ as $(i/2\pi)\log U$. Random sign Hadamard gates are given by the products of the Hadamard gate and diagonal random sign operators. \\

\noindent
\textbf{Hadamard gate}\\
\indent Let us first discuss OTOCs of a RSED with the products of the Hadamard gate as the embedded unitary operator $u$. \hyperlink{fig:Hada}{Supplementary Figure 1} shows OTOCs of this dynamics, which clearly has the periodicity of $T=1$. The periodic nature of this dynamics comes from the equal level spacing of the Hamiltonian. Although its OTOCs saturate with negligible errors of $O(2^{-k})$ at certain times, it cannot be considered as a pseudochaotic dynamics, as the saturation does not persist. To make the saturation persist, its spectrum should be randomized.\\

\hypertarget{sec:random-sign-Hadamard}
\noindent
\textbf{Random-sign Hadamard gate}\\
\indent The spectrum of the products of the Hadamard gates $H^{\otimes k}$ can be randomized by multiplying a diagonal random-sign matrix $[P]_{b,b}=(-1)^{\phi(b)}$ for some random function $\phi:\{0,1\}^k\rightarrow\{0,1\}$. Specifically, we consider $U(\tau)$ embedding $u^\tau = (H^{\otimes k}P)^\tau$ for some positive time $\tau>0$. The matrix elements of $u^\tau$ at an integer $\tau=m$ are given by
\begin{equation}
    (H^{\otimes k}P)^m \ket{p(b_0,a)} = \sum_b c_{b,b0} (-1)^{f(b_0,a)+f(b,a)}\ket{p(b,a)}
\end{equation}
with 
\begin{equation}
    c_{b,b_0} = \frac{1}{2^{mk/2}}\sum_{\{b_i\}_{i=1}^{m-1}}(-1)^{\sum_{i=1}^m[\phi(b_i)+b_{i-1}\cdot b_i]}
    \quad\text{and}\quad
    b_m=b.
\end{equation}
Here, $k$ is the subsystem size. Since $c_{b,b_0}$ is given by the mean value of random signs, its distribution approaches to the standard normal distribution. This can be directly seen from the level statistics of $u=H^{\otimes k}P$ in \hyperlink{fig:HadaP-stat}{Supplementary Figure 2}, which fits with the Wigner-Dyson distribution for the Gaussian orthogonal ensemble {except the peak at $\Delta E=0$ due to exponential degeneracies originated from the identity operator acting on remaining $n-k$ qubits}. This immediately implies that $u$ is not chaotic even within the subspace, which should follow the Gaussian unitary ensemble. This is consistent with the fact that the first $o(K)$ coordinates of a random vector on the unit $K$-sphere follows the standard normal distribution~\cite{Diaconis1987}. If we set $K=2^k$, then $c_{b,b_0}$ has the mean magnitude of $O(2^{-k/2})$, and thus, OTOCs of $u^\tau$ saturates to unity with a negligible error as demonstrated in \hyperlink{fig:HadaP}{Supplementary Figure 3} and due to \textbf{Theorem}~\ref{thm:negl-otoc-supp}. 

At the same time, $u^\tau$ can be used to generate a pseudorandom state ensemble. This is because $c_{b,b_0}$ has a negligible magnitude with a failure probability negligible in $n$. First, since $c_{b,b_0}$ follows a normal distribution, its absolute value square is bounded by the following inequality:
\begin{equation}
    \operatorname{Prob}\left[\abs{c_{b,b_0}}^2 > a K^{-1}\right] \approx e^{-a/2}
\end{equation}
for some constant $a$. Let us set $a=K^{1-\epsilon}$ for some $\epsilon>0$. Then, we have
\begin{equation}
    \operatorname{Prob}\left[\abs{c_{b,b_0}}^2 > K^{-\epsilon}\right] \approx e^{-K^{1-\epsilon}/2} = \operatorname{negl}(n).
\end{equation}
This proves our claim. Then, due to \textbf{Theorem}~\ref{thm:computation-prob-for-randomness}, it generates a pseudorandom state ensemble.\\

\noindent
\textbf{Pauli SYK model}\\
\indent Finally, let us consider the Sachdev–Ye–Kitaev (SYK) model, which is the prototypical model exhibiting quantum chaos. Here, we specifically consider its spin version~\cite{Hanada2024}. The Hamiltonian of the Pauli SYK model is given by~\cite{Hanada2024}
\begin{equation}\label{eq:Pauli-SYK}
    h' = \mathcal{N}(N_\mathrm{maj}) \sum_{1\leq a<b<c<d\leq N_\text{maj}} J_{abcd}i^{\eta_{abcd}}\chi_a\chi_b\chi_c\chi_d
\end{equation}
with
\begin{equation}
    \chi_{2n-1} = X_n, \qquad \chi_{2n} = Y_n, \qquad \eta_{abcd} = \delta_{\lceil \frac{a}{2} \rceil,\lceil \frac{b}{2} \rceil}\delta_{\lceil \frac{b}{2} \rceil,\lceil \frac{c}{2} \rceil}\delta_{\lceil \frac{c}{2} \rceil,\lceil \frac{d}{2} \rceil},
\end{equation}
random variables $\{J_{abcd}\}$ following the standard normal distribution, and a positive constant $ \mathcal{N}$ that determines the overall energy scale of $h'$. The energy scale is normalized so that the maximum and minimum energies are $\epsilon_\mathrm{max}=1$ and $\epsilon_\mathrm{min}=-1$. \hyperlink{fig:Pauli-SYK}{Supplementary Figure 4} shows that Poisson bracket OTOCs of the SYK model saturate to unity for all time $t>1$. 

We note that random-sign Hadamard gates and Pauli SYK Hamiltonians exhibit quadratic growth of $C_{VW}$ at the beginning. This is a consequence of the discussion in \hyperlink{sec:micro-time}{Supplementary Note 5}.

\newpage
{
\subsection*{8. Finite temperature OTOCs}
In this subsection, we discuss OTOCs of a pseudochaotic system at a finite temperature. To this end, we first recall that the correct formula $C_{VW}(U)$ for OTOCs at finite temperature is given by  
\begin{equation}
    \begin{split}
        C_{VW}(U) 
        &= - \frac{1}{2}\langle [V(t),W]^2 \rangle_\beta \\
        &= 1 - \frac{1}{2}\tr(\rho_\beta (V(t)WV(t)W + WV(t)WV(t)))\\
        &= 1 - \Re[\tr(\rho_\beta V(t)WV(t)W)], 
    \end{split}
\end{equation}
with $\rho_\beta = e^{-\beta H}/\mathcal{Z}$. In the following, we illustrate the calculation of $C_{VW}(U)$ for an RSED with \( V = Z_i \) and \( W = Z_j \), where \( i \neq j \) for illustrative purposes. 

We find it useful to expand $e^{-\beta H}$ in terms of the embedded Hamiltonian $h = \frac{1}{2\pi i} \log  u $ as
\begin{equation}
    e^{-\beta H} = \sum_{a\in\{0,1\}^{n-k}} O_a e^{-\beta h} O_a^\dag.
\end{equation}\\ 
It will also turn out to be convenient to define the thermal state without random phase factors $\rho'_\beta=e^{-\beta H'}/\mathcal{Z}$ with
\begin{equation}
    e^{-\beta H'} = \sum_{a\in\{0,1\}^{n-k}} \sum_{b,b'\in\{0,1\}^k} \ketbra{p(b,a)}{b,a}e^{-\beta h}\ketbra{b'a}{p(b'a)}.
\end{equation}
With these, we now compute the OTOC explicitly. Then, $\tr(\rho_\beta V(t)WV(t)W)$ is given by
\begin{equation}
    \begin{split}
        \tr(\rho_\beta V(t)WV(t)W)
        &=\sum_{\left\{b_i\right\}_{i=1}^9} \sum_{\left\{a_i\right\}_{i=1}^5}(-1)^{\sum_{i=1}^8 f\left(b_i , a_i\right)}[\rho_\beta]_{p(b_9,a_5),p(b_1,a_1)}[u]_{b_1, b_2}[V]_{p(b_2,a_1),p(b_3,a_2)}[u^\dag]_{b_3,b_4}  \\
        &\qquad\qquad\qquad\qquad \times[W]_{p(b_4,a_2),p(b_5,a_3)}[u]_{b_5, b_6}[V]_{p(b_6,a_3), p(b_7,a_4)}[u^\dag]_{b_7, b_8}[W]_{p(b_8,a_4),p(b_9,a_5)}\\
        &=\sum_{\left\{b_i\right\}_{i=1}^5} \sum_{a} (-1)^{[p(b_2,a)]_i+[p(b_3,a)]_j+[p(b_4,a)]_i+[p(b_5,a)]_j}\\
        &\qquad\qquad\qquad\times[\rho'_\beta]_{p(b_5,a),p(b_1,a)} [u]_{b_1, b_2}[u^\dag]_{b_2, b_3}[u]_{b_3, b_4}[u^\dag]_{b_4, b_5}.
    \end{split}
\end{equation}
By averaging over random subsets and phases, we get
\begin{equation}
    \begin{split}
        \mathbb{E}_{p, f} \tr(\rho_\beta V(t)WV(t)W) 
        & = \mathbb{E}_p \sum_{z,z'\in[N]}\sum_{\left\{b_i\right\}_{i=1}^5} \sum_a (-1)^{[p(b_1,a)]_j+[p(b_2,a)]_i+[p(b_3,a)]_j+[p(b_4,a)]_i}\\
        &\qquad\qquad\qquad\quad\times\delta_{z,p(b_1,a)}\delta_{z',p(b_5,a)}[\rho'_\beta]_{z,z'} [u]_{b_5, b_2}[u^\dag]_{b_2, b_3}[u]_{b_3, b_4}[u^\dag]_{b_4, b_1}\\
        &= \sum_{\substack{z,z'\in[N]\\a\in[N/K]}}\left(\sum_{\substack{\left\{b_i\right\}_{i=1}^5\in[K]^{1^5}_\mathrm{dist}\\\left\{z_i\right\}_{i=1}^5\in[N]^{1^5}_\mathrm{dist}}}\frac{(N-5)!}{N!}+\sum_{\substack{\left\{b_i\right\}_{i=1}^5\in[K]^{1^3,2}_\mathrm{dist}\\\left\{z_i\right\}_{i=1}^5\in[N]^{1^3,2}_\mathrm{dist}}}\frac{(N-4)!}{N!}\right.\\
        &\qquad\qquad\quad+\sum_{\substack{\left\{b_i\right\}_{i=1}^5\in[K]^{1,2^2}_\mathrm{dist}\\\left\{z_i\right\}_{i=1}^5\in[N]^{1,2^2}_\mathrm{dist}}}\frac{(N-3)!}{N!}+\sum_{\substack{\left\{b_i\right\}_{i=1}^5\in[K]^{1,4}_\mathrm{dist}\\\left\{z_i\right\}_{i=1}^5\in[N]^{1,4}_\mathrm{dist}}}\frac{(N-2)!}{N!}\\
        &\left.\qquad\qquad\quad+\sum_{\substack{\left\{b_i\right\}_{i=1}^5\in[K]^{5}_\mathrm{dist}\\\left\{z_i\right\}_{i=1}^5\in[N]^{5}_\mathrm{dist}}}\frac{1}{N}+\cdots\right)(-1)^{[z_1]_j+[z_2]_i+[z_3]_j+[z_4]_i}\\
        &\qquad\qquad\quad\times\delta_{z,z_1}\delta_{z',z_5}[\rho'_\beta]_{z,z'} [u]_{b_5, b_2}[u^\dag]_{b_2, b_3}[u]_{b_3, b_4}[u^\dag]_{b_4, b_1}.
    \end{split}
\end{equation}
Here, $[N]$ represents the set of numbers from one to $N$, and $[N]^{p}_\mathrm{dist}$ is the set of $p$-partitioned distinct numbers in $[N]$. A partition is denoted, for example, as $1^5$ which means five partitions with size one. If $p=(1^3,2)$, then it represents three partitions with size one and single partition with size two. In this expression, we note that the summation of the sign factor over $z_1,z_2,z_3,z_4\in[N]$  
\begin{align}
(-1)^{s}, \label{signfactor}
\end{align} 
with $s=[z_1]_j+[z_2]_i+[z_3]_j+[z_4]_i$ enforces $z_1=z_3$ and $z_2=z_4$. If $z_1\neq z_3$ or $z_2\neq z_4$ then phase factors different $\{z_2,z_3,z_4\}$ will cancel each others. To explicitly demonstrate this, we will set $i=1$ and $j=2$ for simplicity below. There are seven distinct possibilities: cases with assigning distinct values on different partitions for all possible partitions of $\{b_i\}_{i=1}^5$. Below, we present representative examples of these cases. \\ 

\noindent
\textbf{Example 1. Cases with two distinct bits among $\{ z_1, z_2, z_3, z_4 \}$}, i.e., $z_1=z_2$, $z_3=z_4$, and $z_1\neq z_3$: In this case, we have $s=z_1[2]+z_1[1]+z_3[2]+z_3[1]$ in Eq.\eqref{signfactor}. Let us first assume that $z_1[1]=0$ and $z_1[2]=0$. Then, the possibilities of $z_3$ having the first two bits as $\{00,01,10,11\}$ are given in \hyperlink{tab:example-1}{Supplementary Table 1}. Thus, the fraction of $s$ having the odd parity is given by $\frac{N}{2(N-1)}=\frac{1}{2}+O(N^{-1})$. For other values of the first two bits of $z_1$, the odd parity fractions are given by $\frac{1}{2}+O(N^{-1})$. This implies that the sum of $(-1)^s$ with $z_1=z_2$ and $z_3=z_4$ over $\{z_2,z_3,z_4\}$ multiplied by the overall normalization constant is $O(N^{-1})$. Hence, the terms associated with this case contribute to the OTOC at order \( O(N^{-1}) \), where \( N = 2^n \) denotes the dimension of the full Hilbert space.\\

\hypertarget{tab:example-1}{}
\begin{table}[h] 
    \begin{tabular}{cc}
        \hline\hline
        First two bits of $z_3$ & Possibilities of $\{z_3\}$ \\        
        \hline
        00 & $N/4-1$ \\
        01 & $N/4$ \\
        10 & $N/4$ \\
        11 & $N/4$ \\
        \hline\hline
    \end{tabular}
    \begin{center}
        \textbf{Supplementary Table 1}. $z_1=z_2$, $z_3=z_4$, and $z_1\neq z_3$
    \end{center}
\end{table}

\noindent
\textbf{Example 2. Cases with three distinct bits among $\{ z_1, z_2, z_3, z_4 \}$} with $z_1=z_2$, $z_1\neq z_3$, $z_1\neq z_4$, and $z_3\neq z_4$: In this case, we have $s=z_1[2]+z_1[1]+z_3[2]+z_4[1]$ in Eq.\eqref{signfactor}. Once again, for illustrative purposes, we assume that the first two bits of \( z_1 \) are \( 00 \). Then, the possibilities of $z_3$ and $z_4$ having the first two bits as $\{00,01,10,11\}$ are given in \hyperlink{tab:example-2}{Supplementary Table 2}. Thus, the fraction of $s$ having the odd parity is given by $\frac{N(N-3)}{2(N-1)(N-2)}=\frac{1}{2}+O(N^{-1})$. This implies that the terms associated with this case gives $O(N^{-1})$ contribution to the OTOC. \\ 

\hypertarget{tab:example-2}{}
\begin{table}[h]
    \begin{tabular}{ccc}
        \hline\hline
        First two bits of $z_3$ & First two bits of $z_4$ & Possibilities of $\{z_3,z_4\}$ \\        
        \hline
        00 & 00 & $(N/4-1)(N/4-2)$ \\
        01 & 00 & $(N/4)(N/4-1)$ \\
        10 & 00 & $(N/4)(N/4-1)$ \\
        11 & 00 & $(N/4)(N/4-1)$ \\
        00 & 01 & $(N/4)(N/4-1)$ \\
        00 & 10 & $(N/4)(N/4-1)$ \\
        00 & 11 & $(N/4)(N/4-1)$ \\
        01 & 01 & $(N/4)(N/4-1)$ \\
        10 & 01 & $(N/4)(N/4)$ \\
        11 & 01 & $(N/4)(N/4)$ \\
        01 & 10 & $(N/4)(N/4)$ \\
        10 & 10 & $(N/4)(N/4-1)$ \\
        11 & 10 & $(N/4)(N/4)$ \\
        01 & 11 & $(N/4)(N/4)$ \\
        10 & 11 & $(N/4)(N/4)$ \\
        11 & 11 & $(N/4)(N/4-1)$ \\
        \hline\hline
    \end{tabular}
    \begin{center}
        \textbf{Supplementary Table 2}. $z_1=z_2$, $z_1\neq z_3$, $z_1\neq z_4$, and $z_3\neq z_4$
    \end{center}
\end{table}

\noindent
\textbf{Example 3. Cases with three distinct bits among $\{z_1,z_2,z_3,z_4\}$} with $z_1=z_3$, $z_1\neq z_2$, $z_1\neq z_4$, and $z_2\neq z_4$: In this case, we have $s=z_1[2]+z_2[1]+z_1[2]+z_4[1]=z_2[1]+z_4[1]$ in Eq.\eqref{signfactor}. Let us first assume that $z_1[1]=0$ and $z_1[2]=0$ as in previous examples. Then, the possibilities of $z_3$ having the first two bits as $\{00,01,10,11\}$ are given in \hyperlink{tab:example-3}{Supplementary Table 3}. Thus, the fraction of $s$ having the odd parity is given by $\frac{N}{2(N-1)}=\frac{1}{2}+O(N^{-1})$. Again, this implies that the terms associated with this case gives $O(N^{-1})$ contribution to the OTOC.\\

\hypertarget{tab:example-3}{}
\begin{table}[h]
    \begin{tabular}{ccc}
        \hline\hline
        $z_2[1]$ & $z_4[1]$ & Possibilities of $\{z_2,z_4\}$ \\        
        \hline
        0 & 0 & $(N/2-1)(N/2-2)$ \\
        1 & 0 & $(N/2)(N/2-1)$ \\
        0 & 1 & $(N/2)(N/2-1)$ \\
        1 & 1 & $(N/2)(N/2-1)$ \\
        \hline\hline
    \end{tabular}
    \begin{center}
        \textbf{Supplementary Table 3}. $z_1=z_3$, $z_1\neq z_2$, $z_1\neq z_4$, and $z_2\neq z_4$
    \end{center}
\end{table}

The remaining cases can be analyzed in a similar manner, and each is found to contribute \( O(N^{-1}) \) to the OTOC due to phase cancellations. With this understanding, we can now write down the leading-order terms in the OTOC as follows:
\begin{equation}
    \begin{split}
        \mathbb{E}_{p, f} \tr(\rho_\beta V(t)WV(t)W) 
        &\approx \frac{(N-2)!}{N!}\sum_{\substack{z,z'\in[N]\\a\in[N/K]}}  \sum_{\substack{\{b_i\}_{i=1}^5\in[K]^5\\\{z_1,z_2\}\in[N]^{1^2}_\mathrm{dist}}} \delta_{b_1,b_3}\delta_{b_2,b_4}\delta_{z,z_1}\delta_{z',z_2}[\rho'_\beta]_{z,z'} [u]_{b_5, b_2}[u^\dag]_{b_2, b_3}[u]_{b_3, b_4}[u^\dag]_{b_4, b_1}\\
        &\quad+\frac{1}{N}\sum_{\substack{z,z'\in[N]\\a\in[N/K]}}  \sum_{\substack{\{b_i\}_{i=1}^5\in[K]^5\\z_1\in[N]}} \delta_{b_1,b_3}\delta_{b_2,b_4}\delta_{z,z_1}\delta_{z',z_1}[\rho'_\beta]_{z,z'} [u]_{b_5, b_2}[u^\dag]_{b_2, b_3}[u]_{b_3, b_4}[u^\dag]_{b_4, b_1}\\
        &= \frac{(N-2)!}{N!}\sum_{a\in[N/K]}  \sum_{\substack{\{b_i\}_{i=1}^5\in[K]^5\\\{z_1,z_2\}\in[N]^{1^2}_\mathrm{dist}}} \delta_{b_1,b_3}\delta_{b_2,b_4}[\rho'_\beta]_{z_1,z_2} [u]_{b_5, b_2}[u^\dag]_{b_2, b_3}[u]_{b_3, b_4}[u^\dag]_{b_4, b_1}\\
        &\quad+\frac{1}{N}\sum_{a\in[N/K]} \sum_{\substack{\{b_i\}_{i=1}^5\in[K]^5\\z\in[N]}} \delta_{b_1,b_3}\delta_{b_2,b_4}[\rho_\beta]_{z,z} [u]_{b_5, b_2}[u^\dag]_{b_2, b_3}[u]_{b_3, b_4}[u^\dag]_{b_4, b_1}\\
        &=\left(1+\frac{1}{N-1}\sum_{\{z_1,z_2\}\in[N]^{1^2}_\mathrm{dist}} [\rho'_\beta]_{z_1,z_2} \right)\frac{1}{K}\sum_{\{b_i\}_{i=1}^3\in[K]^3} [u]_{b_3, b_2}[u^\dag]_{b_2, b_1}[u]_{b_1, b_2}[u^\dag]_{b_2, b_1}.
    \end{split}
\end{equation}
We now insert the explicit form of the thermal density matrix \( \rho'_\beta \),
\[
    \rho'_\beta = \frac{K}{N} \cdot \frac{1}{\operatorname{Tr}(e^{-\beta h})} \sum_{a}\sum_{b,b'} \ketbra{p(b,a)}{b,a} e^{-\beta h} \ketbra{b',a}{p(b',a)},
\]
where \( K = 2^k \) is the dimension of the subsystem Hilbert space. This gives 
\begin{equation}
    \Re\mathbb{E}_{p, f} \tr(\rho_\beta V(t)WV(t)W) 
    \approx \left(1+\frac{1}{N-1}\sum_{\{b_1,b_2\}\in[K]^{1^2}_\mathrm{dist}} \left[\frac{e^{-\beta h}}{\tr(e^{-\beta h})}\right]_{b_1,b_2} \right)\frac{1}{K} \Re \sum_{\{b_i\}_{i=1}^3\in[K]^3} [u.*u.*u]_{b_1,b_2} [u^\dag]_{b_2, b_3}. \nonumber 
\end{equation}
We note that the second term in parentheses is upper bounded by \( O(K/N) \), and this bound is saturated in the limit \( \beta \rightarrow \infty \) when the ground state of \( h \) exhibits maximal coherence. Even when the upper bound is saturated, the resulting contribution to the magnitude of the OTOC remains exponentially small. \textit{This proves our claim} that the OTOCs under the pseudochaotic dynamics remain suppressed even at finite temperatures. \\

\noindent
\textbf{Explicit Demonstration.} Now, let us demonstrate our claim for the explicit example, where the embedded subsystem unitary operator is given by $u=H^{\otimes k} P$ with the Hadamard gate $H$ and the random sign gate $P$. For this case, We can explicitly find  
\begin{equation}
    \Re\mathbb{E}_{p, f} \tr(\rho_\beta V(t)WV(t)W) 
    \approx \frac{1}{K}\left(1+\frac{1}{N-1}\sum_{\{b_1,b_2\}\in[K]^{1^2}_\mathrm{dist}} \left[\frac{e^{-\beta h}}{\tr(e^{-\beta h})}\right]_{b_1,b_2} \right) = O\left(\frac{1}{K}\right),
\end{equation}
which explicitly shows the super-polynomial suppression of the OTOC since $K = 2^k$ and $k= \operatorname{polylog} n$. Finally, our numerical calculation \hyperlink{fig:HadaS-finite-temp}{Supplementary Figure 5} confirms this.\\

\noindent
\textbf{Physical Intuition.} There is a simple intuition behind why the OTOCs remain negligible at arbitrary finite temperatures. To illustrate this, we note the following facts. First, in pseudochaotic dynamics, energy eigenstates are pseudorandom. Second, at each energy level, there exists an exponential degeneracy—specifically, a \( 2^{n-k} \)-fold degeneracy. Combining these two observations, the ensemble sum over all degenerate eigenstates at a given energy yields a density matrix that is nearly indistinguishable from the infinite-temperature density matrix. On the other hand, we have shown in \hyperlink{sec:local-OTOC}{Supplementary Note 2} that the infinite-temperature density matrix results in negligible OTOCs. Consequently, any thermal states at finite temperature, which is precisely the weighted sum over the ensemble of energy eigenstates, also leads to negligible OTOCs.
}

\newpage
\subsection*{9. State design of ensemble of states generated by RSED}
Here, we provide some examples of whose ensembles of evolved states have $\epsilon$-approximate state $t$ design with negligible $\epsilon$ and polynomial $t$. Below, we first give two illustrative and simple examples of embedded Haar random unitary operators and Hadamard gates augmented by an random-sign operator. We then give more general examples later.

\begin{theorem}\label{thm:pseudorandom-ensemble-from-chaos}
    A RSED with an ensemble of embedded unitary operators that generate a pseudorandom state ensemble in the subspace generates a pseudorandom state ensemble in the entire space.
\end{theorem}
\begin{proof}
Let $u$ be an element of $\mathcal{E}_K$. Let $\ket{\psi}$ be an initial computational state, then there exist $b_*\in\{0,1\}^k$ and $a\in\{0,1\}^{n-k}$ such that $\ket{\psi}=\ket{p(b_*,a)}$. The unitary evolution $U$ embedding $u$ maps the initial state to 
\begin{equation}
    U\ket{p(b_*,a)} = O_a u O_a^\dag \ket{p(b_*,a)} = \sum_{b'_*\in\{0,1\}^k}u_{b'_*,b_*}(-1)^{f(b_*,a)+f(b'_*,a)}\ket{p(b'_*,a)}.
\end{equation}
Let the ensemble average of $t$-copies of $U\ket{p(b_*,a)}$ over $\mathcal{E}_K$ and random functions $\{f\}$ be $\sigma$. Then, the trace distance between $\mathbb{E}_{S}[\sigma]$ and Haar random states is given by
\begin{equation}\label{eq:TD-pseudo-haar}
\begin{split}
\operatorname{TD}\left(\mathbb{E}_{S}[\sigma], \mathbb{E}_{\ket{\phi} \sim \operatorname{Haar}\left(2^n\right)}\left[\ketbra{\phi}{\phi}^{\otimes t}\right]\right)
& =\operatorname{TD}\left(\mathbb{E}_{S}[\sigma], \frac{\Pi_{\text{sym}}^{N, t}}{\operatorname{tr}\left(\Pi_{\text{sym}}^{N, t}\right)}\right) \\
& \leq \operatorname{TD}\left(\mathbb{E}_{S}[\sigma], \mathbb{E}_{S} \left[\frac{\Pi_{\text{sym}}^{S, t}}{\operatorname{tr}\left(\Pi_{\text{sym}}^{S, t}\right)}\right]\right)+O\left(t^2 / 2^k\right) \\
& \leq \mathbb{E}_{S} \left[\operatorname{TD}\left(\sigma, \frac{\Pi_{\text{sym}}^{S, t}}{\operatorname{tr}\left(\Pi_{\text{sym}}^{S, t}\right)}\right)\right]+O\left(t^2 / 2^k\right) \\
& \leq \operatorname{negl}(n).
\end{split}
\end{equation}
The equality in the first line comes from the fact that the ensemble average of Haar random states is equivalent to the projector $\Pi^{N,t}_\text{sym}$ to the invariant sector of $\mathcal{U}(2^n)\times\mathcal{U}^\dag(2^n)\times S_t$, where $\mathcal{U}(2^n)$ is the group of unitary operators with the dimension $2^n$, and $S_t$ is the symmetric group of order $t$ which acts as permutations of copies. The inequality in the second line comes from \textbf{Lemma 2.2} of \cite{aaronson2023quantum}. The inequality in the third line comes from the convexity of the trace distance. Since $\sigma$ is the ensemble average of a pseudorandom state ensemble, the trace distance at the third line should be negligible. This gives the inequality in the forth line.

\end{proof}

It is also possible to generate pseudorandom states using ensemble of unitary operators far from being the Haar random ensemble due to the following theorem. 

\begin{theorem}\label{thm:computation-prob-for-randomness}
    Let $\mathcal{E}_k$ be an ensemble of unitary operators in a subsystem with the size $k$. Let us assume that for all $k$ and $u\in\mathcal{E}_k$, there exist $\epsilon>0$ such that 
    \begin{equation}\label{eq:magnitude-of-prob-for-randomness}
        \operatorname{Pr}\left[ \abs{u_{b,b'}}^2 \geq K^{-\epsilon} \right] \leq \operatorname{negl}(n)
    \end{equation}
    with $K=2^k$. In addition, let us assume that $\mathbb{E}_{u\sim\mathcal{E}_k}\left[ \abs{u_{b,b'}}^2 \right]=K^{-1}$ holds for all $b$ and $b'$. Then, an ensemble of subsystem embedded dynamics with $\mathcal{E}_k$ generates a pseudorandom state ensemble.
\end{theorem}
\begin{proof}
Let $u$ be an element of $\mathcal{E}_K$. Let $\ket{\psi}$ be an initial computational state, then there exist $b_*\in\{0,1\}^k$ and $a\in\{0,1\}^{n-k}$ such that $\ket{\psi}=\ket{p(b_*,a)}$. The unitary evolution $U$ embedding $u$ maps the initial state to 
\begin{equation}
    U\ket{p(b_*,a)} = O_a u O_a^\dag \ket{p(b_*,a)} = \sum_{b'_*\in\{0,1\}^k}u_{b'_*,b_*}(-1)^{f(b_*,a)+f(b'_*,a)}\ket{p(b'_*,a)}.
\end{equation}
The ensemble average of $U\ket{p(b_*,a)}$ over $f$ is given by 
\begin{equation}
    \begin{split}
        \sigma'_{p,u} 
        &= \mathbb{E}_{f}\left[U\ketbra{p(b_*,a)}{p(b_*,a)}U^\dagger\right] \\
        &= \sum_{\{b_i\}_{i=1}^t}\sum_{\{b'_i\}_{i=1}^t}\delta_{\text{type}(b)\text{  mod  }2,\text{type}(b')\text{  mod  }2} \prod_{i=1}^t u_{b_i,b_*}u_{b'_i,b_*}^* \ketbra{p(b_1,a),\cdots,p(b_t,a)}{p(b'_1,a),\cdots,p(b'_t,a)},
    \end{split}
\end{equation}
where the indices $\{b_i\}_{i=1}^t$ and $\{b'_i\}_{i=1}^t$ run over $\{0,1\}^k$. Let us consider sampling of random bitstrings from $\{0,1\}^k$. Then, the collision probability of $t$ samples is given by $O(t^2/K)$. This together with Eq.~\eqref{eq:magnitude-of-prob-for-randomness} approximate $\sigma'_{p,u}$ by 
\begin{equation}
    \sigma_{p,u} = \bar{X}_u^{-1}  \binom{K}{t}^{-1} \sum_{\abs{\operatorname{type}(b)}=t} \prod_{i=1}^t\abs{u_{b_i,b_*}}^2 \ketbra{\operatorname{type}(b)}{\operatorname{type}(b)}
\end{equation}
with a negligible trace distance and 
\begin{equation}
    \bar{X}_u = \binom{K}{t}^{-1}\sum_{\operatorname{type(b)}=t}\prod_{i=1}^t\abs{u_{b_i,b_*}}^2.
\end{equation}
Here, $\operatorname{type}(b)$ is the vector recording the counts of identical indices of $\{b_i\}_{t=1}^t$, and $\abs{\operatorname{type}(b)}$ is the number of non-vanishing components of $\operatorname{type}(b)$. Additionally, $\ket{\operatorname{type(b)}}$ is defined as
\begin{equation}
    \ket{\operatorname{type(b)}} = \binom{K}{t}^{-1/2} \sum_{\abs{\operatorname{type}(b)}=t} \ket{p(b_1,a),\cdots,p(b_t,a)}.
\end{equation}
Now, let us take the ensemble average over $\mathcal{E}_K$, then by the assumption of $\mathbb{E}_{u\sim\mathcal{E}_k}\left[ \abs{u_{b,b'}}^2 \right]=K^{-1}$, we get 
\begin{equation}
    \sigma_{p} = \binom{K}{t}^{-1} \sum_{\abs{\operatorname{type}(b)}=t} \ketbra{\operatorname{type}(b)}{\operatorname{type}(b)}.
\end{equation}
Importantly, this is \textbf{Hybrid 3} of \cite{aaronson2023quantum}. Let $\rho_p$ be the invariant sector projector in the subset made by $p$, \textit{i.e.},
\begin{equation}
    \rho_p = \frac{\Pi_{\text{sym}}^{S, t}}{\operatorname{tr}\left(\Pi_{\text{sym}}^{S, t}\right)},
\end{equation}
then $\operatorname{TD}(\rho_p,\sigma_p)$ is negligible in $n$ due to \textbf{Lemma 3} of \cite{aaronson2023quantum}. Finally, the triangle inequality and the convexity of the trace distance give
\begin{equation}
    \begin{split}
        \operatorname{TD}\left(\mathbb{E}_{u\sim\mathcal{E}_k}[\sigma_{p,u}'],\rho_p\right) 
        &\leq \operatorname{TD}\left(\mathbb{E}_{u\sim\mathcal{E}_k}[\sigma_{p,u}'],\sigma_p\right) + \operatorname{TD}(\sigma_p,\rho_p)\\
        &\leq \mathbb{E}_{u\sim\mathcal{E}_k}\left[\operatorname{TD}\left(\sigma_{p,u}',\sigma_{p,u}\right)\right] + \operatorname{negl}(n)\\
        &= \operatorname{negl}(n).
    \end{split}
\end{equation}
Thus, the ensemble of states $\{U\ket{p(b_*,a)}\}$ is pseudorandom.
\end{proof}

This theorem tells us that pseudorandom states can be generated by using an ensemble of unitary operators whose elements have magnitudes of $O(2^{-k/2})$. An example of such ensembles is the ensemble generated by \hyperlink{sec:random-sign-Hadamard}{the random sign Hadamard gates}, which approximately forms the circular orthogonal ensemble (COE) so that it satisfies the condition Eq.~\eqref{eq:magnitude-of-prob-for-randomness}. On the other hand, it is also possible to generate a pseudorandom state ensemble using a single unitary operator if magnitudes of its elements are concentrated near $2^{-k/2}$.

\begin{proposition}\label{thm:computation-prob-var-for-randomness}
    Let $u$ be a unitary operator in a subsystem with the size $k$. Let us assume that $u$ satisfies the condition Eq.~\eqref{eq:magnitude-of-prob-for-randomness}. Then, for all $b\in\{0,1\}^k$, if the following holds  
    \begin{equation}\label{eq:variance-of-prob-for-randomness}
        \mathbb{E}_{u\sim\mathcal{E}}\left[ \binom{K}{t}^{-1}\sum_{\abs{\operatorname{type}(b)}=t} \left( X_{u,b}/\bar{X}_u - 1\right)^2\right] \leq K^{-2\epsilon} 
    \end{equation}
    with $X_{u,b}=\prod_{i=1}^t\abs{u_{b_i,b_*}}^2$ and $\bar{X}_u = \binom{K}{t}^{-1} \sum_{\abs{\operatorname{type}(b)}=t} X_{u,b}$ for some constant $\epsilon>0$, then subsystem dynamics with $u$ generates pseodorandom states from initial computational states with a failure probability negligible in $n$. In case of having $\mathbb{E}_{u\sim\mathcal{E}_k}[X_{u,b}]=K^{-t}$ for any $b=\{b_i\}$, the sample mean $\bar{X}$ converges to $K^{-t}$. In this case, Eq.~\eqref{eq:variance-of-prob-for-randomness} can be rewritten as
    \begin{equation}
        \overline{\operatorname{Var}(X/\bar{X})} \leq K^{-2\epsilon},
    \end{equation}
    where $\overline{\operatorname{Var}(X/\bar{X})}$ is the sample mean of the variance of $X/\bar{X}$ with $\binom{K}{t}$ samples. 
\end{proposition}
\begin{proof}
Similar to the proof of \textbf{Theorem}~\ref{thm:computation-prob-for-randomness}, let us consider an initial computational state $\ket{p(b_*,a)}$ with $b_*\in\{0,1\}^k$ and $a\in\{0,1\}^{n-k}$ and its evolution under $U$, which embeds $u$, 
\begin{equation}
    U\ket{p(b_*,a)} = \sum_{b'_*\in\{0,1\}^k}u_{b'_*,b_*}(-1)^{f(b_*,a)+f(b'_*,a)}\ket{p(b'_*,a)}.
\end{equation}
Let the ensemble average of $U\ket{p(b_*,a)}$ over $f$ be $\sigma'_{p,u}$, and let $\sigma_{p,u}$ be \begin{equation}
    \sigma_{p,u} = \bar{X}^{-1}_u \binom{K}{t}^{-1} \sum_{\abs{\operatorname{type}(b)}=t} X_{u,b} \ketbra{\operatorname{type}(b)}{\operatorname{type}(b)}.
\end{equation}
Then, as discussed in the proof of \textbf{Theorem}~\ref{thm:computation-prob-for-randomness}, we have $\operatorname{TD}(\sigma'_{p,u},\sigma_{p,u})=\operatorname{negl}(n)$. Now, let us introduce $\sigma_p$ as
\begin{equation}
    \sigma_{p} = \binom{K}{t}^{-1} \sum_{\abs{\operatorname{type}(b)}=t} \ketbra{\operatorname{type}(b)}{\operatorname{type}(b)}.
\end{equation}
Then, the trace distance between $\sigma_{p,u}$ and $\sigma_p$ is given by
\begin{equation}
    \begin{split}
    \operatorname{TD}(\rho, \sigma) 
    &= \binom{K}{t}^{-1}\left\| \sum_{\abs{\operatorname{type}(b)}=t}\left(X_{u,b}/\bar{X}_u - 1\right) \ketbra{\operatorname{type}(b)}{\operatorname{type}(b)} \right\|_1\\
    &= \binom{K}{t}^{-1}\sum_{\abs{\operatorname{type}(b)}=t}\abs{X_{u,b}/\bar{X}_u - 1}.
    \end{split}
\end{equation}
Using the Cauchy-Schwarz inequality, this can be upper bounded by
\begin{equation}
 \operatorname{TD}(\rho, \sigma) \leq \left[ \binom{K}{t}^{-1}\sum_{\abs{\operatorname{type}(b)}=t} \left( X_{u,b}/\bar{X}_u - 1\right)^2\right]^{1/2}.
\end{equation}
Let us define $\bar{Y}_u$ as 
\begin{equation}
    \bar{Y}_u\equiv  \binom{K}{t}^{-1}\sum_{\abs{\operatorname{type}(b)}=t} \left( X_{u,b}/\bar{X}_u - 1\right)^2,
\end{equation}
then Markov's inequality gives
\begin{equation}
    \operatorname{Pr}\left[ \bar{Y}_u \geq K^{-2\epsilon'} \right] \leq K^{-2\epsilon'} \mathbb{E}_{u\sim\mathcal{E}_k}[\bar{Y}_u]
\end{equation}
for some constant $0<\epsilon'<\epsilon$. This gives an upper bound on the trace distance as
\begin{equation}
    \operatorname{Pr}\left[ \operatorname{TD}(\rho,\sigma) < K^{-\epsilon'} \right] > 1 - K^{\epsilon'} \sqrt{\mathbb{E}_{u\sim\mathcal{E}_k}[\bar{Y}_u]}.
\end{equation}
If $\mathbb{E}_{V\sim\mathcal{E}_k}[\bar{Y}_u]$ is smaller than $K^{2\epsilon}$, then we get
\begin{equation}
    \operatorname{Pr}\left[ \operatorname{TD}(\rho,\sigma) < K^{-\epsilon'} \right] > 1 - K^{\epsilon'-\epsilon}.
\end{equation}
Thus, $\operatorname{TD}(\rho,\sigma)$ is negligible in $n$ with a probability higher than $1-\operatorname{negl}(n)$. Finally, the triangle inequality gives
\begin{equation}
    \operatorname{TD}\left(\sigma,\frac{\Pi_{\text{sym}}^{S, t}}{\operatorname{tr}\left(\Pi_{\text{sym}}^{S, t}\right)}\right) \leq \operatorname{TD}\left(\rho,\frac{\Pi_{\text{sym}}^{S, t}}{\operatorname{tr}\left(\Pi_{\text{sym}}^{S, t}\right)}\right) + \operatorname{TD}(\rho,\sigma).
\end{equation}
The first term in the right hand side is negligible in $n$ due to \textbf{Lemma 2.3} of \cite{aaronson2023quantum}, so the ensemble of states $\{U\ket{p(b_*,a)}\}$ over random functions $f$ and permutations $p$ forms a pseudorandom state ensemble.
\end{proof}

This theorem implies that if an ensemble unitary operators in a subsystem with size $k$ generate states with coefficients of order $O(2^{-k})$ and $o(2^{-k})$ standard deviation from initial computational state, then the corresponding ensemble of subsystem embedded dynamics generate pseudorandom states. Another sufficient condition for an embedded unitary operator to generate a pseudorandom state ensemble is having a near maximal coherence. This is discussed in \textbf{Theorem}~\ref{thm:coherence-for-randomness} of \hyperlink{sec:coherence-importance}{Supplementary Note 11}.

\begin{corollary}
    Let $\mathcal{E}_k$ be an ensemble of unitary operators in a subsystem with size $k$ whose elements have 
    \begin{equation}
        \operatorname{Pr}\left[ \abs{u_{b,b'}}^2 \geq K^{-\epsilon} \right] \leq \operatorname{negl}(n)
    \end{equation}
    with $K=2^k$ for all $k<n$ and for some $\epsilon>0$. Then, if $\abs{u_{b,b'}^2}$ follows a distribution whose mean value and standard deviation are $\mu$ and $o(\mu)$, respectively, then subsystem embedded dynamics generates pseudorandom states from initial computational states.
\end{corollary}

\begin{corollary}
    Subsystem dynamics with the embedded Hadamard gate generates pseudorandom states from initial computational states.
\end{corollary}

\newpage
\subsection*{10. Level statistics of RSED}
Here, we study level statistics of the RSED. We specifically consider the spectral form factor introduced in ~\cite{Cotler2017Blackhole,Cotler2017Chaos} as well as the nearest-neighbor level repulsion distribution. 

Let us first discuss the spectral form factor. Let $k$ be the subsystem size, and let $\mathcal{E}_\mathrm{sub}$ be an ensemble of embedded hamiltonians. Then, the ensemble of hamiltonians in the entire space is given by
\begin{equation}
    \mathcal{E} = \left\{h|h = \sum_{a\in\{0,1\}^{n-k}} O_a h_\mathrm{sub} O_a^\dagger,h_\mathrm{sum}\in\mathcal{E}_\mathrm{sub}\right\}. 
\end{equation}
The energy spectrum of $h$ is that of $h_\mathrm{sub}$ with the degeneracy of $2^{n-k}$. Now, let us consider the spectral form factor defined as
\begin{equation}
    \mathcal{R}_2(\beta,t) = \mathbb{E}_{h\sim\mathcal{E}}\left[\abs{\tr\left(e^{-\beta h - iht}\right)}^2\right].
\end{equation}
Let $\{\epsilon_m\}$ be energies of $h_\mathrm{sub}$, then we have
\begin{equation}
    \mathcal{R}_2(\beta,t) = \mathbb{E}_{h\sim\mathcal{E}} \left[ 4^{n-k}\sum_{m,n=1}^K e^{-\beta(\epsilon_m+\epsilon_n) -i(\epsilon_m-\epsilon_n) t} \right] = 4^{n-k} \mathcal{R}_{2,S}(\beta,t),
\end{equation}
where $\mathcal{R}_{2,S}(\beta,t)$ is the spectral form factor of $\mathcal{E}_\mathrm{sub}$. This implies that if $\mathcal{E}_\mathrm{sub}$ is chaotic, then the spectral form factor of the RSED exhibits a chaotic behavior. Thus, one cannot distinguish a pseudochaotic subsystem embedded dynamics with embedded chaotic hamiltonians with a chaotic dynamics by looking at the spectral form factor.

Next, let us consider the nearest-neighbor level repulsion distribution. In this case, even if the level repulsion distribution of $\{\epsilon_m\}$ follows the Wigner-Dyson distribution, that for $h$ differs from it and has a peak at $\epsilon=0$ due to the exponential degeneracy {(see \hyperlink{fig:HadaP-stat}{Supplementary Figure 2})}. Whether this peak can be detected by polynomial copies of evolved states and poly-time quantum algorithms is an open question.

\newpage
\hypertarget{sec:coherence-importance}{}
\subsection*{11. Importance of coherence for pseudochaotic dynamics}
Coherence of a state regards the number of superposed states in a certain basis to represent the state. Measures for coherence should monotonically decrease under incoherent operations, and two measures satisfying this are the relative entropy of coherence and the $l_1$ norm of coherence~\cite{Baumgratz2014}. Here, we will only consider the relative entropy of coherence defined as 
\begin{equation}
    C_{\text{rel.ent.}}(\rho) = S(\rho^\text{diag}) - S(\rho),
\end{equation}
where $\rho^\text{diag}$ is the diagonal matrix whose entries are that of $\rho$, and $S(\cdot)$ is the von Neumann entropy. Since we are interesed in pure states, $S(\rho)$ will be always to be zero. The maximally coherent state in a Hilbert space with a basis is given by the equal superposition of all the basis states. Any fluctuations around the equal coefficients decrease coherence of the state. Below, we study how these small fluctuations affect pseudorandomness of an ensemble constructed by a subsystem embedded dynamics. We find that the ensemble is pseudorandom if the embedded dynamics generates states with almost saturated coherence. We then argue that a pseudorandom state ensemble generically requires non-negligible coherence for any basis.

\begin{theorem}\label{thm:coherence-for-randomness}
Let us consider a set of orthonormal states $\{\ket{\psi_i}\}_{i=1}^{2^k}$ in the subspace $S$ with the dimension $2^k$ in a $n$-qubit system. If all the $t$ copies of states with $t=\operatorname{poly}(n)$ have the maximal relative entropy of coherence in the computational basis up to $2^{-\epsilon k}$ correction for some constant $\epsilon>0$ and measurement probabilities in $(0,(1+2^{-tk/3})2^{-tk})$, then the ensemble constructed by applying random subset phase isometries $\{O_a\}_{a\in\{0,1\}^{n-k}}$ to the states is a pseudorandom state ensemble with a probability higher than $1-\operatorname{negl}(n)$.
\end{theorem}
\begin{proof}
Let the deviation of coherence of the states from the maximal value be $\Delta C_{\text{rel.ent.}}$. In addition, let $\{p_i\}_{i=1}^{D^t}$ be measurement probabilities of the computational basis measurement on $t$ copies of a state $\ket{\phi}\in\{\ket{\psi_i}_{i=1}^D\}$ with $D=2^k$. Then, $\Delta C_{\text{rel.ent.}}$ is given by
\begin{equation}
    \Delta C_{\text{rel.ent.}} = \log D^t + \sum_{i=1}^{D^t} p_i \log p_i.
\end{equation}
If all of $\{p_i\}_{i=1}^{D^t}$ are in $(0,(1+2^{-tk/3})2^{-tk})$, then each $\log p_i$ is given from the Taylor expansion by
\begin{equation}
    \log p_i = -\log D^t + D^t\left( p_i -\frac{1}{D^t} \right) + R_1(p_i)
\end{equation}
with the remainder $R_1(p_i)$, whose Cauchy form is given by
\begin{equation}
    R_1(p_i) = - \frac{1}{\xi_i^2} (p_i - \xi_i) \left(p_i - \frac{1}{D^t}\right)
\end{equation}
for some real number $\xi_i$ between $p_i$ and $D^{-t}$. By substituting $\log p_i$ into its expansion form, we get
\begin{equation}
    \Delta C_{\text{rel.ent.}} = D^t \sum_{i=1}^{D^t} p_i \left( p_i -\frac{1}{D^t} \right) - \sum_{i=1}^{D^t} \frac{p_i}{\xi_i}\left(\frac{p_i}{\xi_i}-1\right)\left( D^tp_i - 1\right).
\end{equation}
The upper bound of the second term in the right hand side can be made by setting $p_i = (1+D^{-t/3})D^{-t}$ and $\xi_i=D^{-t+\delta}$ for some $\delta>0$ satisfying $D^\delta<1+D^{-t/3}$ and is given by
\begin{equation}
    \begin{split}
    \sum_{i=1}^{D^t} \frac{p_i}{\xi_i}\left(\frac{p_i}{\xi_i}-1\right)\left( D^tp_i - 1\right)
    &< D^{2t/3} [(1+D^{-t/3})D^{-\delta}]([(1+D^{-t/3})D^{-\delta}]-1)\\
    &< D^{t/3-\delta}((1+D^{-t/3})D^{-\delta}-1)\\
    &< D^{-\delta}.
    \end{split}
\end{equation}
This introduces a lower bound on $\Delta C_{\text{rel.ent.}}$ as
\begin{equation}
    D^t \sum_{i=1}^{D^t} \left( p_i -\frac{1}{D^t} \right)^2 - D^{-\delta} < \Delta C_{\text{rel.ent.}}.
\end{equation}
Now, let us assume that $\Delta C_{\text{rel.ent.}}$ is negligibly small. Then, it follows that
\begin{equation}\label{eqn:coherence-for-randomness}
    D^t \sum_{i=1}^{D^t} \left( p_i -\frac{1}{D^t} \right)^2 < \operatorname{negl}(n).
\end{equation}
In other words, the condition Eq.~\eqref{eq:variance-of-prob-for-randomness} of \textbf{Proposition}~\ref{thm:computation-prob-var-for-randomness} is satisfied. Since the probabilities are upper bounded by $2^{-(4/3)tk}$, the condition Eq.~\eqref{eq:magnitude-of-prob-for-randomness} of \textbf{Theorem}~\ref{thm:computation-prob-for-randomness} is also satisfied. Thus, the ensemble of states embedded into the full space by random isometries $\{O_a\}_{a\in\{0,1\}^{n-k}}$ is a pseudorandom state ensemble with a probability higher than $1-\operatorname{negl}(n)$.
\end{proof}

\begin{proposition}
    Let $\ket{\psi}$ be a state has almost saturating coherence in a basis with a negligible deviation and measurement probabilities in $(0,2^{-tk+1})$. Then, it is indistinguishable from any of maximally coherent states by polynomial number $t$ of measurements in the computational basis.
\end{proposition}
\begin{proof}
The total deviation distance between the measurement probability distribution $p$ of $t$ copies of $\ket{\psi}$ and that of the same copies of a maximally coherent state is upper bounded by
\begin{equation}
    d_\text{TV}(p,q) = \frac{1}{2} \sum_{i=1}^{D^t} \abs{p_i-\frac{1}{D^t}} \leq \frac{D^t}{2} \sum_{i=1}^{D^t} \left(p_i-\frac{1}{D^t}\right)^2
\end{equation}
due to the Cauchy-Shwarz inequality. Since $\ket{\psi}$ has the maximal coherence with a negligible deviation, the upper bound is negligible due to \textbf{Theorem}~\ref{thm:coherence-for-randomness} as explicitly stated in Eq.~\eqref{eqn:coherence-for-randomness}. Since the total deviation distance is negligibly small, any strategies on $t$ copies of measurement outcomes fail to distinguish the two distributions.
\end{proof}
This proposition implies that a RSED generating states with super-polynomial number of superpositions is necessary to construct a pseudorandom state ensemble since those states are indistinguishable from the maximally coherent state in the computational basis, including typical Haar random states. Conversely, a pseudorandom state ensemble is also necessitated to have states with super-polynomial number of superpositions in any basis due to the following theorem.

\newpage
\subsection*{12. Resources for chaotic dynamics}
In this section, we study how resources including entanglement, magic, and coherence of a subsystem embedded dynamics affect its pseudochaotic nature. If these resources of the subsystem embedded dynamics are increased simultaneously, then its chaotic behavior can be enhanced so that its OTOCs and the trace distance from Haar random states decrease. To demonstrate this, we adopt the random CCX circuit construction of pseudorandom isometries in Ref.~\cite{Lee2024fast}. Then, this is formalize as the follows. 
\begin{theorem}
    OTOCs of an RSED and the trace distance between Haar random states and states evolved by the RSED scale $2^{-O(n)}$ using exponentially many entangling and non-Clifford gates, and $O(n)$ coherent gates.
\end{theorem}
\begin{proof}
    The magnitudes of OTOCs and the trace distance can be upper bounded by $2^{-O(k)}+\Delta$ with the subsystem size $k$ due to \textbf{Theorem~\ref{thm:negl-otoc-supp}-\ref{thm:computation-prob-for-randomness}}. Here, $\Delta$ is an error term due to a circuit implementation of a random isometry. Thus, they can scale with $2^{-O(n)}$ if $k=O(n)$. This requires the number of coherent gates, such as the Hadamard gate, to be $\Omega(n)$. Next, $\Delta$ can be estimated by the error term of the random CCX gate circuits. Let the number of copies be fixed as $t=\omega(\operatorname{poly}(n))$, then the error of the circuit can become $2^{-O(n)}$ by using $O(n^2 t\log t)$ CCX gates, which are non-Clifford entangling gates~\cite{Lee2024fast}. Thus, exponentially many entangling and non-Clifford gates are sufficient to make them be $2^{-O(n)}$.
\end{proof}

On the other hand, we find that those resources of a pseudochaotic subsystem embedded dynamics can be increased independently without making it become chaotic.

\begin{theorem}
    Each of entanglement, magic, and coherence of a RSED can be increased independently without changing other resources and making it become chaotic.
\end{theorem}
\begin{proof}
Entanglement, magic, and coherence can be controlled independently by attaching random Clifford gates, random T-gates, and random Hadamard gates, respectively, to the end of the circuit in Fig. 3 of the main text as shown in the subsequent sections.
\end{proof}

This theorem implies that subsystem embedded dynamics can only become chaotic by increasing all of those resources.\\

\noindent
\textbf{Controlling entanglement}
\indent Let us consider an ensemble of all possible Clifford gates $\mathcal{E}_\text{cl}=\{U\in \mathcal{U}(2^n)|UPU^\dagger\in\mathbb{P}_n,\forall P\in \mathbb{P}_n\}$, where $\mathbb{P}_n$ is the set of $n$-qubit Pauli operators. Then, its operators typically generate almost maximally entangled states from any stabilizer states~\cite{smith2006,dahlsten2007}. Now, let us compute OTOCs of Clifford gates defined in [Eq.~\ref{eq:inf-otoc}] with on-site Pauli matrices $V$ and $W$:
\begin{equation}\label{eq:Clifford-OTOC}
    \begin{split}
        O_{VW}(U) 
        &= \frac{1}{2^n} \operatorname{tr}(V U W U^\dag V U W U^\dag)\\
        &= \frac{1}{2^n} \operatorname{tr}(V W' V W')\\
        &= \pm 1,
    \end{split}
\end{equation}
with $W'=UWU^\dagger$. At the last line, we use the fact that $W'$ is a Pauli operator which either commutes or anti-commutes with another Pauli operator $V$. Thus, OTOCs of this ensemble never decay to zero. This implies that entanglement can be enhanced by an evolution independent of pseudochaoticity of the evolution. We note that the ensemble average of $O_{VW}(U)$ over Clifford gates gives vanishing OTOCs. However, since the ensemble of Clifford circuits does not form a unitary 4-design~\cite{zhu2016}, its variance of OTOCs is non-vanishing. Thus, each evolution cannot be regarded as quantum chaos. Indeed, an extensive number of non-Clifford gates is required to simulate quantum chaos~\cite{Leone2021}.

Now, let us consider applying a Clifford gate $C\in\mathcal{E}_\text{cl}$ after a pseudochaotic subspace embedded evolution $U$. The Clifford gate increases entanglement of any initial computational states~\cite{gu2023little}. However, it does not increase pseudochaoticity of the evolution based on the magnitudes of OTOCs. $C$ conjugates an on-site Pauli matrix $V$ in the first line of Eq.~\eqref{eq:Clifford-OTOC} and thus maps it to an arbitrary Pauli string of length $n$. According to the discussion in \hyperlink{sec:non-local-otoc}{Supplementary Note 4}, changing a local Pauli matrix $V$ to a non-local one does not alter the calculations of OTOCs, meaning that their magnitudes are not affected by the change. Thus, the additional Clifford gate do increase the entanglement of $U$ while preserving its pseudochaoticity.\\

\noindent
\textbf{Maximally magical ensemble without pseudochaos}\\
\indent Let us consider an ensemble $\mathcal{E}$ of unitary evolution $\mathcal{E}_\text{mg}=\{\prod_{i=1}^n T_i^{x_i}|x\in\{1,3\}^n\}$, where $T_i$ is the T-gate acting on the $i$-th qubit. Then, by evolving initial computational states by this ensemble, one gets an ensemble of states whose magic are extensive. At the same time, OTOCs with on-site Pauli matrices $V$ and $W$ placed at different positions do not vanish, since $V$ and $\tilde{W}$ are placed and localized at different sites and thus commuting with each other. This implies that whether an ensemble of unitary operators could generate extensive magic or not implies anything on pseudochaoticity of the ensemble.

Next, let us consider attaching operators in $\mathcal{E}_\text{mg}$ to a pseudochaotic subspace embedded evolution $U$. It has been shown that those additional operators increase magic of pseudorandom states~\cite{gu2023little}. On the other hand, they do not affect the magnitudes of OTOCs since they are products of on-site operators so that $V$ conjugated by them is still an on-site operator. \\

\noindent
\textbf{Maximally coherent ensemble without pseudochaos}\\
\indent Let us consider an ensemble of unitary operators $\mathcal{E}_\text{ch}=\{\prod_{i=1}^n H X^{x_i}|x\in\{0,1\}\}$. Then, initial computational states evolved by unitary operators of this ensemble has $2^n$ superpositions in the computational basis. Thus, their coherence is maximal. However, the unitary operators of the ensemble have trivial OTOCs since $V$ and $W$ conjugated by them commute each other.

It is also possible to control coherence of evolved states without changing pseudochaoticity of an evolution by applying unitary operators of $\mathcal{E}_\text{ch}$ after the evolution. These operators are on-site so that do not change magnitudes of OTOCs. However, if they applied after generating pseudorandom states, then they can increase their coherence. This can be explicitly demonstrated by the following theorem.
\begin{theorem}
    The random subset phase state has the relative entropy of coherence $k \log 2$ with respect to the computational basis, and its coherence can be enhanced to $O(n)$ with high probability.
\end{theorem}
\begin{proof}
First, since the random subset phase state is equal weight superposisions of $2^k$ states, its relative entropy of coherence is $k \log 2$. Next, let us apply the products of Hadamard gates $\prod_{i=1}^n H_i$ to $\ket{\psi_{p,f,a}}$:
\begin{equation}
    \ket{\phi}=\prod_{i=1}^n H_i \ket{\psi_{p,f,a}} = \sum_{x\in\{0,1\}^n}\left(\frac{1}{\sqrt{2}^{n+k}}\sum_{b\in\{0,1\}^k} (-1)^{f(b,a)+x\cdot p(b,a) }\right)\ket{x} = \sum_{x\in\{0,1\}^n}c_{p,f,x}\ket{x}.
\end{equation}
The ensemble average of $c_{p,f,x}^4$ over random permutations $p$ and random functions $f$ is given by
\begin{equation}
    \mathbb{E}_{f,p}\left[c_x^4\right] = \mathbb{E}_{f,p}\left[\frac{1}{2^{2(n+k)}}\sum_{b_1,b_2,b_3,b_4}(-1)^{\sum_{i=1}^4 f(b_i,a)+x\cdot\left(\sum_{i=1}^4 p(b_i,a)\right)}\right] < \frac{3}{2^{2n}}
\end{equation}
Thus, due to Markov's inequality, $c_{p,f,x}^4$ is exponentially small in $n$ with high probability,
\begin{equation}
    \operatorname{Pr}\left[\sum_{x\in\{0,1\}^n}c_{p,f,x}^4 \geq \frac{1}{2^{(1-\epsilon )n}}\right] < \frac{3}{2^{\epsilon n}}. 
\end{equation}
for some constant $\epsilon>0$. Now, let us lower bound the relative entropy of coherence of $\ket{\psi}$ using Jensen's inequality as
\begin{equation}
    \begin{split}
        C_{\text{rel.ent.}}(\ket{\phi}) 
        &= - \sum_{x\in\{0,1\}^n} c_{p,f,x}^2 \log c_{p,f,x}^2 \\
        &= - \frac{1}{2}\sum_{x\in\{0,1\}^n} c_{p,f,x}^2 \log c_{p,f,x}^4\\
        &\geq - \frac{1}{2} \log \left( \sum_{x\in\{0,1\}^n}c_{p,f,x}^4 \right).
    \end{split}
\end{equation}
Applying Markov's inequality to the last line implies that the following inequality holds with high probability:
\begin{equation}
    \frac{1-\epsilon}{2}n\log 2 \leq C_{\text{rel.ent.}}(\ket{\phi}).
\end{equation}
In other words, applying the products of Hadamard gates to a random subset phase state gives a state having coherence proportional to $n$ with high probability.
\end{proof}

\newpage
\subsection*{13. Near-term device realization of pseudochaos}
In this section, we discuss feasibility of realizing pseudochaos in near-term devices. We will specifically focus on implementation of pseudochaotic RSED. First, pseudochaotic RSED requires the implementation of pseudorandom functions and pseudorandom permutations. Pseudorandom functions can be realized using \( \operatorname{polylog}(n) \)-depth circuits based on the constructions in Refs.~\cite{NAOR1999,Banerjee2012}, which employ tree-like circuit architectures. Pseudorandom permutations can be implemented using the Luby-Rackoff construction~\cite{Hosoyamada2019}, which also relies on tree-like circuits. This is because: (1) the construction internally requires a pseudorandom function, and (2) it necessitates duplicating a bit generated by the pseudorandom function to perform parallel additions of its value across the system. Given these architectural demands, platforms that allow qubit reconfiguration—such as Rydberg atom arrays~\cite{Bluvstein2024} or systems equipped with reconfigurable routers~\cite{Wu2024}—are well suited to implement such circuits within \( \operatorname{polylog}(n) \) depth.  

Regarding gate fidelity, entangling gates in Rydberg atom systems are predicted to reach fidelities of up to 0.9995~\cite{Xue2024}. Since the total number of entangling gates required to implement pseudochaotic dynamics scales as \( \tilde{O}(n) \), where \( \tilde{O}(\cdot) \) suppresses \( \operatorname{polylog}(n) \) factors, the maximum system size \( n \) that achieves an overall fidelity of 0.9 would be approximately \( 200/a \), where \( a \) is a proportionality constant. If one uses parallel entangling gates over 60 atoms with 0.995 fidelity, as demonstrated in Ref.~\cite{Evered2023}, the maximum achievable \( n \) for 0.9 fidelity increases to approximately \( 600/a \). We estimate the value of \( a \) as follows:  
\begin{align}
a &= (\text{number of pseudorandom permutations}) \times (\text{number of Luby-Rackoff rounds}) \nonumber \\
&\quad \times \left[
1 \ (\text{pseudorandom function per round}) + 
\frac{1}{2} \ (\text{bit duplication to } \frac{n}{2} \text{ positions}) +
\frac{1}{2} \ (\text{addition to half the system}) \right] \nonumber \\
&= 2 \times 4 \times 2 = 16. \nonumber
\end{align}
Therefore, under the use of parallel entangling gates, the maximum feasible system size is approximately \( n = 600 / 16 \approx 37 \), which exceeds the number of qubits achievable for the exact simulation of quantum states by classical computers.\\

\newpage
\section*{Supplementary References}

\let\oldaddcontentsline\addcontentsline
\renewcommand{\addcontentsline}[3]{}
\let\addcontentsline\oldaddcontentsline
